\documentclass[a4paper]{article}

\newif\ifluatex
\IfFileExists{ifluatex.sty}
{
    \usepackage{ifluatex}
}
{
    \luatexfalse
}
\ifluatex
    \usepackage[TU]{fontenc}
\else
    \usepackage[utf8]{inputenc}
    \usepackage[T1]{fontenc}
\fi

\usepackage[british]{babel}

\usepackage[languages, complexity, semigroups]{Grossnathans_math}

\usepackage{hyperref}

\newcommand{\evalP}[1]{\xi_{#1}}
\newcommand{\pr}{\emph{p}}

\newcommand{\ccddo}[1]{\mathopen{[}#1\mathclose{]}}

\newcommand{\StVQJsdD}{\StVQuasi(\FSVJsdD)}
\newcommand{\FSVDAintJsdD}{\FSVGen{\FMVDA} \cap \FSVJsdD}
\newcommand{\FSVDAintLJ}{\FSVGen{\FMVDA} \cap \FSVLJ}
\newcommand{\V}{\FMVariety{V}}
\newcommand{\W}{\FMVariety{W}}

\begin{document}
\title{The Power of Programs over Monoids in $\FMVJ$ and Threshold Dot-depth One
       Languages\footnote{Revised and extended version of~\cite{Grosshans-2020}
       that includes a complete algebraic characterisation of threshold
       dot-depth one languages.}}
\author{Nathan \textsc{Grosshans}%
	\thanks{Universität Kassel, Fachbereich Elektrotechnik/Informatik,
		Kassel, Germany,
		\href{mailto:nathan.grosshans@polytechnique.edu}
		     {\nolinkurl{nathan.grosshans@polytechnique.edu}},
		\url{https://nathan.grosshans.me}.}
       }
\date{}
\maketitle

\begin{abstract}
The model of programs over (finite) monoids, introduced by Barrington and
Thérien, gives an interesting way to characterise the circuit complexity class
$\NC[1]$ and its subclasses and showcases deep connections with algebraic
automata theory.
In this article, we investigate the computational power of programs over monoids
in $\FMVJ$, a small variety of finite aperiodic monoids.
First, we give a fine hierarchy within the class of languages recognised by
programs over monoids from $\FMVJ$, based on the length of programs but also
some parametrisation of $\FMVJ$.
Second, and most importantly, we make progress in understanding what regular
languages can be recognised by programs over monoids in $\FMVJ$. To this end, we
introduce a new class of restricted dot-depth one languages, threshold dot-depth
one languages. We show that programs over monoids in $\FMVJ$ actually can
recognise all languages from this class, using a non-trivial trick, and
conjecture that threshold dot-depth one languages with additional positional
modular counting suffice to characterise the regular languages recognised by
programs over monoids in $\FMVJ$. Finally, using a result by J.\ C.\ Costa, we
give an algebraic characterisation of threshold dot-depth one languages that
supports that conjecture and is of independent interest.
 \end{abstract}

\section{Introduction}
In computational complexity theory, many hard still open questions concern
relationships between complexity classes that are expected to be quite small in
comparison to the mainstream complexity class $\PTIME$ of tractable languages.
One of the smallest such classes is $\NC[1]$, the class of languages decided by
Boolean circuits of polynomial length, logarithmic depth and bounded fan-in, a
relevant and meaningful class, that has many characterisations but whose
internal structure still mostly is a mystery. Indeed, among its most important
subclasses, we count $\AC[0]$, $\CC[0]$ and $\ACC[0]$: all of them are
conjectured to be different from each other and strictly within $\NC[1]$, but
despite many efforts for several decades, this could only be proved for the
first of those classes.

In the late eighties, Barrington and Thérien~\cite{Barrington-Therien-1988},
building on Barrington's celebrated theorem~\cite{Barrington-1989}, gave an
interesting viewpoint on those conjectures, relying on algebraic automata
theory. They defined the notion of a program over a monoid $M$: a sequence of
instructions $(i, f)$, associating through function $f$ some element of $M$ to
the letter at position $i$ in the input of fixed length. In that way, the
program outputs an element of $M$ for every input word, by multiplying out the
elements given by the instructions for that word; acceptance or rejection then
depends on that outputted element.
A language of words of arbitrary length is consequently recognised in a
non-uniform fashion, by a sequence of programs over some fixed monoid, one for
each possible input length; when that sequence is of polynomial length, it is
said that the monoid \pr-recognises that language.
Barrington and Thérien's discovery is that $\NC[1]$ and almost all of its
significant subclasses can each be exactly characterised by \pr-recognition over
monoids taken from some suitably chosen variety of finite monoids (a class of
finite monoids closed under basic operations on monoids). For instance,
$\NC[1]$, $\AC[0]$, $\CC[0]$ and $\ACC[0]$ correspond exactly to \pr-recognition
by, respectively, finite monoids, finite aperiodic monoids, finite solvable
groups and finite solvable monoids. Understanding the internal structure of
$\NC[1]$ thus becomes a matter of understanding what finite monoids from some
particular variety are able to \pr-recognise.

\vspace{0pt plus 1pt}

It soon became clear that regular languages play a central role in understanding
\pr-recognition: McKenzie, Péladeau and Thérien indeed
observed~\cite{McKenzie-Peladeau-Therien-1991} that finite monoids from a
variety $\V$ and a variety $\W$ \pr-recognise the same languages if and only if
they \pr-recognise the same regular languages. Otherwise stated, most
conjectures about the internal structure of $\NC[1]$ can be reformulated as a
statement about where one or several regular languages lie within that
structure.
This is why a line of previous works got interested into various notions of
tameness, capturing the fact that for a given variety of finite monoids,
\pr-recognition does not offer much more power than classical
morphism-recognition when it comes to regular languages 
(see~\cite{PhD_thesis/Peladeau,Peladeau-Straubing-Therien-1997,%
Maciel-Peladeau-Therien-2000,Straubing-2000,Straubing-2001,%
PhD_thesis/Tesson,Lautemann-Tesson-Therien-2006,%
Grosshans-McKenzie-Segoufin-2017}).

\vspace{0pt plus 1pt}

This paper is a contribution to an ongoing study of what regular languages can
be \pr-recognised by monoids taken from ``small'' varieties, started with the
author's Ph.D. thesis~\cite{PhD_thesis/Grosshans}. In a previous paper by the
author with McKenzie and Segoufin~\cite{Grosshans-McKenzie-Segoufin-2017}, a
novel notion of tameness was introduced and shown for the ``small'' variety of
finite aperiodic monoids $\FMVDA$. This allowed them to characterise the class
of regular languages \pr-recognised by monoids from $\FMVDA$ as those recognised
by so called quasi-$\FMVDA$ morphisms and represented a first small step towards
a new proof that the variety $\FMVA$ of finite aperiodic monoids is tame. This
is a statement equivalent to Furst's, Saxe's,
Sipser's~\cite{Furst-Saxe-Sipser-1984} and Ajtai's~\cite{Ajtai-1983} well-known
lower bound result about $\AC[0]$.
In~\cite{Grosshans-McKenzie-Segoufin-2017}, the authors also observed that,
while $\FMVDA$ ``behaves well'' with respect to \pr-recognition of regular
languages, the variety $\FMVJ$, a subclass of $\FMVDA$, does, in contrast,
``behave badly'' in the sense that monoids from $\FMVJ$ do \pr-recognise regular
languages that are not recognised by quasi-$\FMVJ$ morphisms.

\vspace{0pt plus 1pt}

Now, $\FMVJ$ is a well-studied and fundamental variety in algebraic automata
theory (see, e.g.,~\cite{Books/Pin-1986,Pin-2017}), corresponding through
classical morphism-recognition to the class of regular languages in which
membership depends on the presence or absence of a finite set of words as
subwords. This paper is a contribution to the understanding of the power of
programs over monoids in $\FMVJ$, a knowledge that certainly does not bring
us closer to a new proof of the tameness of $\FMVA$ (as we are dealing with a
strict subvariety of $\FMVDA$), but that is motivated by the importance of
$\FMVJ$ in algebraic automata theory and the unexpected power of programs over
monoids in $\FMVJ$.
The results we present in this article are threefold: first, we exhibit a fine
hierarchy within the class of languages \pr-recognised by monoids from $\FMVJ$,
depending on the length of those programs and on a parametrisation of $\FMVJ$;
second, we show that a whole class of regular languages, threshold dot-depth one
languages, are \pr-recognised by monoids from $\FMVJ$ while, in general, they
are not recognised by any quasi-$\FMVJ$ morphism; third, we give an algebraic
characterisation of the new class of threshold dot-depth one languages.
This class forms a subclass of that of dot-depth one languages~\cite{Pin-2017}
where, roughly said, detection of a given factor does work only when it does not
appear too often as a subword. We actually even conjecture that this class of
languages with additional positional modular counting (that is, letters can be
differentiated according to their position modulo some fixed number) corresponds
exactly to all regular languages \pr-recognised by monoids in $\FMVJ$. The
characterisation of threshold dot-depth one languages as being exactly the
dot-depth one languages that are also recognised by monoids in $\FMVDA$ is a
statement that is interesting in itself for automata theory and constitutes an
essential step towards the proof of the aforementioned conjecture.

\paragraph{Organisation of the paper.}
Following the present introduction, Section~\ref{sec:Preliminaries} is
dedicated to the necessary preliminaries. In Section~\ref{sec:Fine_hierarchy},
we present the results about the fine hierarchy and in
Section~\ref{sec:Regular_languages} we expose the results concerning the regular
languages \pr-recognised by monoids from $\FMVJ$.
Section~\ref{sec:Algebraic_characterisation_TDDO_languages} is dedicated to the
algebraic characterisation of threshold dot-depth one languages. Finally,
Section~\ref{sec:Conclusion} gives a short conclusion.

\paragraph{Note.}
\hspace{0pt plus 1pt}
This article is partly based on unpublished parts of the author's Ph.D.
thesis~\cite{PhD_thesis/Grosshans}.

\section{Preliminaries}
\label{sec:Preliminaries}
\subsection{Various mathematical materials}

We assume the reader is familiar with the basics of formal language theory,
semigroup theory and recognition by morphisms, that we might designate by
classical recognition; for those, we only specify some things and refer the
reader to the two classical references of the domain by
Eilenberg~\cite{Books/Eilenberg-1974,Books/Eilenberg-1976} and
Pin~\cite{Books/Pin-1986}.

\paragraph{General notations and conventions.}
Let $i, j \in \N$. We shall denote by $\intinterval{i}{j}$ the set of all
$n \in \N$ verifying $i \leq n \leq j$. We shall also denote by $[i]$ the set
$\intinterval{1}{i}$.
Given some set $E$, we shall denote by $\powerset{E}$ the powerset of $E$.
All our alphabets and words will always be finite; the empty word will be
denoted by $\emptyword$.
Given some alphabet $\Sigma$ and some $n \in \N$, we denote by
$\Sigma^{\geq n}$, $\Sigma^{=n}$ and $\Sigma^{< n}$ the set of words over
$\Sigma$ of length, respectively, at least $n$, exactly $n$ and less than $n$.
For any word $u$ over an arbitrary alphabet, we will denote by $\alphabet(u)$
the set of letters that appear in it.

\paragraph{Varieties and languages.}
A \emph{variety of monoids} is a class of finite monoids closed under
submonoids, Cartesian product and morphic images. A \emph{variety of semigroups}
is defined similarly. When dealing with varieties, we consider only finite
monoids and semigroups, each having an \emph{idempotent power}, a smallest
$\omega \in \N_{>0}$ such that $x^\omega = x^{2 \omega}$ for any element $x$.
To give an example, the variety of finite aperiodic monoids, denoted by $\FMVA$,
contains all finite monoids $M$ such that, given $\omega$ its idempotent power,
$x^\omega = x^{\omega + 1}$ for all $x \in M$.

Formally, we see a \emph{class of languages $\mathcal{C}$} as a correspondence
that associates a set of languages $\mathcal{C}(\Sigma^*)$ over $\Sigma$ to each
alphabet $\Sigma$.
To each variety $\V$ of monoids or semigroups we associate the class
$\DLang{\V}$ of languages such that, respectively, their syntactic monoid or
semigroup belongs to $\V$.
For instance, $\DLang{\FMVA}$ is well-known to be the class of star-free
languages.

\paragraph{Quasi $\V$ languages.}
If $S$ is a semigroup we denote by $S^1$ the monoid $S$ if $S$ is already a
monoid and $S \cup \set{1}$ otherwise.

The following definitions are taken from~\cite{Pin-Straubing-2005}.
Let $\varphi$ be a surjective morphism from $\Sigma^*$ to a finite monoid $M$.
For all $k$ consider the subset $\varphi(\Sigma^k)$ of $M$ (where $\Sigma^k$ is
the set of words over $\Sigma$ of length $k$). As $M$ is finite there is a $k$
such that $\varphi(\Sigma^{2k}) = \varphi(\Sigma^k)$. This implies that
$\varphi(\Sigma^k)$ is a semigroup. The semigroup given by the smallest such $k$
is called the \emph{stable semigroup of $\varphi$}. If $S$ is the stable
semigroup of $\varphi$, $S^1$ is called \emph{the stable monoid of $\varphi$}.
If $\V$ is a variety of monoids or semigroups, then we shall denote by
$\StVQuasi{\V}$ the class of such surjective morphisms whose stable monoid or
semigroup, respectively, is in $\V$ and by $\DLang{\StVQuasi{\V}}$ the class of
languages whose syntactic morphism is in $\StVQuasi{\V}$.

\paragraph{Programs over monoids.}
Programs over monoids form a non-uniform model of computation, first defined by
Barrington and Thérien~\cite{Barrington-Therien-1988}, extending Barrington's
permutation branching program model~\cite{Barrington-1989}.
Let $M$ be a finite monoid and $\Sigma$ an alphabet. A \emph{program $P$ over
$M$ on $\Sigma^n$} is a finite
sequence of instructions of the form $(i, f)$ where $i \in [n]$ and
$f \in M^\Sigma$; said otherwise, it is a word over $([n] \times M^\Sigma)$.
The \emph{length} of $P$, denoted by $\length{P}$, is the number of its
instructions.
The program $P$ defines a function from $\Sigma^n$ to $M$ as follows.
On input $w \in \Sigma^n$, each instruction $(i, f)$ outputs the monoid element
$f(w_i)$. A sequence of instructions then yields a sequence of elements of $M$
and their product is the output $P(w)$ of the program. A language
$L \subseteq \Sigma^n$ is consequently recognised by $P$ whenever there exists
$F \subseteq M$ such that $L = P^{-1}(F)$.

A language $L$ over $\Sigma$ is \emph{recognised} by a sequence of programs
$(P_n)_{n \in \N}$ over some finite monoid $M$ if for each $n$, the program
$P_n$ is on $\Sigma^n$ and recognises $L^{=n}$.
We say $(P_n)_{n \in \N}$ is of length $s(n)$ for $s\colon \N \to \N$ whenever
$\length{P_n} = s(n)$ for all $n \in \N$ and that it is of length at most $s(n)$
whenever there exists $\alpha \in \R_{>0}$ verifying
$\length{P_n} \leq \alpha \cdot s(n)$ for all $n \in \N$.

For $s\colon \N \to \N$ and $\V$ a variety of monoids, we denote by
$\Prog{\V, s(n)}$ the class of languages recognised by sequences of programs
over monoids in $\V$ of length at most $s(n)$.
The class $\Prog{\V} = \bigcup_{k \in \N} \Prog{\V, n^k}$ is then the class of
languages \pr-recognised by a monoid in $\V$, i.e. recognised by sequences of
programs over monoids in $\V$ of polynomial length.

The following is an important property of $\Prog{\V}$.

\begin{proposition}[{\cite[Corollary 3.5]{McKenzie-Peladeau-Therien-1991}}]
\label{lemma-simple-closure-P}
    Let $\FMVariety{V}$ be a variety of monoids, then $\Prog{\FMVariety{V}}$ is
    closed under Boolean operations.
\end{proposition}

Given two alphabets $\Sigma$ and $\Gamma$, a $\Gamma$-program on $\Sigma^n$ for
$n \in \N$ is defined just like a program over some finite monoid $M$ on
$\Sigma^n$, except that instructions output letters from $\Gamma$ and thus that
the program outputs words over $\Gamma$.
Let now $L \subseteq \Sigma^*$ and $K \subseteq \Gamma^*$.
We say that \emph{$L$ program-reduces to $K$} if and only if there exists a
sequence $(\Psi_n)_{n \in \N}$ of $\Gamma$-programs (the program-reduction) such
that $\Psi_n$ is on $\Sigma^n$ and $L^{=n} = \Psi_n^{-1}(K^{=\length{\Psi_n}})$
for each $n \in \N$.
The following proposition shows closure of $\Prog{\V}$ also under
program-reductions.

\begin{proposition}[{\cite[Proposition~3.3.12 and Corollary~3.4.3]
			  {PhD_thesis/Grosshans}}]
\label{ptn:P(V)_program-reduction_closure}
\label{ptn:Program-reduction_to_regular_language}
    Let $\Sigma$ and $\Gamma$ be two alphabets.
    Let $\V$ be a variety of monoids.
    Given $K \subseteq \Gamma^*$ in $\Prog{\V, s(n)}$ for $s\colon \N \to \N$
    and $L \subseteq \Sigma^*$ from which there exists a program-reduction to
    $K$ of length $t(n)$, for $t\colon \N \to \N$, we have that
    $L \in \Prog{\V, s(t(n))}$.
    In particular, when $K$ is recognised (classically) by a monoid in $\V$, we
    have that $L \in \Prog{\V, t(n)}$.
\end{proposition}

\subsection{Tameness and the variety \texorpdfstring{$\FMVJ$}{J}}

We won't introduce any of the proposed notions of tameness but will only state
that the main consequence for a variety of monoids $\V$ to be tame in the sense
of~\cite{Grosshans-McKenzie-Segoufin-2017} is that
$\Prog{\V} \cap \Reg \subseteq \DLang{\StVQuasi\V}$. This consequence has
far-reaching implications from a computational-complexity-theoretic standpoint
when $\Prog{\V}$ happens to be equal to a circuit complexity class.
For instance, tameness for $\FMVA$ implies that
$\Prog{\FMVA} \cap \Reg \subseteq \DLang{\StVQA}$, which is equivalent to the
fact that $\AC[0]$ does not contain the language $\FLang{MOD_m}$ of words over
$\set{0, 1}$ containing a number of $1$s not divisible by $m$ for any
$m\!\in\!\N, m\!\geq\!2$ (a central result in complexity
theory~\cite{Furst-Saxe-Sipser-1984,Ajtai-1983}).

Let us now define the variety of monoids $\FMVJ$.
A finite monoid $M$ of idempotent power $\omega$ belongs to $\FMVJ$ if and only
if $(xy)^\omega = (xy)^\omega x = y (xy)^\omega$ for all $x, y \in M$. It is a
strict subvariety of the variety $\FMVDA$, containing all finite monoids $M$ of
idempotent power $\omega$ such that $(xy)^\omega = (xy)^\omega x (xy)^\omega$
for all $x, y \in M$, itself a strict subvariety of $\FMVA$.
The variety $\FMVJ$ is a ``small'' one, well within $\FMVA$.

We now give some specific definitions and results about $\FMVJ$ that we will
use, based essentially on~\cite{Klima-Polak-2010}, but also
on~\cite[Chapter 4, Section 1]{Books/Pin-1986}.

For some alphabet $\Sigma$ and each $k \in \N$, let us define the equivalence
relation $\sim_k$ on $\Sigma^*$ by $u \sim_k v$ if and only if $u$ and $v$ have
the same set of $k$-subwords (subwords of length at most $k$), for all
$u, v \in \Sigma^*$. The relation $\sim_k$ is a congruence of finite index on
$\Sigma^*$.
For an alphabet $\Sigma$ and a word $u \in \Sigma^*$, we shall write
$u \shuffle \Sigma^*$ for the language of all words over $\Sigma$ having $u$ as
a subword. In the following, we consider that $\shuffle$ has precedence over
$\cup$ and $\cap$ (but of course not over concatenation).

We define the \emph{class of piecewise testable languages $\LVPT$} as the class
of regular languages such that for every alphabet $\Sigma$, the set
$\LVPT(\Sigma^*)$ contains all languages over $\Sigma$ that are Boolean
combinations of languages of the form $u \shuffle \Sigma^*$ where
$u \in \Sigma^*$.
In fact, $\LVPT(\Sigma^*)$ is the set of languages over $\Sigma$ equal to a
union of $\sim_k$-classes for some $k \in \N$ (see~\cite{Simon-1975}).
Simon showed~\cite{Simon-1975} that a language is piecewise testable if and only
if its syntactic monoid is in $\FMVJ$, i.e. $\LVPT = \DLang{\FMVJ}$.

We can define a hierarchy of piecewise testable languages in a natural way.
For $k \in \N$, let the \emph{class of $k$-piecewise testable languages
$\LVPT[k]$} be the class of regular languages such that for every alphabet
$\Sigma$, the set $\LVPT[k](\Sigma^*)$ contains all languages over $\Sigma$ that
are Boolean combinations of languages of the form $u \shuffle \Sigma^*$ where
$u \in \Sigma^*$ with $\length{u} \leq k$.
We then have that $\LVPT[k](\Sigma^*)$ is the set of languages over $\Sigma$
equal to a union of $\sim_k$-classes.
Let us define $\FMVJ[k]$ the inclusion-wise smallest variety of monoids
containing the quotients of $\Sigma^*$ by $\sim_k$ for any alphabet $\Sigma$: we
have that a language is $k$-piecewise testable if and only if its syntactic
monoid belongs to $\FMVJ[k]$, i.e. $\LVPT[k] = \DLang{\FMVJ[k]}$.
(See~\cite[Section 3]{Klima-Polak-2010}.)

\section{Fine Hierarchy}
\label{sec:Fine_hierarchy}
The first part of our investigation of the computational power of programs over
monoids in $\FMVJ$ concerns the influence of the length of programs on their
computational capabilities.

We say two programs over a same monoid on the same set of input words are
\emph{equivalent} if and only if they recognise the same languages.
Tesson and Thérien proved in~\cite{Tesson-Therien-2001} that for any monoid $M$
in $\FMVDA$, there exists some $k \in \N$ such that for any alphabet $\Sigma$
there is a constant $c \in \N_{>0}$ verifying that any program over $M$ on
$\Sigma^n$ for $n \in \N$ is equivalent to a program over $M$ on $\Sigma^n$ of
length at most $c \cdot n^k$.
Since $\FMVJ \subset \FMVDA$, any monoid in $\FMVJ$ does also have this
property. However, this does not imply that there exists some $k \in \N$ working
for all monoids in $\FMVJ$, i.e. that $\Prog{\FMVJ}$ collapses to
$\Prog{\FMVJ, n^k}$.

In this section, we show on the one hand that, as for $\FMVDA$, while
$\Prog{\FMVJ, s(n)}$ collapses to $\Prog{\FMVJ}$ for any super-polynomial
function $s\colon \N \to \N$, there does not exist any $k \in \N$ such that
$\Prog{\FMVJ}$ collapses to $\Prog{\FMVJ, n^k}$; and on the other hand that
$\Prog{\FMVJ[k]}$ does optimally collapse to
$\Prog{\FMVJ[k], n^{\ceiling{k / 2}}}$ for each $k \in \N$.

\subsection{Strict hierarchy}

Given $k, n \in \N$, we say that $\sigma$ is a \emph{$k$-selector over $n$} if
$\sigma$ is a function of $\powerset{[n]}^{[n]^k}$ that associates a subset of
$[n]$ to each vector in $[n]^k$.
For any sequence $\Delta = (\sigma_n)_{n \in \N}$ such that $\sigma_n$ is a
$k$-selector over $n$ for each $n \in \N$ --- a sequence we will call a
\emph{sequence of $k$-selectors} ---, we set
$L_\Delta = \bigcup_{n \in \N} K_{n, \sigma_n}$, where for each $n \in \N$, the
language $K_{n, \sigma_n}$ is the set of words over $\set{0, 1}$ of length
$(k + 1) \cdot n$ that can be decomposed into $k + 1$ consecutive blocks
$u^{(1)}, u^{(2)}, \ldots, u^{(k)}, v$ of $n$ letters where the first $k$ blocks
each contain $1$ exactly once and uniquely define a vector $\rho$ in $[n]^k$,
where for all $i \in [k]$, $\rho_i$ is given by the position of the only $1$ in
$u^{(i)}$ (i.e. $u^{(i)}_{\rho_i} = 1$) and $v$ is such that there exists
$j \in \sigma_n(\rho)$ verifying that $v_j$ is $1$.
Observe that for any $k$-selector $\sigma_0$ over $0$, we have
$K_{0, \sigma_0} = \emptyset$.

We now proceed similarly to what has been done in Subsection~5.1
in~\cite{Grosshans-McKenzie-Segoufin-2017} to show, on one hand, that for all
$k \in \N$, there is a monoid $M_k$ in $\FMVJ[2 k + 1]$ such that for any
sequence of $k$-selectors $\Delta$, the language $L_\Delta$ is recognised by a
sequence of programs over $M_k$ of length at most $n^{k + 1}$; and, on the other
hand, that for all $k \in \N$ there is a sequence of $k$-selectors $\Delta$ such
that for any finite monoid $M$ and any sequence of programs $(P_n)_{n \in \N}$
over $M$ of length at most $n^k$, the language $L_\Delta$ is not recognised by
$(P_n)_{n \in \N}$.

\paragraph{Upper bound.}
We start with the upper bound.
Given $k \in \N$, we define the alphabet
$Y_k = \set{e, \#} \cup \set{\bot_l, \top_l \mid l \in [k]}$; we are going to
prove that for all $k \in \N$ there exists a language
$Z_k \in \LVPT[2 k + 1](Y_k^*)$ such that for all
$\Delta = (\sigma_n)_{n \in \N}$ sequences of $k$-selectors, there exists a
program-reduction from $L_\Delta$ to $Z_k$ of length at most
$2 \cdot (k + 1)^{-k} \cdot n^{k + 1}$.
To this end, we use the following proposition and the fact that the language of
words of length $n \in \N$ of $L_\Delta$ is exactly $K_{n', \sigma_{n'}}$ when
there exists $n' \in \N$ verifying $n = (k + 1) \cdot n'$ and $\emptyset$
otherwise.

\begin{proposition}
\label{ptn:k-selectors_languages_in_P(J)}
    For all $k \in \N$ there is a language $Z_k \in \LVPT[2 k + 1](Y_k^*)$ such
    that $\emptyword \notin Z_k$ and for all $n \in \N$ and all $k$-selectors
    $\sigma_n$ over $n$, we have
    $K_{n, \sigma_n} = \Psi_{(k + 1) \cdot n, \sigma_n}^{-1}
		       (Z_k^{=\length{\Psi_{(k + 1) \cdot n, \sigma_n}}})$
    where $\Psi_{(k + 1) \cdot n, \sigma_n}$ is a $Y_k$-program on
    $\set{0, 1}^{(k + 1) \cdot n}$ of length at most
    $2 \cdot (k + 1) \cdot n^{k + 1}$.
\end{proposition}

\begin{proof}
    We first define by induction on $k$ a family of languages $Z_k$ over the
    alphabet $Y_k$.
    For $k = 0$, set $Z_0 = Y_0^* \# Y_0^*$.
    For $k \in \N_{>0}$, the language $Z_k$ is the set of words containing each
    of $\top_k$ and $\bot_k$ exactly once, the first before the latter, and
    verifying that the factor between the occurrence of $\top_k$ and the
    occurrence of $\bot_k$ belongs to $Z_{k - 1}$, i.e.
    $Z_k = Y_{k - 1}^* \top_k Z_{k - 1} \bot_k Y_{k - 1}^*$.
    A simple induction on $k$ shows that $Z_k$ for $k \in \N$ is defined by the
    expression
    \[
	Y_{k - 1}^* \top_k Y_{k - 2}^* \top_{k - 1} \cdots Y_1^* \top_2 Y_0^*
	\top_1 Y_0^* \# Y_0^* \bot_1 Y_0^* \bot_2 Y_1^* \cdots \bot_{k - 1}
	Y_{k - 2}^* \bot_k Y_{k - 1}^*
	\displaypunct{,}
    \]
    hence it belongs to $\LVPT[2 k + 1](Y_k^*)$ and in particular does not
    contain the empty word $\emptyword$.

    Fix $n \in \N$. If $n = 0$, the proposition follows trivially since for any
    $k$-selector $\sigma_0$ over $0$, we have $K_{0, \sigma_0} = \emptyset$ and
    $\emptyword \notin Z_k$; otherwise, we define by induction on $k$ a
    $Y_k$-program $P_k(d, \sigma)$ on $\set{0, 1}^{(d + k + 1) \cdot n}$ for
    every $k$-selector $\sigma$ over $n$ and every $d \in \N$.

    For any $j \in [n]$ and $\sigma$ a $0$-selector over $n$, which is just a
    function in $\powerset{[n]}^{\set{\emptyword}}$, let
    $h_{j, \sigma}\colon \set{0, 1} \to Y_0$ be the function defined by
    $h_{j, \sigma}(0) = e$ and
    $h_{j, \sigma}(1) =
     \begin{cases}
	 \# & \text{if $j \in \sigma(\emptyword)$}\\
	 e & \text{otherwise}
     \end{cases}$.
    For all $k \in \N_{>0}$, we also let $f_k$ and $g_k$ be the functions in
    ${Y_k}^{\set{0, 1}}$ defined by $f_k(0) = g_k(0) = e$, $f_k(1) = \top_k$ and
    $g_k(1) = \bot_k$. Moreover, for any $k$-selector $\sigma$ over $n$, the
    symbol $\sigma|j$ for $j \in [n]$ denotes the $(k - 1)$-selector over $n$
    such that for all $\rho' \in [n]^{k - 1}$, we have $i \in \sigma|j(\rho')$
    if and only if $i \in \sigma((j, \rho'))$.

    For $k \in \N_{>0}$, for $d \in \N$ and $\sigma$ a $k$-selector over $n$,
    the $Y_k$-program $P_k(d, \sigma)$ on $\set{0, 1}^{(d + k + 1) \cdot n}$ is
    the following sequence of instructions:
    \begin{align*}
	& (d \cdot n + 1, f_k) P_{k - 1}(d + 1, \sigma|1) (d \cdot n + 1, g_k)\\
	\cdots &
	(d \cdot n + n, f_k) P_{k - 1}(d + 1, \sigma|n) (d \cdot n + n, g_k)
	\displaypunct{.}
    \end{align*}
    In words, for each position
    $i \in \intinterval{d \cdot n + 1}{d \cdot n + n}$ with a $1$ in the
    $(d + 1)$-th block of $n$ letters in the input, the program runs, between
    the symbols $\top_k$ and $\bot_k$, the program $P_{k - 1}(d + 1, \sigma|i)$
    obtained by induction for $\sigma|i$ the $(k - 1)$-selector over $n$
    obtained by restricting $\sigma$ to all vectors in $[n]^k$ whose first
    coordinate is $i$.

    For $k = 0$, for $d \in \N$ and $\sigma$ a $0$-selector over $n$, the
    $Y_0$-program $P_0(d, \sigma)$ on $\set{0, 1}^{(d + 1) \cdot n}$ is the
    following sequence of instructions:
    \[
	(d \cdot n + 1, h_{1, \sigma}) (d \cdot n + 2, h_{2, \sigma}) \cdots
	(d \cdot n + n, h_{n, \sigma})
	\displaypunct{.}
    \]
    In words, for each position
    $i \in \intinterval{d \cdot n + 1}{d \cdot n + n}$ with a $1$ in the
    $(d + 1)$-th block of $n$ letters in the input, the program outputs $\#$ if
    and only if $(i - d \cdot n)$ does belong to the set $\sigma(\emptyword)$.

    In short, $P_k(d, \sigma)$ is designed so that for any
    $w \in \set{0, 1}^{(d + k + 1) \cdot n}$, the word $P_k(d, \sigma)(w)$
    belongs to $Z_k$ if and only if the last $(k + 1) \cdot n$ letters of $w$
    form a word of $K_{n, \sigma}$.

    A simple computation shows that for any $k \in \N$, any $d \in \N$ and
    $\sigma$ a $k$-selector over $n$, the number of instructions in
    $P_k(d, \sigma)$ is at most $2 \cdot (k + 1) \cdot n^{k + 1}$.

    A simple induction on $k$ shows that for any $k \in \N$ and $d \in \N$, when
    running on a word $w \in \set{0, 1}^{(d + k + 1) \cdot n}$, for any $\sigma$
    a $k$-selector over $n$, the program $P_k(d, \sigma)$ returns a word in
    $Z_k$ if and only if when $u^{(1)}, u^{(2)}, \ldots, u^{(k)}, v$ are the
    last $k + 1$ consecutive blocks of $n$ letters of $w$, then
    $u^{(1)}, u^{(2)}, \ldots, u^{(k)}$ each contain $1$ exactly once and define
    the vector $\rho$ in $[n]^k$ where for all $i \in [k]$, the value $\rho_i$
    is given by the position of the only $1$ in $u^{(i)}$, verifying that there
    exists $j \in \sigma_n(\rho)$ such that $v_j$ is $1$.

    Therefore, for any $k \in \N$ and $\sigma_n$ a $k$-selector over $n$, if we
    set $\Psi_{(k + 1) \cdot n, \sigma_n} = P_k(0, \sigma_n)$, we have
    $K_{n, \sigma_n} = \Psi_{(k + 1) \cdot n, \sigma_n}^{-1}
		       (Z_k^{=\length{\Psi_{(k + 1) \cdot n, \sigma_n}}})$
    where $\Psi_{(k + 1) \cdot n, \sigma_n}$ is a $Y_k$-program on
    $\set{0, 1}^{(k + 1) \cdot n}$ of length at most
    $2 \cdot (k + 1) \cdot n^{k + 1}$.
\end{proof}

Consequently, for all $k \in \N$ and any sequence of $k$-selectors $\Delta$,
since the language $Z_k$ is in $\LVPT[2 k + 1](Y_k^*)$ and thus recognised by a
monoid from $\FMVJ[2 k + 1]$, we have, by
Proposition~\ref{ptn:Program-reduction_to_regular_language}, that
$L_\Delta \in \Prog{\FMVJ[2 k + 1], n^{k + 1}}$.

\paragraph{Lower bound.}
For the lower bound, we use the following claim, whose proof can be found
in~\cite[Claim 10]{Grosshans-McKenzie-Segoufin-2017}.

\begin{claim}\label{claim-monoid-fixed}
    For all $i \in \N_{>0}$ and $n \in \N$, the number of languages in
    $\set{0, 1}^n$ recognised by programs over a monoid of order $i$ on
    $\set{0, 1}^n$, with at most $l \in \N$ instructions, is upper-bounded by
    $i^{i^2} 2^i \cdot (n \cdot i^2)^l$.
\end{claim}

If for some $k \in \N$ and $i \in [\alpha]$ with $\alpha \in \N_{>0}$, we apply
this claim for all $n \in \N$ and $l = \alpha \cdot ((k + 1) \cdot n)^k$, we get
a number $\mu_i(n)$ of languages in $\set{0, 1}^{(k + 1) \cdot n}$ recognised by
programs over a monoid of order $i$ on $\set{0, 1}^{(k + 1) \cdot n}$ with at
most $l$ instructions that is in $2^{\Omicron(n^k \log_2(n))}$, which is
asymptotically strictly smaller than the number of distinct $K_{n, \sigma_n}$
when the $k$-selector $\sigma_n$ over $n$ varies, which is $2^{n^{k + 1}}$, i.e.
$\mu_i(n)$ is in $\omicron(2^{n^{k + 1}})$.

Hence, for all $j \in \N_{>0}$, there exist an $n_j \in \N$ and $\tau_j$ a
$k$-selector over $n_j$ such that no program over a monoid of order $i \in [j]$
on $\set{0, 1}^{(k + 1) \cdot n_j}$ and of length at most
$j \cdot ((k + 1) \cdot n_j)^k$ recognises $K_{n_j, \tau_j}$. Moreover, we can
assume without loss of generality that the sequence $(n_j)_{j \in \N_{>0}}$ is
increasing.
Let $\Delta = (\sigma_n)_{n \in \N}$ be such that $\sigma_{n_j} = \tau_j$ for
all $j \in \N_{>0}$ and
$\sigma_n\colon [n]^k \to \powerset{[n]}, \rho \mapsto \emptyset$ for any
$n \in \N$ verifying that it is not equal to any $n_j$ for $j \in \N_{>0}$. We
show that no sequence of programs over a finite monoid of length $\Omicron(n^k)$
can recognise $L_\Delta$. If this were the case, then let $i$ be the order of
the monoid. Let $j \in \N, j \geq i$ be such that for any $n \in \N$, the $n$-th
program has length at most $j \cdot n^k$. But, by construction, we know that
there does not exist any such program on $\set{0, 1}^{(k + 1) \cdot n_j}$
recognising $K_{n_j, \tau_j}$, a contradiction.

This implies the following hierarchy, using the fact that for all $k \in \N$ and
all $d \in \N, d \leq \ceiling{\frac{k}{2}} - 1$, any monoid from $\FMVJ[d]$ is
also a monoid from $\FMVJ[k]$.

\begin{proposition}
\label{ptn:P(J)_strict_hierarchy}
    For all $k \in \N$, we have
    $\Prog{\FMVJ, n^k} \subset \Prog{\FMVJ, n^{k + 1}}$.
    More precisely, for all $k \in \N$ and
    $d \in \N, d \leq \ceiling{\frac{k}{2}} - 1$, we have
    $\Prog{\FMVJ[k], n^d} \subset \Prog{\FMVJ[k], n^{d + 1}}$.
\end{proposition}

\subsection{Collapse}

Looking at Proposition~\ref{ptn:P(J)_strict_hierarchy}, it looks at first glance
rather strange that, 
for each $k \in \N$, we can only prove strictness of the hierarchy inside
$\Prog{\FMVJ[k]}$ up to exponent $\ceiling{\frac{k}{2}}$. We now show, in a way
similar to Subsection~5.2 in~\cite{Grosshans-McKenzie-Segoufin-2017}, that in
fact $\Prog{\FMVJ[k]}$ does collapse to $\Prog{\FMVJ[k], n^{\ceiling{k / 2}}}$
for all $k \in \N$, showing Proposition~\ref{ptn:P(J)_strict_hierarchy} to be
optimal in some sense.

\begin{proposition}
\label{ptn:P(J)_collapse}
    Let $k \in \N$.
    Let $M \in \FMVJ[k]$ and $\Sigma$ be an alphabet.
    Then there exists a constant $c \in \N_{>0}$ such that any program over $M$
    on $\Sigma^n$ for $n \in \N$ is equivalent to a program over $M$ on
    $\Sigma^n$ of length at most $c \cdot n^{\ceiling{k / 2}}$.

    In particular, $\Prog{\FMVJ[k]} = \Prog{\FMVJ[k], n^{\ceiling{k / 2}}}$ for
    all $k \in \N$.
\end{proposition}

Actually, the equivalent shorter program we give is even a \emph{subprogram} of
the original one, i.e. a subsequence of the latter. For $P$ some program over a
finite monoid $M$, we may denote by $\evalP{P}$ the function that associates to
each possible input word $w$ the word in $M^{\length{P}}$ obtained by
successively evaluating the instructions of $P$ for $w$.

Observe that given $P$ a program over some finite monoid $M$ on $\Sigma^n$ for
$n \in \N$ and $\Sigma$ an alphabet, a subprogram $P'$ of $P$ is equivalent to
$P$ if and only if for every language $K \subseteq M^*$ recognised by the
evaluation morphism $\eta_M$ of $M$, the unique morphism from $M^*$ to $M$
extending the identity on $M$, we have
$\evalP{P}(w) \in K \Leftrightarrow \evalP{P'}(w) \in K$ for all
$w \in \Sigma^n$. Moreover, every language recognised by $\eta_M$ is precisely a
language of $\LVPT[k](M^*)$ when $M \in \FMVJ[k]$ for some $k \in \N$.

The result is hence a consequence of the following lemma and the fact that every
language in $\LVPT[k](M^*)$ is a union of $\sim_k$-classes, each of those
classes corresponding to all words over $M$ having the same set of $k$-subwords,
that is finite.

\begin{lemma}
\label{lem:Program_compression_subword_presence}
    Let $\Sigma$ be an alphabet and $M$ a finite monoid.

    For all $k \in \N$, there exists a constant $c \in \N_{>0}$ verifying that
    for any program $P$ over $M$ on $\Sigma^n$ for $n \in \N$ and any word
    $t \in M^k$, there exists a subprogram $Q$ of $P$ of length at most
    $c \cdot n^{\ceiling{k / 2}}$ such that for any subprogram $Q'$ of $P$ that
    has $Q$ as a subprogram, we have that $t$ is a subword of $\evalP{P}(w)$ if
    and only if $t$ is a subword of $\evalP{Q'}(w)$ for all $w \in \Sigma^n$.
\end{lemma}

\begin{proof}
    A program $P$ over $M$ on $\Sigma^n$ for $n \in \N$ is a finite sequence
    $(p_i, f_i)$ of instructions where each $p_i$ is a positive natural number
    which is at most $n$ and each $f_i$ is a function from $\Sigma$ to $M$. We
    denote by $l$ the number of instructions of $P$. For each set
    $I \subseteq [l]$ we denote by $P[I]$ the subprogram of $P$ consisting of
    the subsequence of instructions of $P$ obtained after removing all
    instructions whose index is not in $I$. When $I = \intinterval{i}{j}$ for
    some $i, j \in [l]$, we may write $P[i, j]$ instead of $P[I]$.

    We prove the lemma by induction on $k$, fixing the constant to be
    $c_k = k! \cdot \card{\Sigma}^{\ceiling{k / 2}}$ for a given $k \in \N$.

    The intuition behind the proof for a program $P$ on inputs of length $n$ and
    some $t$ of length at least $3$ is as follows. Given $l$ the length of $P$,
    we will select a subset $I$ of the indices of instructions numbered from $1$
    to $l$ to obtain $P[I]$ verifying the conditions of the lemma.
    Consider all the indices $1 \leq i_1 < i_2 < \cdots < i_s \leq l$ that each
    correspond, for some letter $a$ and some position $p$ in the input, to the
    first instruction of $P$ that would output the element $t_1$ when reading
    $a$ at position $p$ or to the last instruction of $P$ that would output the
    element $t_k$ when reading $a$ at position $p$.
    We then have that, given some $w$ as input, $t$ is a subword of
    $\evalP{P}(w)$ if and only if there exist $1 \leq \gamma < \delta \leq s$
    verifying that the element at position $i_\gamma$ of $\evalP{P}(w)$ is
    $t_1$, the element at position $i_\delta$ of $\evalP{P}(w)$ is $t_k$ and
    $t_2 \cdots t_{k - 1}$ is a subword of
    $\evalP{P[i_\gamma + 1, i_\delta - 1]}(w)$.
    The idea is then that if we set $I$ to contain $i_1, i_2, \ldots, i_s$ as
    well as all indices obtained by induction for $P[i_j + 1, i_{j + 1} - 1]$
    and $t_\alpha \cdots t_\beta$ for all $1 \leq j \leq s - 1$ and
    $1 < \alpha \leq \beta < k$, we would have that for all $w$, the word $t$ is
    a subword of $\evalP{P}(w)$ if and only if it is a subword of
    $\evalP{P[I]}(w)$, that is $\evalP{P}(w)$ where only the elements at indices
    in $I$ have been kept. The length upper bound of the order of
    $n^{\ceiling{k / 2}}$ would be met because the number of possible values for
    $j$ is $s - 1$, hence at most linear in $n$, and the number of possible
    values for $(\alpha, \beta)$ is quadratic in $k$, a constant.

    The intuition behind the proof when $t$ is of length less than $3$ is
    essentially the same, but without induction.

    \paragraph{Inductive step.}
    Let $k \in \N, k \geq 3$ and assume the lemma proven for all
    $k' \in \N, k' < k$. Let $P$ be a program over $M$ on $\Sigma^n$ for
    $n \in \N$ of length $l \in \N$ and some word $t \in M^k$.

    Observe that when $n = 0$, we necessarily have $P = \emptyword$, so that the
    lemma is trivially proven in that case. So we now assume $n > 0$.

    For each $p \in [n]$ and each $a \in \Sigma$ consider within the sequence of
    instructions of $P$ the first instruction of the form $(p, f)$ with
    $f(a) = t_1$ and the last instruction of that form with $f(a) = t_k$, if
    they exist. We let $I_{(1, k)}$ be the set of indices of these instructions
    for all $a$ and $p$. Notice that the size of $I_{(1, k)}$ is at most
    $2 \cdot \card{\Sigma} \cdot n$.

    Let $s = \card{I_{(1, k)}}$ and let us denote
    $I_{(1, k)} = \set{i_1, i_2, \ldots, i_s}$ where
    $i_1 < i_2 < \cdots < i_s$. Given $\alpha, \beta \in [k]$, we also set
    $t^{(\alpha, \beta)} = t_\alpha t_{\alpha + 1} \cdots t_\beta$.
    For all $\alpha, \beta \in [k]$ such that $1 < \alpha \leq \beta < k$ and
    $j \in [s - 1]$, we let $J_{j, (\alpha, \beta)}$ be the set of indices of
    the instructions within $P[i_j + 1, i_{j + 1} - 1]$ appearing in its
    subprogram obtained by induction for $P[i_j + 1, i_{j + 1} - 1]$ and
    $t^{(\alpha, \beta)}$.

    We now let $I$ be the union of $I_{(1, k)}$ and
    $J_{j, (\alpha, \beta)}' = \set{e + i_j \mid e \in J_{j, (\alpha, \beta)}}$
    for all $\alpha, \beta \in [k]$ such that $1 < \alpha \leq \beta < k$ and
    $j \in [s - 1]$ (the translation being required because the first
    instruction in $P[i_j + 1, i_{j + 1} - 1]$ is the $(i_j + 1)$-th instruction
    in $P$).
    We claim that $Q = P[I]$, a subprogram of $P$, has the desired properties.

    First notice that by induction the size of $J_{j, (\alpha, \beta)}'$ for all
    $\alpha, \beta \in [k]$ such that $1 < \alpha \leq \beta < k$ and
    $j \in [s - 1]$ is upper bounded by
    \[
	(\beta - \alpha + 1)! \cdot
	\card{\Sigma}^{\ceiling{(\beta - \alpha + 1) / 2}} \cdot
	n^{\ceiling{(\beta - \alpha + 1) / 2}} \leq
	(k - 2)! \cdot \card{\Sigma}^{\ceiling{(k - 2) / 2}} \cdot
	n^{\ceiling{(k - 2) / 2}}
	\displaypunct{.}
    \]
    Hence, the size of $I$ is at most
    \begin{align*}
	& \card{I_{(1, k)}} +
	  \sum_{j = 1}^{s - 1} \sum_{1 < \alpha \leq \beta < k}
	\card{J_{j, (\alpha, \beta)}'}\\
	\leq & 2 \cdot \card{\Sigma} \cdot n +
	       (2 \cdot \card{\Sigma} \cdot n - 1) \cdot
	       \frac{(k - 1) \cdot (k - 2)}{2} \cdot (k - 2)! \cdot
	       \card{\Sigma}^{\ceiling{(k - 2) / 2}} \cdot
	       n^{\ceiling{(k - 2) / 2}}\!\!\\
	\leq & 2 \cdot \card{\Sigma} \cdot n +
	       (2 \cdot \card{\Sigma} \cdot n - 1) \cdot
	       \frac{k \cdot (k - 1)}{2} \cdot (k - 2)! \cdot
	       \card{\Sigma}^{\ceiling{(k - 2) / 2}} \cdot
	       n^{\ceiling{(k - 2) / 2}}\\
	\leq & k! \cdot \card{\Sigma}^{\ceiling{k / 2}} \cdot
	       n^{\ceiling{k / 2}} = c_k \cdot n^{\ceiling{k / 2}}
    \end{align*}
    as
    $\card{\set{(\alpha, \beta) \in \N^2 \mid 1 < \alpha \leq \beta < k}} =
     \sum_{j = 2}^{k - 1} (k - j) = \sum_{j = 1}^{k - 2} j =
     \frac{(k - 1) \cdot (k - 2)}{2}$
    and
    $2 \cdot \card{\Sigma} \cdot n \leq
     \frac{k!}{2} \cdot \card{\Sigma}^{\ceiling{(k - 2) / 2}} \cdot
     n^{\ceiling{(k - 2) / 2}}$
    since $k \geq 3$, so that $P[I]$ has at most the required length.

    Let $Q'$ be a subprogram of $P$ that has $Q$ as a subprogram: it means that
    there exists some set $I' \subseteq [l]$ containing $I$ such that
    $Q' = P[I']$.

    Take $w \in \Sigma^n$.

    Assume now that $t$ is a subword of $\evalP{P}(w)$.
    It means that there exist $r_1, r_2, \ldots, \allowbreak r_k \in [l]$,
    $r_1 < r_2 < \cdots < r_k$, such that for all $j \in [k]$, we have
    $f_{r_j}(w_{p_{r_j}}) = t_j$. By definition of $I_{(1, k)}$, there exist
    $\gamma, \delta \in [s], \gamma < \delta$, such that
    $i_\gamma \leq r_1 < r_k \leq i_\delta$ and
    $f_{i_\gamma}(w_{p_{i_\gamma}}) = t_1$ and
    $f_{i_\delta}(w_{p_{i_\delta}}) = t_k$. For each
    $j \in \intinterval{\gamma}{\delta - 1}$, let $m_j \in \intinterval{2}{k}$
    be the smallest integer in $\intinterval{2}{k - 1}$ such that
    $i_j \leq r_{m_j} < i_{j + 1}$ and $k$ if it does not exist, and
    $M_j \in \intinterval{1}{k - 1}$ be the biggest integer
    in $\intinterval{2}{k - 1}$ such that $i_j \leq r_{M_j} < i_{j + 1}$ and $1$
    if it does not exist. Observe that, since for each
    $j \in \intinterval{\gamma}{\delta - 1}$, we have
    $t^{(m_j, M_j)} = t^{(k, 1)} = \emptyword$ if there does not exist any
    $o \in \intinterval{2}{k - 1}$ verifying $i_j \leq r_o < i_{j + 1}$, it
    holds that
    $t^{(2, k - 1)} = \prod_{j = \gamma}^{\delta - 1} t^{(m_j, M_j)}$.
    For all $j \in \intinterval{\gamma}{\delta - 1}$, we have that for any set
    $J \subseteq [i_{j + 1} - i_j - 1]$ containing
    $\bigcup_{1 < \alpha \leq \beta < k} J_{j, (\alpha, \beta)}$, the word
    $t^{(m_j, M_j)}$ is a subword of
    $f_{i_j}(w_{p_{i_j}}) \evalP{P[i_j + 1, i_{j + 1} - 1][J]}(w)$ when
    $m_j < k$ and $r_{m_j} = i_j$, and of
    $\evalP{P[i_j + 1, i_{j + 1} - 1][J]}(w)$ otherwise.
    Indeed, let $j \in \intinterval{\gamma}{\delta - 1}$.
    \begin{itemize}
	\item
	    If $m_j < k$ and $r_{m_j} = i_j$, then
	    $f_{i_j}(w_{p_{i_j}}) = f_{r_{m_j}}(w_{p_{r_{m_j}}}) = t_{m_j}$ and
	    $i_j = r_{m_j} < r_{m_j + 1} < \cdots < r_{M_j} < i_{j + 1}$, so
	    $t^{(m_j + 1, M_j)}$ is a subword of
	    $\evalP{P[i_j + 1, i_{j + 1} - 1]}(w)$. This implies, directly when
	    $m_j = M_j$ or by induction otherwise, that for any set
	    $J \subseteq [i_{j + 1} - i_j - 1]$ containing
	    $\bigcup_{1 < \alpha \leq \beta < k} J_{j, (\alpha, \beta)}$, the
	    word $t^{(m_j + 1, M_j)}$ is a subword of
	    $\evalP{P[i_j + 1, i_{j + 1} - 1][J]}(w)$. This implies in turn that
	    $t^{(m_j, M_j)}$ is a subword of
	    $f_{i_j}(w_{p_{i_j}}) \evalP{P[i_j + 1, i_{j + 1} - 1][J]}(w)$.
	\item
	    Otherwise, when $m_j = k$, there does not exist any
	    $o \in \intinterval{2}{k - 1}$ verifying $i_j \leq r_o < i_{j + 1}$,
	    so $t^{(m_j, M_j)} = \emptyword$ is trivially a subword of
	    $\evalP{P[i_j + 1, i_{j + 1} - 1][J]}(w)$ for any set
	    $J \subseteq [i_{j + 1} - i_j - 1]$ containing
	    $\bigcup_{1 < \alpha \leq \beta < k} J_{j, (\alpha, \beta)}$. And
	    when $m_j < k$ but $r_{m_j} \neq i_j$, it means that
	    $r_{m_j} > i_j$, hence
	    $i_j < r_{m_j} < r_{m_j + 1} < \cdots < r_{M_j} < i_{j + 1}$, so
	    $t^{(m_j, M_j)}$ is a subword of
	    $\evalP{P[i_j + 1, i_{j + 1} - 1]}(w)$. This implies, by induction,
	    that $t^{(m_j, M_j)}$ is a subword of
	    $\evalP{P[i_j + 1, i_{j + 1} - 1][J]}(w)$ for any set
	    $J \subseteq [i_{j + 1} - i_j - 1]$ containing
	    $\bigcup_{1 < \alpha \leq \beta < k} J_{j, (\alpha, \beta)}$.
    \end{itemize}
    Therefore, using the convention that $i_0 = 0$ and $i_{s + 1} = l + 1$, if
    we define, for each $j \in \intinterval{0}{s}$, the set
    $I_j' = \set{e - i_j \mid e \in I', i_j < e < i_{j + 1}}$ as the subset of
    $I'$ of elements strictly between $i_j$ and $i_{j + 1}$ translated by
    $-i_j$, we have that $t^{(2, k - 1)}$ is a subword of
    \begin{align*}
	\evalP{P[i_\gamma + 1, i_{\gamma + 1} - 1][I_\gamma']}(w)
	& f_{i_{\gamma + 1}}(w_{p_{i_{\gamma + 1}}})
	  \evalP{P[i_{\gamma + 1} + 1, i_{\gamma + 2} - 1][I_{\gamma + 1}']}(w)
	  \cdots\\
	& f_{i_{\delta - 1}}(w_{p_{i_{\delta - 1}}})
	  \evalP{P[i_{\delta - 1} + 1, i_\delta - 1][I_{\delta - 1}']}(w)
    \end{align*}
    (since we have $r_{m_\gamma} \geq r_2 > r_1 \geq i_\gamma$), so that, as
    $f_{i_\gamma}(w_{p_{i_\gamma}}) = t_1$ and
    $f_{i_\delta}(w_{p_{i_\delta}}) = t_k$, we have that
    $t = t_1 t^{(2, k - 1)} t_k$ is a subword of
    \begin{align*}
	& \evalP{P[1, i_1 - 1][I_0']}(w)
	f_{i_1}(w_{p_{i_1}}) \evalP{P[i_1 + 1, i_2 - 1][I_1']}(w)
	\cdots
	f_{i_s}(w_{p_{i_s}}) \evalP{P[i_s + 1, l][I_s']}(w)\\
	= & \evalP{P[I']}(w)
	\displaypunct{.}
    \end{align*}

    Assume finally that $t$ is a subword of $\evalP{P[I']}(w)$. Then it is
    obviously a subword of $\evalP{P}(w)$, as $\evalP{P[I']}(w)$ is a subword of
    $\evalP{P}(w)$.

    Therefore, $t$ is a subword of $\evalP{P}(w)$ if and only if $t$ is a
    subword of $\evalP{Q'}(w) = \evalP{P[I']}(w)$, as desired.

    \paragraph{Base case.}
    There are three subcases to consider.

    \emph{Subcase $k = 2$.}
    Let $P$ be a program over $M$ on $\Sigma^n$ for $n \in \N$ of length
    $l \in \N$ and some word $t \in M^2$.

    We use the same idea as in the inductive step.

    Observe that when $n = 0$, we necessarily have $P = \emptyword$, so that the
    lemma is trivially proven in that case. So we now assume $n > 0$.

    For each $p \in [n]$ and each $a \in \Sigma$ consider within the sequence of
    instructions of $P$ the first instruction of the form $(p, f)$ with
    $f(a) = t_1$ and the last instruction of that form with $f(a) = t_2$, if
    they exist. We let $I$ be the set of indices of these instructions for all
    $a$ and $p$. Notice that the size of $I$ is at most
    $2 \cdot \card{\Sigma} \cdot n =
     2! \cdot \card{\Sigma}^{\ceiling{2 / 2}} \cdot n^{\ceiling{2 / 2}} =
     c_2 \cdot n^{\ceiling{2 / 2}}$.

    We claim that $Q = P[I]$, a subprogram of $P$, has the desired properties.
    We just showed it has at most the required length.
    
    Let $Q'$ be a subprogram of $P$ that has $Q$ as a subprogram: it means that
    there exists some set $I' \subseteq [l]$ containing $I$ such that
    $Q' = P[I']$.
    
    Take $w \in \Sigma^n$.

    Assume now that $t$ is a subword of $\evalP{P}(w)$.
    It means there exist $i_1, i_2 \in [l], i_1 < i_2$ such that
    $f_{i_1}(w_{p_{i_1}}) = t_1$ and $f_{i_2}(w_{p_{i_2}}) = t_2$. By definition
    of $I$, there exist ${i_1}', {i_2}' \in I$, such that
    ${i_1}' \leq i_1 < i_2 \leq {i_2}'$ and $f_{{i_1}'}(w_{p_{{i_1}'}})= t_1$
    and $f_{{i_2}'}(w_{p_{{i_2}'}}) = t_2$.
    Hence, as $f_{{i_1}'}(w_{p_{{i_1}'}}) f_{{i_2}'}(w_{p_{{i_2}'}})$ is a
    subword of $\evalP{P[I']}(w)$ (because $I \subseteq I'$), we get that
    $t = t_1 t_2$ is a subword of $\evalP{P[I']}(w)$.

    Assume finally that $t$ is a subword of $\evalP{P[I']}(w)$. Then it is
    obviously a subword of $\evalP{P}(w)$, as $\evalP{P[I']}(w)$ is a subword of
    $\evalP{P}(w)$.

    Therefore, $t$ is a subword of $\evalP{P}(w)$ if and only if $t$ is a
    subword of $\evalP{Q'}(w) = \evalP{P[I']}(w)$, as desired.

    \emph{Subcase $k = 1$.}
    Let $P$ be a program over $M$ on $\Sigma^n$ for $n \in \N$ of length
    $l \in \N$ and some word $t \in M^1$.

    We again use the same idea as before.

    Observe that when $n = 0$, we necessarily have $P = \emptyword$, so that the
    lemma is trivially proven in that case. So we now assume $n > 0$.

    For each $p \in [n]$ and each $a \in \Sigma$ consider within the sequence of
    instructions of $P$ the first instruction of the form $(p, f)$ with
    $f(a) = t_1$, if it exists. We let $I$ be the set of indices of these
    instructions for all $a$ and $p$. Notice that the size of $I$ is at most
    $\card{\Sigma} \cdot n =
     1! \cdot \card{\Sigma}^{\ceiling{1 / 2}} \cdot n^{\ceiling{1 / 2}} =
     c_1 \cdot n^{\ceiling{1 / 2}}$.

    We claim that $Q = P[I]$, a subprogram of $P$, has the desired properties.
    We just showed it has at most the required length.
    
    Let $Q'$ be a subprogram of $P$ that has $Q$ as a subprogram: it means that
    there exists some set $I' \subseteq [l]$ containing $I$ such that
    $Q' = P[I']$.
    
    Take $w \in \Sigma^n$.

    Assume now that $t$ is a subword of $\evalP{P}(w)$.
    It means there exists $i \in [l]$ such that $f_i(w_{p_i}) = t_1$.
    By definition of $I$, there exists $i' \in I$ such that $i' \leq i$ and
    $f_{i'}(w_{p_{i'}}) = t_1$.
    Hence, as $f_{i'}(w_{p_{i'}})$ is a subword of $\evalP{P[I']}(w)$ (because
    $I' \subseteq I$), we get that $t = t_1$ is a subword of $\evalP{P[I']}(w)$.

    Assume finally that $t$ is a subword of $\evalP{P[I']}(w)$. Then it is
    obviously a subword of $\evalP{P}(w)$, as $\evalP{P[I']}(w)$ is a subword of
    $\evalP{P}(w)$.

    Therefore, $t$ is a subword of $\evalP{P}(w)$ if and only if $t$ is a
    subword of $\evalP{Q'}(w) = \evalP{P[I']}(w)$, as desired.
    
    \emph{Subcase $k = 0$.}
    Let $P$ be a program over $M$ on $\Sigma^n$ for $n \in \N$ of length
    $l \in \N$ and some word $t \in M^0$.
    
    We claim that $Q = \emptyword$, a subprogram of $P$, has the desired
    properties.

    First notice that the length of $Q$ is
    $0 \leq 0! \cdot \card{\Sigma}^{\ceiling{0 / 2}} \cdot n^{\ceiling{0 / 2}} =
     c_0 \cdot n^{\ceiling{0 / 2}}$,
    at most the required length.

    Let $Q'$ be a subprogram of $P$ that has $Q$ as a subprogram.
    As $t \in M^0$, we necessarily have that $t = \emptyword$, which is a
    subword of any word in $M^*$. Therefore, we immediately get that for all
    $w \in \Sigma^n$, the word $t$ is a subword of $\evalP{P}(w)$ if and only if
    $t$ is a subword of $\evalP{Q'}(w)$, as desired.
\end{proof}

\section{Regular Languages in \texorpdfstring{$\Prog{\FMVJ}$}{P(J)}}
\label{sec:Regular_languages}
The second part of our investigation of the computational power of programs over
monoids in $\FMVJ$ is dedicated to understanding exactly what regular languages
can be \pr-recognised by monoids in $\FMVJ$.

\subsection{Non-tameness of \texorpdfstring{$\FMVJ$}{J}}
\label{sse:Non-tameness_of_J}

It is shown in~\cite{Grosshans-McKenzie-Segoufin-2017} that
$\Prog{\FMVJ} \cap \Reg \nsubseteq \DLang{\StVQJ}$, thus giving an example of a
well-known subvariety of $\FMVA$ for which \pr-recognition allows to do
unexpected things when recognising a regular language.
How far does this unexpected power go?

The first thing to notice is that, though none of them is in $\DLang{\StVQJ}$,
all languages of the form $\Sigma^* u$ and $u \Sigma^*$ for $\Sigma$ an alphabet
and $u \in \Sigma^+$ are in $\Prog{\FMVJ}$. Indeed, each of them can be
recognised by a sequence of constant-length programs over the syntactic monoid
of $u \shuffle \Sigma^*$: for every input length, just output the image, through
the syntactic morphism of $u \shuffle \Sigma^*$, of the word made of the
$\length{u}$ first or last letters. So, informally stated, programs over monoids
in $\FMVJ$ can check for some constant-length beginning or ending of their input
words.

But they can do much more.
Indeed, the language $(a + b)^* a c^+$ does not belong to $\DLang{\StVQJ}$
(compute the stable monoid), yet it is in $\Prog{\FMVJ}$. The crucial insight is
that it can be program-reduced in linear length to the piecewise testable
language of all words over $\set{a, b, c}$ having $ca$ as a subword but not the
subwords $cca$, $caa$ and $cb$ by using the following trick (that we shall call
``feedback-sweeping'') for input length $n \in \N$: read the input letters in
the order $2, 1, 3, 2, 4, 3, 5, 4, \ldots, n, n - 1$, output the letters read.
This has already been observed
in~\cite[Proposition 5]{Grosshans-McKenzie-Segoufin-2017}; here we give a formal
proof of the following lemma.

\begin{lemma}
\label{lem:Unexpected_language_in_P(J)}
    $(a + b)^* a c^+ \in \Prog{\FMVJ, n}$.
\end{lemma}

\begin{proof}
    Let $\Sigma = \set{a, b, c}$.

    Let
    \[
	L = ca \shuffle \Sigma^* \cap (cca \shuffle \Sigma^*)^\complement \cap
	    (caa \shuffle \Sigma^*)^\complement \cap
	    (cb \shuffle \Sigma^*)^\complement
    \]
    be the language of all words over $\Sigma$ having $ca$ as a subword but not
    the subwords $cca$, $caa$ and $cb$, that by construction is piecewise
    testable, i.e. belongs to $\DLang{\FMVJ}$.

    We are now going to build a program-reduction from $(a + b)^* a c^+$ to $L$.
    Let $n \in \N$. If $n \leq 1$, we set $\Psi_n$ to be $\emptyword$, the empty
    $\Sigma$-program on $\Sigma^n$. Otherwise, if $n \geq 2$, we set
    \[
	\Psi_n = (2, \id_\Sigma) (1, \id_\Sigma) (3, \id_\Sigma) (2, \id_\Sigma)
		 (4, \id_\Sigma) (3, \id_\Sigma) \cdots
		 (n, \id_\Sigma) (n - 1, \id_\Sigma)
	\displaypunct{.}
    \]

    Let us define $s\colon \N \to \N$ by $s(n) = \length{\Psi_n}$ for all
    $n \in \N$, which is such that
    \[
	s(n) =
	\begin{cases}
	    0 & \text{if $n \leq 1$}\\
	    2 n - 2 & \text{otherwise ($n \geq 2$)}
	\end{cases}
    \]
    for all $n \in \N$.
    Fix $n \in \N$.

    Let $w \in ((a + b)^* a c^+)^{=n}$: it means $n \geq 2$ and there exist
    $u \in (a + b)^{n_1}$ with $n_1 \in \intinterval{0}{n - 2}$ and
    $n_2 \in \intinterval{0}{n - 2}$ verifying that
    $w = u a c c^{n_2}$ and $n_1 + n_2 = n - 2$. We therefore have
    \[
	\Psi_n(w) =
	\begin{cases}
	    c a c^{2 n_2} & \text{when $n_1 = 0$}\\
	    u_2 u_1 \cdots u_{n_1} u_{n_1 - 1} a u_{n_1} c a c^{2 n_2} &
		\text{otherwise ($n_1 > 0$)}
	\displaypunct{,}
	\end{cases}
    \]
    a word easily seen to belong to $L^{= 2 n - 2}$. Since this is true for all
    $w \in ((a + b)^* a c^+)^{=n}$, it follows that
    $((a + b)^* a c^+)^{=n} \subseteq \Psi_n^{-1}(L^{=s(n)})$.

    Let conversely $w \in \Psi_n^{-1}(L^{=s(n)})$. Since this means that
    $\Psi_n(w) \in L^{=s(n)}$, we necessarily have $n \geq 2$ as it must contain
    $ca$ as a subword, so that
    \[
	\Psi_n(w) = w_2 w_1 w_3 w_2 w_4 w_3 \cdots w_n w_{n - 1}
	\displaypunct{.}
    \]
    Let $i, j \in [n]$ verifying that $w_i = c$, that $w_j = a$ and $w_i w_j$ is
    a subword of $\Psi_n(w)$. This means that $j \geq i - 1$, and we will now
    show that, actually, $j = i - 1$. Assume that $j \geq i + 2$; by
    construction, this would mean that $w_i w_j w_j = c a a$ is a subword of
    $\Psi_n(w)$, a contradiction to the fact it belongs to $L$. Assume otherwise
    that $j = i + 1$; by construction, this would either mean that
    $w_i w_{i - 1} w_{i + 1} w_i$ is a subword of $\Psi_n(w)$, which would imply
    one of $c a a$, $c b a$ and $c c a$ is a subword of $\Psi_n(w)$, or that 
    $w_{i + 1} w_i w_{i + 2} w_{i + 1}$ is a subword of $\Psi_n(w)$, which would
    imply one of $c a a$, $c b a$ and $c c a$ is a subword of $\Psi_n(w)$, in
    both cases contradicting the fact $\Psi_n(w)$ belongs to $L$. Hence, we
    indeed have $j = i - 1$, and in particular that $i \geq 2$.
    Now, by construction, for each $t \in [i - 2]$, we have that
    $w_t w_i w_{i - 1} = w_t c a$ is a subword of $\Psi_n(w) \in L$, so that
    $w_t$ cannot be equal to $c$. Similarly, for each
    $t \in \intinterval{i + 1}{n}$, we have that $w_i w_{i - 1} w_t = c a w_t$
    is a subword of $\Psi_n(w) \in L$, so that $w_t$ must be equal to $c$.
    This means that $w_1 \cdots w_{i - 2} \in (a + b)^*$ and
    $w_{i + 1} \cdots w_n \in c^*$, so that
    $w \in (a + b)^* a c c^* = (a + b)^* a c^+$. Since this is true for all
    $w \in \Psi_n^{-1}(L^{=s(n)})$, it follows that
    $((a + b)^* a c^+)^{=n} \supseteq \Psi_n^{-1}(L^{=s(n)})$.

    Therefore, we have that $((a + b)^* a c^+)^{=n} = \Psi_n^{-1}(L^{=s(n)})$
    for all $n \in \N$, so $(\Psi_n)_{n \in \N}$ is a program reduction from
    $(a + b)^* a c^+$ to $L$ of length $s(n)$.
    So since $L \in \DLang{\FMVJ}$, we can conclude that
    $(a + b)^* a c^+ \in \Prog{\FMVJ, s(n)} = \Prog{\FMVJ, n}$ by
    Proposition~\ref{ptn:Program-reduction_to_regular_language}.
\end{proof}

Using variants of the ``feedback-sweeping'' reading technique, we can prove that
the phenomenon just described is not an isolated case.

\begin{lemma}
\label{lem:Examples_regular_languages_in_P(J)}
    The languages $(a + b)^* a c^+$, $(a + b)^* a c^+ a (a + b)^*$,
    $c^+ a (a + b)^* a c^+$, $(a + b)^* b a c^+$ and
    $(a + b)^* a c^+ (a + b)^* a c^+$ do all belong to
    $\Prog{\FMVJ} \setminus \DLang{\StVQJ}$.
\end{lemma}

Hence, we are tempted to say that there are ``much more'' regular languages in
$\Prog{\FMVJ}$ than just those in $\DLang{\StVQJ}$, even though it is not clear
to us whether $\DLang{\StVQJ} \subseteq \Prog{\FMVJ}$ or not. But can we show
any upper bound on $\Prog{\FMVJ} \cap \Reg$? It turns out that we can, relying
on two known results.

First, since $\FMVJ \subseteq \FMVDA$, we have
$\Prog{\FMVJ} \subseteq \Prog{\FMVDA}$, so Theorem~6
in~\cite{Grosshans-McKenzie-Segoufin-2017}, that states
$\Prog{\FMVDA} \cap \Reg = \DLang{\StVQDA}$, implies that
$\Prog{\FMVJ} \cap \Reg \subseteq \DLang{\StVQDA}$.

Second, let us define an important superclass of the class of piecewise testable
languages.
Let $\Sigma$ be an alphabet and $u_1, \ldots, u_k \in \Sigma^+$
($k \in \N_{>0}$); we define
$\ccddo{u_1, \ldots, u_k} = \Sigma^* u_1 \Sigma^* \cdots \Sigma^* u_k \Sigma^*$.
The \emph{class of dot-depth one languages} is the class of Boolean combinations
of languages of the form $\Sigma^* u$, $u \Sigma^*$ and
$\ccddo{u_1, \ldots, u_k}$ for $\Sigma$ an alphabet, $k \in \N_{>0}$ and
$u, u_1, \ldots, u_k \in \Sigma^+$.
The inclusion-wise smallest variety of semigroups containing all syntactic
semigroups of dot-depth one languages is denoted by $\FSVJsdD$ and verifies that
$\DLang{\FSVJsdD}$ is exactly the class of dot-depth one languages.
(See~\cite{Straubing-1985,Maciel-Peladeau-Therien-2000,Pin-2017}.)
It has been shown in~\cite[Corollary 8]{Maciel-Peladeau-Therien-2000} that
$\Prog{\FSVJsdD} \cap \Reg = \DLang{\StVQJsdD}$ (if we extend the
program-over-monoid formalism in the obvious way to finite semigroups).
Now, we have $\FMVJ \subseteq \FSVJsdD$, so that
$\Prog{\FMVJ} \subseteq \Prog{\FSVJsdD}$ and hence
$\Prog{\FMVJ} \cap \Reg \subseteq \DLang{\StVQJsdD}$.

To summarise, we have the following.

\begin{proposition}
    $\Prog{\FMVJ} \cap \Reg \subseteq \DLang{\StVQDA} \cap \DLang{\StVQJsdD}$.
\end{proposition}

In fact, we conjecture that the inverse inclusion does also hold.

\begin{conjecture}
\label{cjt:Regular_languages_in_P(J)}
    $\Prog{\FMVJ} \cap \Reg = \DLang{\StVQDA} \cap \DLang{\StVQJsdD}$.
\end{conjecture}

Why do we think this should be true?
Though, for a given alphabet $\Sigma$, we cannot decide whether some word
$u \in \Sigma^+$ of length at least $2$ appears as a factor of any given word
$w$ in $\Sigma^*$ with programs over monoids in $\FMVJ$ (because
$\Sigma^* u \Sigma^* \notin \DLang{\StVQDA}$),
Lemma~\ref{lem:Examples_regular_languages_in_P(J)} and the possibilities offered
by the ``feedback-sweeping'' technique give the impression that we can do it
when we are guaranteed that $u$ appears at most a fixed number of times in $w$,
which seems somehow to be what dot-depth one languages become when restricted to
belong to $\DLang{\StVQDA}$. This intuition motivates the definition of
\emph{threshold dot-depth one languages}.

\subsection{Threshold dot-depth one languages}

The idea behind the definition of threshold dot-depth one languages is that we
take the basic building blocks of dot-depth one languages, of the form
$\ccddo{u_1, \ldots, u_k}$ for an alphabet $\Sigma$, for $k \in \N_{>0}$ and
$u_1, \ldots, u_k \in \Sigma^+$, and restrict them so that, given $l \in \N_{>0}$,
membership of a word does really depend on the presence of a given word $u_i$ as
a factor if and only if it appears less than $l$ times as a subword.

\begin{definition}
    Let $\Sigma$ be an alphabet.
    For all $u \in \Sigma^+$ and $l \in \N_{>0}$, we define
    $\ccddo{u}_l$ to be the language of words over $\Sigma$ containing $u^l$ as
    a subword or $u$ as a factor, i.e.
    $\ccddo{u}_l = \Sigma^* u \Sigma^* \cup u^l \shuffle \Sigma^*$.
    Then, for all $u_1, \ldots, u_k \in \Sigma^+$ ($k \in \N, k \geq 2$) and
    $l \in \N_{>0}$, we define
    $\ccddo{u_1, \ldots, u_k}_l = \ccddo{u_1}_l \cdots \ccddo{u_k}_l$.
\end{definition}

Obviously, for each $\Sigma$ an alphabet, $k, l \in \N_{>0}$ and
$u_1, \ldots, u_k \in \Sigma^+$, the language $\ccddo{u_1, \ldots, u_k}_l$
equals $u_1 \cdots u_k \shuffle \Sigma^*$ when $l = 1$ or $u_1, \ldots, u_k$ are
all restricted to one letter.
Over $\set{a, b, c}$, the language $\ccddo{ab, c}_3$ contains all words
containing a letter $c$ verifying that in the prefix up to that letter, $ababab$
appears as a subword or $ab$ appears as a factor.
Finally, the language $(a + b)^* a c^+$ over $\set{a, b, c}$ of
Lemma~\ref{lem:Unexpected_language_in_P(J)} is equal to
${\ccddo{c, a}_2}^\complement \cap {\ccddo{c, b}_2}^\complement \cap
 \ccddo{ac}_2$.

We then define a \emph{threshold dot-depth one language} as any Boolean
combination of languages of the form $\Sigma^* u$, $u \Sigma^*$ and
$\ccddo{u_1, \ldots, u_k}_l$ for $\Sigma$ an alphabet, for $k, l \in \N_{>0}$
and $u, u_1, \ldots, u_k \in \Sigma^+$.

Confirming the intuition briefly given above, the technique of
``feedback-sweeping'' can indeed be pushed further to prove that the whole class
of threshold dot-depth one languages is contained in $\Prog{\FMVJ}$, and we
dedicate the remainder of this section to prove it.
Concerning Conjecture~\ref{cjt:Regular_languages_in_P(J)}, our
intuition leads us to believe that, in fact, the class of threshold dot-depth
one languages with additional positional modular counting is exactly
$\DLang{\StVQDA} \cap \DLang{\StVQJsdD}$. In support of this belief, in the next
section (Section~\ref{sec:Algebraic_characterisation_TDDO_languages}) we prove
that the class of threshold dot-depth one languages is exactly
$\DLang{\FMVDA} \cap \DLang{\FSVJsdD}$.

Let us now move on to the proof of the following theorem.

\begin{theorem}
\label{thm:TDDO_languages_in_P(J)}
    Every threshold dot-depth one language belongs to $\Prog{\FMVJ}$.
\end{theorem}

As $\Prog{\FMVJ}$ is closed under Boolean operations
(Proposition~\ref{lemma-simple-closure-P}), our goal is to prove, given an
alphabet $\Sigma$, given $l \in \N_{>0}$ and $u_1, \ldots, u_k \in \Sigma^+$
($k \in \N_{>0}$), that $\ccddo{u_1, \ldots, u_k}_l$ is in $\Prog{\FMVJ}$; the
case of $\Sigma^* u$ and $u \Sigma^*$ for $u \in \Sigma^+$ is easily handled
(see the discussion at the beginning of Subsection~\ref{sse:Non-tameness_of_J}).
To do this, we need to put $\ccddo{u_1, \ldots, u_k}_l$ in some normal form. It
is readily seen that
$\ccddo{u_1, \ldots, u_k}_l =
 \bigcup_{q_1, \ldots, q_k \in \set{1, l}}
 L^{(l)}_{(u_1, q_1)} \cdots L^{(l)}_{(u_k, q_k)}$
where the $L^{(l)}_{(u_i, q_i)}$'s are defined thereafter.

\begin{definition}
\label{def:TDDO_alternative}
    Let $\Sigma$ be an alphabet.

    For all $u \in \Sigma^+$, $l \in \N_{>0}$ and $\alpha \in [l]$, set
    $L^{(l)}_{(u, \alpha)} =
     \begin{cases}
	\Sigma^* u \Sigma^* & \text{if $\alpha < l$}\\
	u^l \shuffle \Sigma^* & \text{otherwise}
     \end{cases}$.
\end{definition}

Building directly a sequence of programs over a monoid in $\FMVJ$ that decides
$L^{(l)}_{(u_1, q_1)} \cdots L^{(l)}_{(u_k, q_k)}$ for some alphabet $\Sigma$
and $q_1, \ldots, q_k \in \set{1, l}$ seems however tricky. We need to split
things further by controlling precisely how many times each $u_i$ for
$i \in [k]$ appears in the right place when it does less than $l$ times.
To do this, we consider, for each $\alpha \in [l]^k$, the language
$R_l^\alpha(u_1, \ldots, u_k)$ defined below.

\begin{definition}
    Let $\Sigma$ be an alphabet.

    For all $u_1, \ldots, u_k \in \Sigma^+$ ($k \in \N_{>0}$), $l \in \N_{>0}$,
    $\alpha \in [l]^k$, we set
    \begin{align*}
	R_l^\alpha(u_1, \ldots, u_k)
	= & ({u_1}^{\alpha_1} \cdots {u_k}^{\alpha_k}) \shuffle \Sigma^* \cap\\
	& \bigcap_{i \in [k], \alpha_i < l}
	  \bigl(({u_1}^{\alpha_1} \cdots {u_i}^{\alpha_i + 1} \cdots
		 {u_k}^{\alpha_k}) \shuffle \Sigma^*\bigr)^\complement
	\displaypunct{.}
    \end{align*}
\end{definition}

Now, for a given $\alpha \in [l]^k$, we are interested in the words of
$R_l^\alpha(u_1, \ldots, u_k)$ such that for each $i \in [k]$ verifying
$\alpha_i < l$, the word $u_i$ indeed appears as a factor in the right place. We
thus introduce a last language $S_l^\alpha(u_1, \ldots, u_k)$ defined as
follows.

\begin{definition}
    Let $\Sigma$ be an alphabet.

    For all $u_1, \ldots, u_k \in \Sigma^+$ ($k \in \N_{>0}$), $l \in \N_{>0}$,
    $\alpha \in [l]^k$, we set
    \[
	S_l^\alpha(u_1, \ldots, u_k) =
	\bigcap_{i \in [k], \alpha_i < l}
	\bigl(({u_1}^{\alpha_1} \cdots {u_{i - 1}}^{\alpha_{i - 1}}) \shuffle
	      \Sigma^*\bigr)
	u_i
	\bigl(({u_{i + 1}}^{\alpha_{i + 1}} \cdots {u_k}^{\alpha_k}) \shuffle
	      \Sigma^*\bigr)
	.
    \]
\end{definition}

We now have the normal form we were looking for to prove
Theorem~\ref{thm:TDDO_languages_in_P(J)}:
$\ccddo{u_1, \ldots, u_k}_l$ is equal to the union, over all $\alpha \in [l]^k$,
of the intersection of $R_l^\alpha(u_1, \ldots, u_k)$ and
$S_l^\alpha(u_1, \ldots, u_k)$.
Though rather intuitive, the correctness of this decomposition is not so
straightforward to prove and, actually, we can only prove it when for each
$i \in [k]$, the letters in $u_i$ are all distinct.

\begin{lemma}
\label{lem:TDDO-Equality}
    Let $\Sigma$ be an alphabet, $l \in \N_{>0}$ and
    $u_1, \ldots, u_k \in \Sigma^+$ ($k \in \N_{>0}$) such that for each
    $i \in [k]$, the letters in $u_i$ are all distinct. Then,
    \[
	\bigcup_{q_1, \ldots, q_k \in \set{1, l}}
	L^{(l)}_{(u_1, q_1)} \cdots L^{(l)}_{(u_k, q_k)} =
	\bigcup_{\alpha \in [l]^k}
	\bigl(R_l^\alpha(u_1, \ldots, u_k) \cap
	      S_l^\alpha(u_1, \ldots, u_k)\bigr)
	\displaypunct{.}
    \]
\end{lemma}

\begin{proof}
    Let $\Sigma$ be an alphabet and $l \in \N_{>0}$. We prove it by induction on
    $k \in \N_{>0}$.
    
    \paragraph{Base case $k = 1$.}
    Let $u_1 \in \Sigma^+$ such that the letters in $u_1$ are all distinct.
    It is clear that
    \begin{align*}
	& \bigcup_{q_1 \in \set{1, l}} L^{(l)}_{(u_1, q_1)}\\
	= & (\Sigma^* u_1 \Sigma^* \cup {u_1}^l \shuffle \Sigma^*)\\
	= & \Bigl(\bigcup_{\alpha_1 = 1}^{l - 1}
		  \bigl({u_1}^{\alpha_1} \shuffle \Sigma^* \cap
			({u_1}^{\alpha_1 + 1} \shuffle
			\Sigma^*)^\complement \cap
			\Sigma^* u_1 \Sigma^*\bigr) \cup
		  ({u_1}^l \shuffle \Sigma^*)\Bigr)\\
	= & \bigcup_{\alpha_1 \in [l]}
	    \bigl(R_l^{\alpha_1}(u_1) \cap S_l^{\alpha_1}(u_1)\bigr)
	\displaypunct{.}
    \end{align*}

    \paragraph{Induction.}
    Let $k \in \N_{>0}$ and assume that for all $u_1, \ldots, u_k \in \Sigma^+$
    such that for each $i \in [k]$, the letters in $u_i$ are all distinct, we
    have
    \[
	\bigcup_{q_1, \ldots, q_k \in \set{1, l}}
	L^{(l)}_{(u_1, q_1)} \cdots L^{(l)}_{(u_k, q_k)} =
	\bigcup_{\alpha \in [l]^k}
	\bigl(R_l^\alpha(u_1, \ldots, u_k) \cap
	      S_l^\alpha(u_1, \ldots, u_k)\bigr)
	\displaypunct{.}
    \]

    Let now $u_1, \ldots, u_{k + 1} \in \Sigma^+$ such that for each
    $i \in [k + 1]$, the letters in $u_i$ are all distinct.

    \emph{Right-to-left inclusion.}
    Let
    \[
	w \in
	\bigcup_{\alpha \in [l]^{k + 1}}
	\bigl(R_l^\alpha(u_1, \ldots, u_{k + 1}) \cap
	      S_l^\alpha(u_1, \ldots, u_{k + 1})\bigr)
	\displaypunct{.}
    \]

    Let $\alpha \in [l]^{k + 1}$ witnessing this fact.
    As $w \in R_l^\alpha(u_1, \ldots, u_{k + 1})$, we can decompose it as
    $w = x y$ where
    $x \in ({u_1}^{\alpha_1} \cdots {u_k}^{\alpha_k}) \shuffle \Sigma^*$ and
    $y \in {u_{k + 1}}^{\alpha_{k + 1}} \shuffle \Sigma^*$ with $\length{y}$
    being minimal.
    What we are going to do is, on the one hand, to prove that 
    $x \in R_l^{\alpha'}(u_1, \ldots, u_k) \cap
	   S_l^{\alpha'}(u_1, \ldots, u_k)$
    where $\alpha' = (\alpha_1, \ldots, \alpha_k)$, so that we can apply the
    inductive hypothesis on $x$ and get that there exist
    $q_1, \ldots, q_k \in \set{1, l}$ such that
    $x \in L^{(l)}_{(u_1, q_1)} \cdots L^{(l)}_{(u_k, q_k)}$; and, on the other
    hand, we are going to prove that there exists $q_{k + 1} \in \set{1, l}$
    verifying $y \in L^{(l)}_{(u_{k + 1}, q_{k + 1})}$.
    We now spell out the details.
    
    For each $i \in [k], \alpha_i < l$, we have
    $x \notin ({u_1}^{\alpha_1} \cdots {u_i}^{\alpha_i + 1} \cdots
	       {u_k}^{\alpha_k}) \shuffle \Sigma^*$,
    otherwise we would have
    $w = x y \in ({u_1}^{\alpha_1} \cdots {u_i}^{\alpha_i + 1} \cdots
		  {u_{k + 1}}^{\alpha_{k + 1}}) \shuffle \Sigma^*$.
    Also, for all $i \in [k], \alpha_i < l$, we have that
    $x \in \bigl(({u_1}^{\alpha_1} \cdots {u_{i - 1}}^{\alpha_{i - 1}}) \shuffle
		 \Sigma^*\bigr)
	   u_i
	   \bigl(({u_{i + 1}}^{\alpha_{i + 1}} \cdots {u_k}^{\alpha_k}) \shuffle
		 \Sigma^*\bigr)$,
    otherwise it would mean that $y = y_1 y_2$ with $\length{y_1} > 0$, that
    $x y_1 \in \bigl(({u_1}^{\alpha_1} \cdots \allowbreak
		      {u_{i - 1}}^{\alpha_{i - 1}})
		     \shuffle \Sigma^*\bigr)
	       u_i
	       \bigl(({u_{i + 1}}^{\alpha_{i + 1}} \cdots {u_k}^{\alpha_k})
		     \shuffle \Sigma^*\bigr)$
    and $y_2 \in {u_{k + 1}}^{\alpha_{k + 1}} \shuffle \Sigma^*$,
    contradicting the minimality of $\length{y}$.
    So $x \in R_l^{\alpha'}(u_1, \ldots, u_k) \cap
	      S_l^{\alpha'}(u_1, \ldots, u_k)$,
    which means by inductive hypothesis that there exist
    $q_1, \ldots, q_k \in \set{1, l}$ such that
    $x \in L^{(l)}_{(u_1, q_1)} \cdots L^{(l)}_{(u_k, q_k)}$.

    Remember now that the letters in $u_{k + 1}$ are all distinct. 
    If $\alpha_{k + 1} < l$, since
    $w \in
     \bigl(({u_1}^{\alpha_1} \ldots {u_k}^{\alpha_k}) \shuffle \Sigma^*\bigr)
     u_{k + 1} \Sigma^*$,
    we must have $y \in \Sigma^* u_{k + 1} \Sigma^*$. Indeed, by minimality of
    $\length{y}$, the word $y$ starts with the first letter of $u_{k + 1}$,
    which has pairwise distinct letters, so that $u_{k + 1}$ cannot appear as a
    factor of $xy$ partly in $x$ and partly in $y$; so if it were the case that
    $y$ does not contain $u_{k + 1}$ as a factor, we would have
    $x \in
     \bigl(({u_1}^{\alpha_1} \ldots {u_k}^{\alpha_k}) \shuffle \Sigma^*\bigr)
     u_{k + 1} \Sigma^*$,
    so that
    $x y = w \in ({u_1}^{\alpha_1} \ldots {u_k}^{\alpha_k}
		  {u_{k + 1}}^{\alpha_{k + 1} + 1}) \shuffle \Sigma^*$,
    a contradiction with the hypothesis on $w$.
    Hence, $y \in L^{(l)}_{(u_{k + 1}, \alpha_{k + 1})}$.
    If $\alpha_{k + 1} = l$, then
    $y \in {u_{k + 1}}^{\alpha_{k + 1}} \shuffle \Sigma^* =
     L^{(l)}_{(u_{k + 1}, \alpha_{k + 1})}$.
    So, if we set
    $q_{k + 1} = \begin{cases}
		 1 & \text{if $\alpha_{k + 1} < l$}\\
		 l & \text{otherwise}
		 \end{cases}$,
    then we get that $y \in L^{(l)}_{(u_{k + 1}, q_{k + 1})}$.

    We can conclude that
    $w = x y \in L^{(l)}_{(u_1, q_1)} \cdots L^{(l)}_{(u_k, q_k)}
		 L^{(l)}_{(u_{k + 1}, q_{k + 1})}$.

    \emph{Left-to-right inclusion.}
    Let
    $w \in \bigcup_{q_1, \ldots, q_{k + 1} \in \set{1, l}}
	   L^{(l)}_{(u_1, q_1)} \cdots L^{(l)}_{(u_{k + 1}, q_{k + 1})}$.
    The rough idea of our proof here is to take $\alpha_{k + 1} \in [l]$ the
    biggest integer in $[l]$ such that
    $w \in \bigl(\bigcup_{q_1, \ldots, q_k \in \set{1, l}}
		 L^{(l)}_{(u_1, q_1)} \cdots L^{(l)}_{(u_k, q_k)})\bigr)
	   ({u_{k + 1}}^{\alpha_{k + 1}} \shuffle \Sigma^*)$
    and decompose $w$ as $w = x y$ where
    $x \in \bigcup_{q_1, \ldots, q_k \in \set{1, l}}
	   L^{(l)}_{(u_1, q_1)} \cdots L^{(l)}_{(u_k, q_k)}$ and
    $y \in {u_{k + 1}}^{\alpha_{k + 1}} \shuffle \Sigma^*$ with $\length{y}$
    being minimal. By inductive hypothesis, we know there exists
    $\alpha \in [l]^k$ such that
    $x \in R_l^\alpha (u_1, \ldots, u_k) \cap S_l^\alpha (u_1, \ldots, u_k)$ and
    we then prove that 
    $x y \in R_l^{(\alpha_1, \ldots, \alpha_{k + 1})}(u_1, \ldots, u_{k + 1})
	     \cap
	     S_l^{(\alpha_1, \ldots, \alpha_{k + 1})}(u_1, \ldots, u_{k + 1})$
    by distinguishing between the case in which $\alpha_{k + 1} = l$ and the
    case in which $\alpha_{k + 1} < l$. The first one is easy to handle, the
    second one is much trickier.

    We now spell out the details.
    \begin{itemize}
	\item
	    Suppose we have
	    \begin{align*}
		w \in &
		    \bigcup_{q_1, \ldots, q_k \in \set{1, l}}
		    L^{(l)}_{(u_1, q_1)} \cdots L^{(l)}_{(u_k, q_k)}
		    L^{(l)}_{(u_{k + 1}, l)}\\
		= &
		   \Bigl(\bigcup_{q_1, \ldots, q_k \in \set{1, l}}
			 L^{(l)}_{(u_1, q_1)} \cdots L^{(l)}_{(u_k, q_k)}\Bigr)
		   ({u_{k + 1}}^l \shuffle \Sigma^*)
		\displaypunct{.}
	    \end{align*}
	    Then $w$ can be decomposed as $w = x y$ where
	    $x \in \bigcup_{q_1, \ldots, q_k \in \set{1, l}}
		   L^{(l)}_{(u_1, q_1)} \cdots \allowbreak
		   L^{(l)}_{(u_k, q_k)}$
	    and $y \in {u_{k + 1}}^l \shuffle \Sigma^*$ with $\length{y}$ being
	    minimal.
	    So by inductive hypothesis, there exists $\alpha \in [l]^k$ such
	    that
	    $x \in R_l^\alpha(u_1, \ldots, u_k) \cap
		   S_l^\alpha(u_1, \ldots, u_k)$.
	    Observe that this means we have
	    $w \in ({u_1}^{\alpha_1} \cdots {u_k}^{\alpha_k} {u_{k + 1}}^l)
		   \shuffle \Sigma^*$
	    and for each $i \in [k], \alpha_i < l$, that
	    $w \notin ({u_1}^{\alpha_1} \cdots {u_i}^{\alpha_i + 1} \cdots
		       \allowbreak
		       {u_k}^{\alpha_k} {u_{k + 1}}^l) \shuffle \Sigma^*$,
	    otherwise it would mean that
	    $x \in ({u_1}^{\alpha_1} \cdots {u_i}^{\alpha_i + 1} \cdots
		    \allowbreak {u_k}^{\alpha_k}) \shuffle \Sigma^*$
	    by minimality of $\length{y}$. Similarly, for all
	    $i \in [k], \alpha_i < l$, it is obvious that we have
	    \[
		w = x y \in \bigl(({u_1}^{\alpha_1} \cdots
				  {u_{i - 1}}^{\alpha_{i - 1}}) \shuffle
				  \Sigma^*\bigr)
			    u_i
			    \bigl(({u_{i + 1}}^{\alpha_{i + 1}} \cdots
				  {u_k}^{\alpha_k} {u_{k + 1}}^l) \shuffle
				  \Sigma^*\bigr)
	    \]
	    as
	    $x \in \bigl(({u_1}^{\alpha_1} \cdots {u_{i - 1}}^{\alpha_{i - 1}})
			 \shuffle \Sigma^*\bigr)
		   u_i
		   \bigl(({u_{i + 1}}^{\alpha_{i + 1}} \cdots {u_k}^{\alpha_k})
			 \shuffle \Sigma^*\bigr)$
	    and $y \in {u_{k + 1}}^l \allowbreak \shuffle \Sigma^*$.
	    Hence,
	    $w \in
	     R_l^{(\alpha_1, \ldots, \alpha_{k + 1})}(u_1, \ldots, u_{k + 1})
	     \cap
	     S_l^{(\alpha_1, \ldots, \alpha_{k + 1})}(u_1, \ldots, u_{k + 1})$.

	\item
	    Or we have
	    \begin{align*}
		w \notin & \bigcup_{q_1, \ldots, q_k \in \set{1, l}}
			   L^{(l)}_{(u_1, q_1)} \cdots L^{(l)}_{(u_k, q_k)}
			   L^{(l)}_{(u_{k + 1}, l)}\\
	       	= & \Bigl(\bigcup_{q_1, \ldots, q_k \in \set{1, l}}
			  L^{(l)}_{(u_1, q_1)} \cdots L^{(l)}_{(u_k, q_k)}\Bigr)
		    ({u_{k + 1}}^l \shuffle \Sigma^*)
	    \end{align*}
	    but
	    \[
		w \in \bigcup_{q_1, \ldots, q_k \in \set{1, l}}
		      L^{(l)}_{(u_1, q_1)} \cdots
		      L^{(l)}_{(u_k, q_k)}
		      L^{(l)}_{(u_{k + 1}, 1)}
		\displaypunct{.}
	    \]
	    Let $\alpha_{k + 1} \in [l - 1]$ be the biggest integer in
	    $[l - 1]$ such that
	    \[
		w \in \Bigl(\bigcup_{q_1, \ldots, q_k \in \set{1, l}}
			    L^{(l)}_{(u_1, q_1)} \cdots
			    L^{(l)}_{(u_k, q_k)}\Bigr)
		      ({u_{k + 1}}^{\alpha_{k + 1}} \shuffle \Sigma^*)
	    \]
	    which does exist by hypothesis.
	    We can decompose $w$ as $w = x y$ where
	    $x \in \bigcup_{q_1, \ldots, q_k \in \set{1, l}}
		   L^{(l)}_{(u_1, q_1)} \cdots L^{(l)}_{(u_k, q_k)}$
	    and $y \in {u_{k + 1}}^{\alpha_{k + 1}} \shuffle \Sigma^*$ with
	    $\length{y}$ being minimal. So by inductive hypothesis, there exists
	    $\alpha \in [l]^k$ such that
	    $x \in R_l^\alpha(u_1, \ldots, u_k) \cap
		   S_l^\alpha(u_1, \ldots, u_k)$.
	    We are now going to prove that
	    \[
		w = x y \in
		R_l^{(\alpha_1, \ldots, \alpha_{k + 1})}(u_1, \ldots, u_{k + 1})
		\cap
	   	S_l^{(\alpha_1, \ldots, \alpha_{k + 1})}(u_1, \ldots, u_{k + 1})
		\displaypunct{.}
	    \]

	    Among the obvious things to observe is that we have
	    $w \in ({u_1}^{\alpha_1} \cdots {u_k}^{\alpha_k} \allowbreak
		    {u_{k + 1}}^{\alpha_{k + 1}}) \shuffle \Sigma^*$
	    and for each $i \in [k], \alpha_i < l$, that
	    \[
		w \notin ({u_1}^{\alpha_1} \cdots {u_i}^{\alpha_i + 1}
			  \cdots {u_k}^{\alpha_k}
			  {u_{k + 1}}^{\alpha_{k + 1}}) \shuffle \Sigma^*
		\displaypunct{,}
	    \]
	    otherwise it would mean that
	    $x \in ({u_1}^{\alpha_1} \cdots {u_i}^{\alpha_i + 1} \cdots
		    {u_k}^{\alpha_k}) \shuffle \Sigma^*$
	    by minimality of $\length{y}$. Similarly, for all
	    $i \in [k], \alpha_i < l$, it is obvious that we have
	    \[
		w = x y \in \bigl(({u_1}^{\alpha_1} \cdots
				   {u_{i - 1}}^{\alpha_{i - 1}}) \shuffle
				  \Sigma^*\bigr)
			    u_i
			    \bigl(({u_{i + 1}}^{\alpha_{i + 1}} \cdots
				   {u_k}^{\alpha_k}
				   {u_{k + 1}}^{\alpha_{k + 1}}) \shuffle
				  \Sigma^*\bigr)
	    \]
	    because
	    $x \in \bigl(({u_1}^{\alpha_1} \cdots {u_{i - 1}}^{\alpha_{i - 1}})
			 \shuffle \Sigma^*\bigr)
		   u_i
		   \bigl(({u_{i + 1}}^{\alpha_{i + 1}} \cdots {u_k}^{\alpha_k})
			 \shuffle \Sigma^*\bigr)$
	    and $y \in {u_{k + 1}}^{\alpha_{k + 1}} \shuffle \Sigma^*$.

	    Now let us show that we have $y \in \Sigma^* u_{k + 1} \Sigma^*$.
	    Assume it weren't the case: the letters in $u_{k + 1}$ are pairwise
	    distinct and moreover $y$ starts with the first letter of
	    $u_{k + 1}$ by minimality of $\length{y}$, so $u_{k + 1}$ cannot
	    appear as a factor of $x y$ partly in $x$ and partly in $y$ and,
	    additionally,
	    \begin{align*}
		w & \in \bigcup_{q_1, \ldots, q_k \in \set{1, l}}
			L^{(l)}_{(u_1, q_1)} \cdots L^{(l)}_{(u_k, q_k)}
			L^{(l)}_{(u_{k + 1}, 1)}\\
		& = \Bigl(\bigcup_{q_1, \ldots, q_k \in \set{1, l}}
			  L^{(l)}_{(u_1, q_1)} \cdots L^{(l)}_{(u_k, q_k)}\Bigr)
		    \Sigma^* u_{k + 1} \Sigma^*
		\displaypunct{,}
	    \end{align*}
	    so we would have
	    $x \in (\bigcup_{q_1, \ldots, q_k \in \set{1, l}}
		    L^{(l)}_{(u_1, q_1)} \cdots L^{(l)}_{(u_k, q_k)})
		   \Sigma^* u_{k + 1} \Sigma^*$.
	    But this either contradicts the maximality of $\alpha_{k + 1}$ or
	    the fact that
	    \[
		w \notin \Bigl(\bigcup_{q_1, \ldots, q_k \in \set{1, l}}
			       L^{(l)}_{(u_1, q_1)} \cdots
			       L^{(l)}_{(u_k, q_k)}\Bigr)
			 ({u_{k + 1}}^l \shuffle \Sigma^*)
		\displaypunct{.}
	    \]
	    Thus, we have
	    $w = x y \in \bigl(({u_1}^{\alpha_1} \cdots {u_k}^{\alpha_k})
			       \shuffle \Sigma^*\bigr)
			 u_{k + 1} \Sigma^*$
	    as
	    $x \in ({u_1}^{\alpha_1} \cdots \allowbreak {u_k}^{\alpha_k})
		   \allowbreak \shuffle \Sigma^*$.

	    Let us finish with the trickiest part, namely showing that
	    $w \notin ({u_1}^{\alpha_1} \cdots \allowbreak {u_k}^{\alpha_k}
		       \allowbreak {u_{k + 1}}^{\alpha_{k + 1} + 1}) \shuffle
		      \Sigma^*$.
	    Assume that
	    $w \in ({u_1}^{\alpha_1} \cdots {u_k}^{\alpha_k}
		    {u_{k + 1}}^{\alpha_{k + 1} + 1}) \shuffle \Sigma^*$.
	    We then have that
	    $x \in ({u_1}^{\alpha_1} \cdots \allowbreak {u_k}^{\alpha_k}
		    u_{k + 1})
		   \shuffle \Sigma^*$,
	    otherwise it would mean that $y = y_1 y_2$ with
	    $\length{y_1} > 0$, with
	    $x y_1 \in ({u_1}^{\alpha_1} \cdots {u_k}^{\alpha_k} \allowbreak
			u_{k + 1})
		       \shuffle \Sigma^*$
	    and $y_2 \in {u_{k + 1}}^{\alpha_{k + 1}} \shuffle \Sigma^*$,
	    contradicting the minimality of $\length{y}$.
	    We can decompose $x$ as $x = x_1 x_2$ where
	    $x_1 \in ({u_1}^{\alpha_1} \cdots {u_k}^{\alpha_k}) \shuffle
		     \Sigma^*$
	    and $x_2 \in u_{k + 1} \shuffle \Sigma^*$ with $\length{x_2}$ being
	    minimal.
	    We claim that, actually,
	    $x_1 \in R_l^\alpha(u_1, \ldots, u_k) \cap
		     S_l^\alpha(u_1, \ldots, u_k)$,
	    so that by inductive hypothesis,
	    $x_1 \in \bigcup_{q_1, \ldots, q_k \in \set{1, l}}
		     L^{(l)}_{(u_1, q_1)} \cdots L^{(l)}_{(u_k, q_k)}$.
	    But since
	    $x_2 y \in {u_{k + 1}}^{\alpha_{k + 1} + 1} \shuffle \Sigma^*$, this
	    means that
	    \[
		w = x_1 x_2 y \in
		\Bigl(\bigcup_{q_1, \ldots, q_k \in \set{1, l}}
		      L^{(l)}_{(u_1, q_1)} \cdots L^{(l)}_{(u_k, q_k)}\Bigr)
		({u_{k + 1}}^{\alpha_{k + 1} + 1} \shuffle \Sigma^*)
		\displaypunct{,}
	    \]
	    contradicting the maximality of $\alpha_{k + 1}$ or the fact that
	    \[
		w \notin \Bigl(\bigcup_{q_1, \ldots, q_k \in \set{1, l}}
			       L^{(l)}_{(u_1, q_1)} \cdots
			       L^{(l)}_{(u_k, q_k)}\Bigr)
			 ({u_{k + 1}}^l \shuffle \Sigma^*)
		\displaypunct{.}
	    \]
	    So we can conclude that
	    $w \notin ({u_1}^{\alpha_1} \cdots {u_k}^{\alpha_k}
		       {u_{k + 1}}^{\alpha_{k + 1} + 1}) \shuffle \Sigma^*$.

	    The claim that
	    $x_1 \in R_l^\alpha(u_1, \ldots, u_k) \cap
		     S_l^\alpha(u_1, \ldots, u_k)$
	    remains to be shown.
	    We directly see that
	    $x_1 \notin ({u_1}^{\alpha_1} \cdots {u_i}^{\alpha_i + 1} \cdots
			 {u_k}^{\alpha_k}) \shuffle \Sigma^*$
	    for all $i \in [k], \alpha_i < l$, otherwise it would mean that
	    $x \in ({u_1}^{\alpha_1} \cdots {u_i}^{\alpha_i + 1} \cdots
		    \allowbreak
		    {u_k}^{\alpha_k}) \shuffle \Sigma^*$.
	    Let now $i \in [k], \alpha_i < l$, and assume that
	    $x_1 \notin \bigl(({u_1}^{\alpha_1} \cdots
			      {u_{i - 1}}^{\alpha_{i - 1}}) \shuffle
			      \Sigma^*\bigr)
			u_i \allowbreak
			\bigl(({u_{i + 1}}^{\alpha_{i + 1}} \cdots \allowbreak
			      {u_k}^{\alpha_k}) \shuffle \Sigma^*\bigr)$.
	    We can decompose $x_1$ as $x_1 = x_{1, 1} x_{1, 2}$ where
	    $x_{1, 1} \in ({u_1}^{\alpha_1} \cdots \allowbreak {u_i}^{\alpha_i})
			  \shuffle \Sigma^*$
	    and
	    $x_{1, 2} \in ({u_{i + 1}}^{\alpha_{i + 1}} \cdots
			   {u_k}^{\alpha_k}) \shuffle \Sigma^*$
	    with $\length{x_{1, 1}}$ being minimal.
	    By hypothesis, we have
	    $x_{1, 1} \notin \bigl(({u_1}^{\alpha_1} \cdots
				   {u_{i - 1}}^{\alpha_{i - 1}}) \shuffle
				   \Sigma^*\bigr)
			     u_i \Sigma^*$,
	    otherwise we would have
	    \[
		x_1 = x_{1, 1} x_{1, 2} \in
		\bigl(({u_1}^{\alpha_1} \cdots {u_{i - 1}}^{\alpha_{i - 1}})
		      \shuffle \Sigma^*\bigr)
		u_i
		\bigl(({u_{i + 1}}^{\alpha_{i + 1}} \cdots {u_k}^{\alpha_k})
		      \shuffle \Sigma^*\bigr)
		\displaypunct{.}
	    \]
	    As previously, the letters in $u_i$ are pairwise distinct, and
	    $x_{1, 1}$ ends with the last letter of $u_i$ by minimality of
	    $\length{x_{1, 1}}$, so $u_i$ cannot appear as a factor of $x$
	    partly in $x_{1, 1}$ and partly in $x_{1, 2} x_2$. Thus, we have
	    that
	    \[
		x_{1, 2} x_2 \in \Sigma^* u_i
				 \bigl(({u_{i + 1}}^{\alpha_{i + 1}} \cdots
					{u_k}^{\alpha_k}) \shuffle \Sigma^*
				 \bigr)
	    \]
	    because we know that
	    $x \in \bigl(({u_1}^{\alpha_1} \cdots {u_{i - 1}}^{\alpha_{i - 1}})
			 \shuffle \Sigma^*\bigr)
		   u_i
		   \bigl(({u_{i + 1}}^{\alpha_{i + 1}} \cdots {u_k}^{\alpha_k})
			 \allowbreak \shuffle \Sigma^*\bigr)$.
	    But this means that
	    $x = x_{1, 1} x_{1, 2} x_2 \in
	     ({u_1}^{\alpha_1} \cdots {u_i}^{\alpha_i + 1} \cdots
	      {u_k}^{\alpha_k}) \shuffle \Sigma^*$,
	    a contradiction. Hence, we can deduce that for all
	    $i \in [k], \alpha_i < l$, we have
	    $x_1 \in \bigl(({u_1}^{\alpha_1} \cdots
			    {u_{i - 1}}^{\alpha_{i - 1}}) \shuffle \Sigma^*
		     \bigr)
		     u_i
		     \bigl(({u_{i + 1}}^{\alpha_{i + 1}} \cdots
			    {u_k}^{\alpha_k}) \allowbreak \shuffle \Sigma^*
		     \bigr)$.
	    This finishes to show that
	    \[
		x_1 \in R_l^\alpha(u_1, \ldots, u_k) \cap
			S_l^\alpha(u_1, \ldots, u_k)
		\displaypunct{.}
	    \]

	    \medskip

	    Putting all together, we indeed also have that
	    \[
		w \in
		R_l^{(\alpha_1, \ldots, \alpha_{k + 1})}(u_1, \ldots, u_{k + 1})
	   	\cap
		S_l^{(\alpha_1, \ldots, \alpha_{k + 1})}(u_1, \ldots, u_{k + 1})
	    \]
	    in the present case.
    \end{itemize}

    In conclusion, in both cases,
    \[
	w \in \bigcup_{\alpha \in [l]^{k + 1}}
	      \bigl(R_l^\alpha(u_1, \ldots, u_{k + 1}) \cap
		    S_l^\alpha(u_1, \ldots, u_{k + 1})\bigr)
	\displaypunct{.}
    \]

    So we can finally conclude that
    \begin{align*}
	& \bigcup_{q_1, \ldots, q_{k + 1} \in \set{1, l}}
	   L^{(l)}_{(u_1, q_1)} \cdots L^{(l)}_{(u_{k + 1}, q_{k + 1})}\\
	= & \bigcup_{\alpha \in [l]^{k + 1}}
	    \bigl(R_l^\alpha(u_1, \ldots, u_{k + 1}) \cap
		  S_l^\alpha(u_1, \ldots, u_{k + 1})\bigr)
	\displaypunct{.}
    \end{align*}

    This concludes the proof of the lemma.
\end{proof}

Our goal now is to prove, given an alphabet $\Sigma$, given $l \in \N_{>0}$ and
$u_1, \ldots, u_k \in \Sigma^+$ ($k \in \N_{>0}$) such that for each
$i \in [k]$, the letters in $u_i$ are all distinct, that for any
$\alpha \in [l]^k$, the language
$R_l^\alpha(u_1, \ldots, u_k) \cap S_l^\alpha(u_1, \ldots, u_k)$ is in
$\Prog{\FMVJ}$; closure of $\Prog{\FMVJ}$ under union
(Proposition~\ref{lemma-simple-closure-P}) consequently entails that
$\ccddo{u_1, \ldots, u_k}_l \in \Prog{\FMVJ}$.
The way $R_l^\alpha(u_1, \ldots, u_k)$ and $S_l^\alpha(u_1, \ldots, u_k)$ are
defined allows us to reason as follows. For each $i \in [k]$ verifying
$\alpha_i < l$, let $L_i$ be the language of words $w$ over $\Sigma$ containing
$x_{i, 1} {u_i}^{\alpha_i} x_{i, 2}$ as a subword but not
$x_{i, 1} {u_i}^{\alpha_i + 1} x_{i, 2}$ and such that $w = y_1 u_i y_2$ with
$y_1 \in x_{i, 1} \shuffle \Sigma^*$ and $y_2 \in x_{i, 2} \shuffle \Sigma^*$,
where $x_{i, 1} = {u_1}^{\alpha_1} \cdots {u_{i - 1}}^{\alpha_{i - 1}}$ and
$x_{i, 2} = {u_{i + 1}}^{\alpha_{i + 1}} \cdots {u_k}^{\alpha_k}$.
If we manage to prove that for each $i \in [k]$ verifying $\alpha_i < l$ we have
$L_i \in \Prog{\FMVJ}$, we can conclude that
$R_l^\alpha(u_1, \ldots, u_k) \cap S_l^\alpha(u_1, \ldots, u_k) =
 ({u_1}^{\alpha_1} \cdots {u_k}^{\alpha_k}) \shuffle \Sigma^* \cap
 \bigcap_{i \in [k], \alpha_i < l} L_i$
does belong to $\Prog{\FMVJ}$ by closure of $\Prog{\FMVJ}$ under intersection,
Proposition~\ref{lemma-simple-closure-P}.
The lemma that follows, the main lemma in the proof of
Theorem~\ref{thm:TDDO_languages_in_P(J)}, exactly shows this. The proof
crucially uses the ``feedback sweeping'' technique, but note that
we actually don't know how to prove it when we do not enforce that for each
$i \in [k]$, the letters in $u_i$ are all distinct.

\begin{lemma}
\label{lem:Building-block_program}
    Let $\Sigma$ be an alphabet and $u \in \Sigma^+$ such that its letters are
    all distinct.
    For all $\alpha \in \N_{>0}$ and $x_1, x_2 \in \Sigma^*$, we have
    \[
	(x_1 u^\alpha x_2) \shuffle \Sigma^* \cap
	\bigl((x_1 u^{\alpha + 1} x_2) \shuffle \Sigma^*\bigr)^\complement \cap
	(x_1 \shuffle \Sigma^*) u (x_2 \shuffle \Sigma^*)
	\in \Prog{\FMVJ}
	\displaypunct{.}
    \]
\end{lemma}

\begin{proof}
    Before proving this lemma, we need a useful decomposition sublemma, that is
    straightforward to prove.

    \begin{lemma}
    \label{lem:Decomposition}
	Let $\Sigma$ be an alphabet and $u \in \Sigma^+$.
	Then, for all $\alpha \in \N_{>0}$, each
	$w \in u^\alpha \shuffle \Sigma^* \cap
	       (u^{\alpha + 1} \shuffle \Sigma^*)^\complement$
	verifies
	\[
	    w = \bigl(\prod_{i = 1}^\alpha
		      \prod_{j = 1}^{\length{u}} (v_{i, j} u_j)\bigr) y
	\]
	where $v_{i, j} \in (\Sigma \setminus \set{u_j})^*$ for all
	$i \in [\alpha]$ and $j \in [\length{u}]$, and
	$y \in \bigcup_{i = 1}^{\length{u}}
	       \Bigl(\prod_{j = 1}^{i - 1}
		     \bigl((\Sigma \setminus \set{u_j})^* u_j\bigr)
		     (\Sigma \setminus \set{u_i})^*\Bigr)$.
    \end{lemma}

    \begin{proof}[Proof of sublemma]
	Let $\Sigma$ be an alphabet and $u \in \Sigma^+$.

	Take $\alpha \in \N_{>0}$ and
	$w \in u^\alpha \shuffle \Sigma^* \cap
	       (u^{\alpha + 1} \shuffle \Sigma^*)^\complement$.

	As $w \in u^\alpha \shuffle \Sigma^*$, the word $w$ can be decomposed as
	$w = x y$ where $x \in u^\alpha \shuffle \Sigma^*$ and $\length{x}$ is
	minimal.
	Then, it is clearly necessarily the case that
	$x = \prod_{i = 1}^\alpha \prod_{j = 1}^{\length{u}} (v_{i, j} u_j)$
	with $v_{i, j} \in (\Sigma \setminus \set{u_j})^*$ for all
	$i \in [\alpha]$ and $j \in [\length{u}]$.
	Moreover, as $x y \notin u^{\alpha + 1} \shuffle \Sigma^*$, we
	necessarily have that $y \notin u \shuffle \Sigma^*$, so that there
	exists some $i \in [\length{u}]$ verifying that $u_1 \cdots u_{i - 1}$
	is a subword of $y$ but not $u_1 \cdots u_i$. Thus, we have that
	$y \in \prod_{j = 1}^{i - 1}
	       \bigl((\Sigma \setminus \set{u_j})^* u_j\bigr)
	       (\Sigma \setminus \set{u_i})^*$.

	This concludes the proof of the sublemma.
    \end{proof}

    We can now prove Lemma~\ref{lem:Building-block_program}.

    Let $\Sigma$ be an alphabet and $u \in \Sigma^+$ such that its letters are
    all distinct.
    Let $\alpha \in \N_{>0}$ and $x_1, x_2 \in \Sigma^*$.
    We let
    \[
	L =
	(x_1 u^\alpha x_2) \shuffle \Sigma^* \cap
	\bigl((x_1 u^{\alpha + 1} x_2) \shuffle \Sigma^*\bigr)^\complement \cap
	(x_1 \shuffle \Sigma^*) u (x_2 \shuffle \Sigma^*)
	\displaypunct{.}
    \]
    If $\length{u} = 1$, the lemma follows trivially because $L$ is piecewise
    testable and hence belongs to $\DLang{\FMVJ}$, so we assume
    $\length{u} > 1$.

    For each letter $a \in \Sigma$, we shall use $2 \length{u} - 1$ distinct
    decorated letters of the form $a^{(i)}$ for some
    $i \in \intinterval{0}{2 \length{u} - 2}$, using the convention that
    $a^{(0)} = a$; of course, for two distinct letters $a, b \in \Sigma$, we
    have that $a^{(i)}$ and $b^{(j)}$ are distinct for all
    $i, j \in \intinterval{0}{2 \length{u} - 2}$. We denote by $A$ the alphabet
    of these decorated letters.
    The main idea of the proof is, for a given input length $n \in \N$, to build
    an $A$-program $\Psi_n$ over $\Sigma^n$ such that, given an input word
    $w \in \Sigma^n$, it first ouputs the $\length{u} - 1$ first letters of $w$
    and then, for each $i$ going from $\length{u}$ to $n$, outputs $w_i$,
    followed by
    $w_{i - 1}^{(1)} \cdots w_{i - \length{u} + 1}^{(\length{u} - 1)}$ (a
    ``sweep'' of $\length{u} - 1$ letters backwards down to position
    $i - \length{u} + 1$, decorating the letters incrementally) and finally by
    $w_{i - \length{u} + 2}^{(\length{u})} \cdots w_i^{(2 \length{u} - 2)}$ (a
    ``sweep'' forwards up to position $i$, continuing the incremental decoration
    of the letters). The idea behind this way of rearranging and decorating
    letters is that, given an input word $w \in \Sigma^n$, as long as we make
    sure that $w$ and thus $\Psi_n(w)$ do contain $x_1 u^\alpha x_2$ as a
    subword but not $x_1 u^{\alpha + 1} x_2$, then $\Psi_n(w)$ can be decomposed
    as $\Psi_n(w) = y_1 z y_2$ where $y_1 \in x_1 \shuffle \Sigma^*$,
    $y_2 \in x_2 \shuffle \Sigma^*$, and $\length{y_1}, \length{y_2}$ are
    minimal, with $z$ containing
    $u^\beta u_{\length{u} - 1}^{(1)} \cdots u_1^{(\length{u} - 1)}
     u_2^{(\length{u})} \cdots u_{\length{u}}^{(2 \length{u} - 2)}
     u^{\alpha - \beta}$
    as a subword for some $\beta \in [\alpha]$ if and only if
    $w \in (x_1 \shuffle \Sigma^*) u (x_2 \shuffle \Sigma^*)$. This means we can
    check whether $w \in L$ by testing whether $w$ belongs to some fixed
    piecewise testable language over $A$.
    Let's now write the proof formally.

    For each $i \in \intinterval{0}{2 \length{u} - 2}$, let
    \[
	\function{f^{(i)}}{\Sigma}{A}{a}{a^{(i)}}
	\displaypunct{.}
    \]
    For all $i \in \N, i \geq \length{u}$, we define
    \[
	\Phi_i = (i, f^{(0)})
		 \prod_{j = 1}^{\length{u} - 1} (i - j, f^{(j)})
		 \prod_{j = 2}^{\length{u}}
		 (i - \length{u} + j, f^{(\length{u} + j - 2)})
	\displaypunct{.}
    \]
    For all $n \in \N, n < \length{u}$, we define $\Psi_n = \emptyword$.
    For all $n \in \N, n \geq \length{u}$, we define
    \[
	\Psi_n = \prod_{i = 1}^{\length{u} - 1} (i, f^{(0)})
		 \prod_{i = \length{u}}^n \Phi_i
	\displaypunct{.}
    \]
    Finally, let $K$ be the language of words over $A$ having
    \[
	\zeta_\beta =
	x_1 u^{\beta - 1} u
	\prod_{j = 1}^{\length{u} - 1} u_{\length{u} - j}^{(j)}
	\prod_{j = 2}^{\length{u}} u_j^{(\length{u} + j - 2)}
	u^{\alpha - \beta} x_2
    \]
    for some $\beta \in [\alpha]$ as a subword but not $x_1 u^{\alpha + 1} x_2$.

    \begin{claim}
	The sequence $(\Psi_n)_{n \in \N}$ of $A$-programs is a
	program-reduction from $L$ to $K$.
    \end{claim}

    Let
    \[
	\function{s}{\N}{\N}{n}
	{\begin{cases}
	    0 & \text{if $n < \length{u}$}\\
	    \length{u} - 1 + (n - \length{u} + 1) \cdot (2 \length{u} - 1) &
	    \text{otherwise}
	    \displaypunct{.}
	 \end{cases}}
    \]
    It is direct to see that
    $s(n) = \length{\Psi_n} \leq (2 \length{u} - 1) \cdot n$ for all $n \in \N$.

    Therefore, using this claim, $(\Psi_n)_{n \in \N}$ is a program-reduction
    from $L$ to $K$ of length $s(n)$, so since $K$ is piecewise testable and
    hence is recognised (classically) by some monoid from $\FMVJ$,
    Proposition~\ref{ptn:Program-reduction_to_regular_language} tells us that
    $L \in \Prog{\FMVJ, s(n)} = \Prog{\FMVJ, n}$.

    \begin{proof}[Proof of claim]
	Let $n \in \N$.
	If $n < \length{u}$, then it is obvious that for all $w \in \Sigma^n$,
	we have $w \notin (x_1 \shuffle \Sigma^*) u (x_2 \shuffle \Sigma^*)$ so
	$w \notin L^{=n}$ and also $\Psi_n(w) = \emptyword \notin K^{=s(n)}$,
	hence $L^{=n} = \emptyset = \Psi_n^{-1}(K^{=s(n)})$.
	Otherwise, $n \geq \length{u}$.
	We are going to show that $L^{=n} = \Psi_n^{-1}(K^{=s(n)})$.

	\paragraph{Left-to-right inclusion.}
	Let $w \in L^{=n}$. We want to show that $\Psi_n(w) \in K^{=s(n)}$.

	We are first going to show that there exists some $\beta \in [\alpha]$
	such that $\zeta_\beta$ is a subword of $\Psi_n(w)$.
	The fact that $w \in L^{=n}$ means in particular that
	$w \in (x_1 \shuffle \Sigma^*) u (x_2 \shuffle \Sigma^*)$ and we can
	hence decompose $w$ as $w = y_1 z y_2$ where
	$y_1 \in (x_1 \shuffle \Sigma^*)$ and $y_2 \in (x_2 \shuffle \Sigma^*)$
	with $\length{y_1}$ and $\length{y_2}$ being minimal. It follows
	necessarily that
	$z \in u^\alpha \shuffle \Sigma^* \cap
	       (u^{\alpha + 1} \shuffle \Sigma^*)^\complement \cap
	       \Sigma^* u \Sigma^*$
	by minimality of $\length{y_1}$ and $\length{y_2}$. By
	Lemma~\ref{lem:Decomposition}, we have
	$z = \bigl(\prod_{i = 1}^\alpha
		   \prod_{j = 1}^{\length{u}} (v_{i, j} u_j)\bigr) y$
	where $v_{i, j} \in (\Sigma \setminus \set{u_j})^*$ for all
	$i \in [\alpha]$ and $j \in [\length{u}]$, and
	$y \in \bigcup_{i = 1}^{\length{u}}
	       \Bigl(\prod_{j = 1}^{i - 1}
		     \bigl((\Sigma \setminus \set{u_j})^* u_j\bigr)
		     (\Sigma \setminus \set{u_i})^*\Bigr)$.
	We know the letters in $u$ are all distinct, so this means that there is
	no $\beta \in [\alpha - 1]$ such that $u$ is a factor of $z$ partly in
	$\prod_{j = 1}^{\length{u}} (v_{\beta, j} u_j)$ and partly in
	$\prod_{j = 1}^{\length{u}} (v_{\beta + 1, j} u_j)$, and that $u$ cannot
	appear as a factor of $z$ partly in
	$\prod_{j = 1}^{\length{u}} (v_{\alpha, j} u_j)$ and partly in $y$
	either. Hence, since $z \in \Sigma^* u \Sigma^*$, by the way we
	decomposed $z$, there necessarily exists $\beta \in [\alpha]$ such that
	$\prod_{j = 1}^{\length{u}} (v_{\beta, j} u_j) \in \Sigma^* u \Sigma^*$.
	Let $\gamma, \delta \in [n]$ such that
	$w_\gamma \cdots w_\delta =
	 \prod_{j = 1}^{\length{u}} (v_{\beta, j} u_j)$,
	$w_1 \cdots w_{\gamma - 1} =
	 y_1 \bigl(\prod_{i = 1}^{\beta - 1}
		   \prod_{j = 1}^{\length{u}} (v_{i, j} u_j)\bigr)$
	and
	$w_{\delta + 1} \cdots w_n =
	 \bigl(\prod_{i = \beta + 1}^\alpha
	       \prod_{j = 1}^{\length{u}} (v_{i, j} u_j)\bigr)
	 y y_2$.
	By the way $\beta$ is defined, we have
	$w_{\delta - \length{u} + 1} \cdots w_\delta = u$, because $\delta$ is
	the first and only position in $w$ with the letter $u_{\length{u}}$
	within the interval $\intinterval{\gamma}{\delta}$ verifying that
	$w_\gamma \cdots w_{\delta - 1}$ contains
	$u_1 \cdots u_{\length{u} - 1}$ as a subword, and we observe
	additionally that $\delta \geq \gamma + \length{u} - 1 \geq \length{u}$.
	This means that
	\begin{align*}
	    & \Phi_\delta(w)\\
	    = & f^{(0)}(w_\delta)
	        f^{(1)}(w_{\delta - 1}) \cdots
	        f^{(\length{u} - 1)}(w_{\delta - \length{u} + 1})
	        f^{(\length{u})}(w_{\delta - \length{u} + 2}) \cdots
	        f^{(2 \length{u} - 2)}(w_\delta)\\
	    = & u_{\length{u}}
		\prod_{j = 1}^{\length{u} - 1} u_{\length{u} - j}^{(j)}
		\prod_{j = 2}^{\length{u}} u_j^{(\length{u} + j - 2)}
	    \displaypunct{.}
	\end{align*}
	Moreover,
	\[
	    \prod_{i = 1}^{\gamma - 1} f^{(0)}(w_i) =
	    w_1 \cdots w_{\gamma - 1} =
	    y_1 \bigl(\prod_{i = 1}^{\beta - 1}
		      \prod_{j = 1}^{\length{u}} (v_{i, j} u_j)\bigr)
	    \displaypunct{,}
	\]
	\[
	    \prod_{i = \delta - \length{u} + 1}^{\delta - 1}f^{(0)}(w_i) =
	    w_{\delta - \length{u} + 1} \cdots w_{\delta - 1} =
	    u_1 \cdots u_{\length{u} - 1}
	\]
	and
	\[
	    \prod_{i = \delta + 1}^n f^{(0)}(w_i) =
	    w_{\delta + 1} \cdots w_n =
	    \bigl(\prod_{i = \beta + 1}^\alpha
		  \prod_{j = 1}^{\length{u}} (v_{i, j} u_j)\bigr)
	    y y_2
	    \displaypunct{.}
	\]
	So as
	$\prod_{i = 1}^{\gamma - 1} (i, f^{(0)})
	 \prod_{i = \delta - \length{u} + 1}^{\delta - 1} (i, f^{(0)})
	 \Phi_\delta
	 \prod_{i = \delta + 1}^n (i, f^{(0)})$
	is a subword of $\Psi_n$, we have that
	\[
	    \zeta_\beta =
	    x_1 u^{\beta - 1} u
	    \prod_{j = 1}^{\length{u} - 1} u_{\length{u} - j}^{(j)}
	    \prod_{j = 2}^{\length{u}} u_j^{(\length{u} + j - 2)}
	    u^{\alpha - \beta} x_2
	\]
	is a subword of $\Psi_n(w)$.
	
	We secondly show that $x_1 u^{\alpha + 1} x_2$ cannot be a subword of
	$\Psi_n(w)$. But this is direct by construction of $\Psi_n$, otherwise
	we would have that $x_1 u^{\alpha + 1} x_2$ is a subword of $w$,
	contradicting the fact that $w \in L^{=n}$.
	
	Hence, $\Psi_n(w) \in K^{=s(n)}$, and since this is true for all
	$w \in L^{=n}$, we have $L^{=n} \subseteq \Psi_n^{-1}(K^{=s(n)})$.

	\paragraph{Right-to-left inclusion.}
	We are going to prove the ``contrapositive inclusion''.

	Let $w \in \Sigma^n \setminus L^{=n}$. We want to show that
	$\Psi_n(w) \notin K^{=s(n)}$.
	
	Let us start with the easy cases. If we have
	$w \notin (x_1 u^\alpha x_2) \shuffle \Sigma^*$, then it means that
	$x_1 u^\alpha x_2$ is not a subword of $w$ and hence, by construction of
	$\Psi_n$, not a subword of $\Psi_n(w)$ either, so that there does not
	exist any $\beta \in [\alpha]$ such that $\zeta_\beta$ is a subword of
	$\Psi_n(w)$.
	Similarly, if we have
	$w \in (x_1 u^{\alpha + 1} x_2) \shuffle \Sigma^*$, then it means that
	$x_1 u^{\alpha + 1} x_2$ is a subword of $w$ and hence, by construction
	of $\Psi_n$, a subword of $\Psi_n(w)$.

	We now assume that
	$w \in
	 (x_1 u^\alpha x_2) \shuffle \Sigma^* \cap
	 \bigl((x_1 u^{\alpha + 1} x_2) \shuffle \Sigma^*\bigr)^\complement$
	while $w \notin (x_1 \shuffle \Sigma^*) u (x_2 \shuffle \Sigma^*)$.
	We want to show that in this case, there does not exist any
	$\beta \in [\alpha]$ such that $\zeta_\beta$ is a subword of
	$\Psi_n(w)$. Suppose for a contradiction that such a $\beta$ exists; our
	goal is to show, through a careful observation of what this implies on
	the letters in $w$ by examining how $\Psi_n$ decorates the letters, that
	this contradictingly entails $x_1 u^{\alpha + 1} x_2$ is a subword of
	$w$.

	Since $\zeta_\beta$ is a subword of $\Psi_n(w)$, it is not too difficult
	to see there exist
	\[
	    p_1, \ldots, p_{\length{x_1} + (\beta - 1) \cdot \length{u}},
	    q_1, \ldots, \allowbreak q_{3 \length{u} - 2}, r_1, \ldots,
	    r_{(\alpha - \beta) \cdot \length{u} + \length{x_2}} \in [n]
	\]
	verifying that
	\[
	    w_{p_1} \cdots w_{p_{\length{x_1} + (\beta - 1) \cdot \length{u}}} =
	    x_1 u^{\beta - 1}
	    \displaypunct{,}
	\]
	\[
	    w_{q_1} \cdots w_{q_{3 \length{u} - 2}} =
	    u \prod_{j = 1}^{\length{u} - 1} u_{\length{u} - j}
	    \prod_{j = 2}^{\length{u}} u_j
	    \displaypunct{,}
	\]
	\[
	    w_{r_1} \cdots
	    w_{r_{(\alpha - \beta) \cdot \length{u} + \length{x_2}}} =
	    u^{\alpha - \beta} x_2
	\]
	and
	\begin{align*}
	    (p_1, f^{(0)}) \cdots
	    (p_{\length{x_1} + (\beta - 1) \cdot \length{u}}, f^{(0)})
	    (q_1, f^{(0)}) \cdots (q_{\length{u}}, f^{(0)})\\
	    (q_{\length{u} + 1}, f^{(1)}) \cdots
	    (q_{2 \length{u} - 1}, f^{(\length{u} - 1)})
	    (q_{2 \length{u}}, f^{(\length{u})}) \cdots
	    (q_{3 \length{u} - 2}, f^{(2 \length{u} - 2)})\\
	    (r_1, f^{(0)}) \cdots
	    (r_{(\alpha - \beta) \cdot \length{u} + \length{x_2}}, f^{(0)})
	\end{align*}
	is a subword of $\Psi_n$.
	By construction of $\Psi_n$, we have
	\[
	    p_1 < \cdots < p_{\length{x_1} + (\beta - 1) \cdot \length{u}} <
	    q_1 < \cdots < q_{\length{u}} <
	    r_1 < \cdots < r_{(\alpha - \beta) \cdot \length{u} + \length{x_2}}
	    \displaypunct{,}
	\]
	so this implies that $w$ can be decomposed as $w = y_1 z y_2$ where
	$y_1 \in x_1 \shuffle \Sigma^*$, where
	$z \in u^\alpha \shuffle \Sigma^*$ and $y_2 \in x_2 \shuffle \Sigma^*$,
	the positions $p_1, \ldots, p_{\length{x_1}}$ corresponding to letters
	in $y_1$, the positions
	$p_{\length{x_1} + 1}, \ldots,
	 p_{\length{x_1} + (\beta - 1) \cdot \length{u}}, \allowbreak
	 q_1, \ldots, q_{\length{u}}, \allowbreak
	 r_1, \ldots, r_{(\alpha - \beta) \cdot \length{u}}$
	corresponding to letters in $z$ and the positions
	$r_{(\alpha - \beta) \cdot \length{u} + 1}, \ldots, \allowbreak
	 r_{(\alpha - \beta) \cdot \length{u} + \length{x_2}}$
	corresponding to letters in $y_2$.

	We are now going to show that, in fact,
	$q_{\length{u}} < q_{2 \length{u} - 1} < q_{2 \length{u}} < \cdots <
	 q_{3 \length{u} - 2} < r_1$,
	which implies $z \in u^{\alpha + 1} \shuffle \Sigma^*$ and thus the
	contradiction we are aiming for.
	Since $w \notin (x_1 \shuffle \Sigma^*) u (x_2 \shuffle \Sigma^*)$, we
	have $z \notin \Sigma^* u \Sigma^*$, hence as
	$w_{q_{\length{u}}} = u_{\length{u}}$ and $\length{u} > 1$, there must
	exist $j \in [\length{u} - 1]$ such that
	$w_{q_{\length{u}} - j} \neq u_{\length{u} - j}$ and
	$w_{q_{\length{u}} - \iota} = u_{\length{u} - \iota}$ for all
	$\iota \in \intinterval{0}{j - 1}$.
	By construction of $\Psi_n$, we know that
	$q_{\length{u} + j} \geq q_{\length{u}} - j$ (because the instructions
	with $f^{(j)}$ after an instruction with $f^{(0)}$ querying position
	$p \in [n]$ all query a position at least equal to $p - j$), but since
	$u_{\length{u} - j} \neq w_{q_{\length{u}} - j}$ and
	$u_{\length{u} - j} \neq u_{\length{u} - \iota} =
	 w_{q_{\length{u}} - \iota}$
	for all $\iota \in \intinterval{0}{j - 1}$ as the letters in $u$ are all
	distinct, we get that $q_{\length{u} + j} > q_{\length{u}}$.
	By (backward) induction, we can show that for all
	$\iota \in \intinterval{j + 1}{\length{u} - 1}$, we have
	$q_{\length{u} + \iota} > q_{\length{u}}$. Indeed, given
	$\iota \in \intinterval{j + 1}{\length{u} - 1}$, we have
	$q_{\length{u} + \iota - 1} > q_{\length{u}}$, either by inductive
	hypothesis or directly in the base case $\iota = j + 1$ by what we have
	just seen. So by construction of $\Psi_n$, we know that
	$q_{\length{u} + \iota} \geq q_{\length{u}}$ (because the instructions
	with $f^{(\iota)}$ after an instruction with $f^{(\iota - 1)}$ querying
	position $p \in [n]$ all query a position at least equal to $p - 1$),
	but since
	$u_{\length{u} - \iota} \neq u_{\length{u}} = w_{q_{\length{u}}}$ as the
	letters in $u$ are all distinct, it follows that
	$q_{\length{u} + \iota} > q_{\length{u}}$.
	Therefore, we have that $q_{2 \length{u} - 1} > q_{\length{u}}$.
	Moreover, by construction of $\Psi_n$, we also have
	$q_{2 \length{u} - 1} < q_{2 \length{u}} < \cdots < q_{3 \length{u} - 2}
	 < r_1$
	(because for each $\iota \in \intinterval{0}{\length{u} - 2}$, the
	instructions with $f^{(\length{u} + \iota)}$ after an instruction with
	$f^{(\length{u} + \iota - 1)}$ querying position $p \in [n]$ all query a
	position at least equal to $p + 1$ and similarly for the instructions
	with $f^{(0)}$ after an instruction with $f^{(2 \length{u} - 2)}$).
	So, to conclude, we have
	$p_1 < \cdots < p_{\length{x_1} + (\beta - 1) \cdot \length{u}} <
	 q_1 < \cdots < q_{\length{u}} < q_{2 \length{u} - 1} <
	 q_{2 \length{u}} < \cdots < q_{3 \length{u} - 2} <
	 r_1 < \cdots < r_{(\alpha - \beta) \cdot \length{u} + \length{x_2}}$
	and
	\begin{align*}
	    & w_{p_1} \cdots w_{p_{\length{x_1} + (\beta - 1) \cdot \length{u}}}
	      w_{q_1} \cdots w_{q_{\length{u}}} w_{q_{2 \length{u} - 1}}
	      w_{q_{2 \length{u}}} \cdots w_{q_{3 \length{u} - 2}}
	      w_{r_1} \cdots
	      w_{r_{(\alpha - \beta) \cdot \length{u} + \length{x_2}}}\\
	    = & x_1 u^{\beta - 1} u u_1 u_2 \cdots u_{\length{u}}
		u^{\alpha - \beta} x_2
		= x_1 u^{\alpha + 1} x_2
	    \displaypunct{.}
	\end{align*}
	This implies that $w \in (x_1 u^{\alpha + 1} x_2) \shuffle \Sigma^*$, a
	contradiction.
	So there does not exist $\beta \in [\alpha]$ such that $\zeta_\beta$ is
	a subword of $\Psi(w)$.

	Therefore, in every case $\Psi_n(w) \notin K^{=s(n)}$, and since this is
	true for all $w \in \Sigma^n \setminus L^{=n}$, we have
	$\Sigma^n \setminus L^{=n} \subseteq
	 \Psi_n^{-1}(A^{s(n)} \setminus K^{=s(n)})$,
	which is equivalent to $L^{=n} \supseteq \Psi_n^{-1}(K^{=s(n)})$.

	\bigskip

	This concludes the proof of the claim.
    \end{proof}

    And the one of the lemma.
\end{proof}

As explained before stating the previous lemma, we can now use it to prove the
result we were aiming for.

\begin{proposition}
\label{ptn:RegLangPJ-All_distinct_simple_TCCDDO}
    Let $\Sigma$ be an alphabet, $l \in \N_{>0}$ and
    $u_1, \ldots, u_k \in \Sigma^+$ ($k \in \N_{>0}$) such that for each
    $i \in [k]$, the letters in $u_i$ are all distinct.
    For all $\alpha \in [l]^k$, we have
    $R_l^\alpha(u_1, \ldots, u_k) \cap S_l^\alpha(u_1, \ldots, u_k) \in
     \Prog{\FMVJ}$.
\end{proposition}

\begin{proof}
    Let $\Sigma$ be an alphabet, $l \in \N_{>0}$ and
    $u_1, \ldots, u_k \in \Sigma^+$ ($k \in \N_{>0}$) such that for each
    $i \in [k]$, the letters in $u_i$ are all distinct.
    Let $\alpha \in [l]^k$.

    For each $i \in [k]$ verifying $\alpha_i < l$, we define
    \begin{align*}
	L_i
	= & ({u_1}^{\alpha_1} \cdots {u_k}^{\alpha_k}) \shuffle \Sigma^* \cap
	    \bigl(({u_1}^{\alpha_1} \cdots {u_i}^{\alpha_i + 1} \cdots
		   {u_k}^{\alpha_k}) \shuffle \Sigma^*\bigr)^\complement \cap\\
	& \bigl(({u_1}^{\alpha_1} \cdots {u_{i - 1}}^{\alpha_{i - 1}}) \shuffle
		\Sigma^*\bigr)
	  u_i
	  \bigl(({u_{i + 1}}^{\alpha_{i + 1}} \cdots {u_k}^{\alpha_k}) \shuffle
		\Sigma^*\bigr)
	\displaypunct{.}
    \end{align*}
    It is immediate to show that
    \[
	R_l^\alpha(u_1, \ldots, u_k) \cap S_l^\alpha(u_1, \ldots, u_k) =
	({u_1}^{\alpha_1} \cdots {u_k}^{\alpha_k}) \shuffle \Sigma^* \cap
	\bigcap_{i \in [k], \alpha_i < l} L_i
	\displaypunct{.}
    \]

    By Lemma~\ref{lem:Building-block_program}, $L_i \in \Prog{\FMVJ}$ for each
    $i \in [k]$ verifying $\alpha_i < l$. Moreover, since
    $({u_1}^{\alpha_1} \cdots {u_k}^{\alpha_k}) \shuffle \Sigma^*$ obviously is
    a piecewise testable language, it belongs to $\Prog{\FMVJ}$.
    Thus, we can conclude that
    $R_l^\alpha(u_1, \ldots, u_k) \cap S_l^\alpha(u_1, \ldots, u_k)$ belongs to
    $\Prog{\FMVJ}$ by closure of $\Prog{\FMVJ} \cap \Reg$ under intersection,
    Proposition~\ref{lemma-simple-closure-P}.
\end{proof}

We thus derive the awaited corollary.

\begin{corollary}
\label{cor:RegLangPJ-All_distinct_TCCDDO}
    Let $\Sigma$ be an alphabet, $l \in \N_{>0}$ and
    $u_1, \ldots, u_k\!\in\!\Sigma^+$ ($k \in \N_{>0}$) such that for each
    $i \in [k]$, the letters in $u_i$ are all distinct.
    Then, $\ccddo{u_1, \ldots, u_k}_l \in \Prog{\FMVJ}$.
\end{corollary}

However, what we really want to obtain is that
$\ccddo{u_1, \ldots, u_k}_l \in \Prog{\FMVJ}$ without putting any restriction on
the $u_i$'s. But, in fact, to remove the constraint that the letters must be
all distinct in each of the $u_i$'s, we simply have to decorate each of the
input letters with its position minus $1$ modulo a big
enough $d \in \N_{>0}$. This finally leads to the following proposition.

\begin{proposition}
\label{ptn:RegLangPJ-TCCDDO}
    Let $\Sigma$ be an alphabet, $l \in \N_{>0}$ and
    $u_1, \ldots, u_k \in \Sigma^+$ ($k \in \N_{>0}$).
    Then $\ccddo{u_1, \ldots, u_k}_l \in \Prog{\FMVJ}$.
\end{proposition}

\begin{proof}
    Let $\Sigma$ be an alphabet, $l \in \N_{>0}$ and
    $u_1, \ldots, u_k \in \Sigma^+$ ($k \in \N_{>0}$).

    Let $d = \max_{i \in [k]} \length{u_i}$.
    If $d = 1$, then the result is straightforward because the language
    $\ccddo{u_1, \ldots, u_k}_l$ then belongs to $\DLang{\FMVJ}$, so now we
    assume $d \geq 2$.
    We let $\Sigma_d = \Sigma \times \Z \quotient d \Z$ and for all
    $w \in \Sigma^*$, for all $i \in \Z \quotient d \Z$, we define
    $\widetilde{w}^i = \prod_{j = 1}^{\length{w}} (w_j, (j + i - 1) \mod d)$.
    We also let $\widetilde{w} = \widetilde{w}^0$ for all $w \in \Sigma^*$.

    For all $v \in \Sigma^+, \length{v} \leq d$, we define $\mu(v, 1) = v$ and
    \[
	\mu(v, l) =
	\underbrace{v_1, \ldots, v_{\length{v}}, \ldots \ldots \ldots,
		    v_1, \ldots, v_{\length{v}}}_{\text{$l$ times}}
	\displaypunct{.}
    \]
    For all $v_1, \ldots, v_{k'} \in \Sigma^+$ ($k' \in \N_{>0}$) such that
    $\length{v_i} \leq d$ for each $i \in [k']$, we let
    \[
	\ccddo{v_1, \ldots, v_{k'}}_{l, d} =
	\bigcup_{i_1, \ldots, i_{k'} \in \Z \quotient d \Z}
	\ccddo{\widetilde{v_1}^{i_1}, \ldots,
	       \widetilde{v_{k'}}^{i_{k'}}}_l
	\displaypunct{,}
    \]
    a language over $\Sigma_d$, that does belong to $\Prog{\FMVJ}$ by
    Corollary~\ref{cor:RegLangPJ-All_distinct_TCCDDO} and closure of
    $\Prog{\FMVJ} \cap \Reg$ under finite union
    (Proposition~\ref{lemma-simple-closure-P}), because since
    $\length{v_i} \leq d$ for each $i \in [k']$, each $\widetilde{v_i}^j$ for
    $j \in \Z \quotient d \Z$ has all distinct letters.

    This implies that for all $q_1, \ldots, q_k \in \set{1, l}$, we have that
    $\ccddo{\mu(u_1, q_1), \ldots, \allowbreak \mu(u_k, q_k)}_{l, d}$ does
    belong to $\Prog{\FMVJ}$, so that
    \[
	\bigcup_{q_1, \ldots, q_k \in \set{1, l}}
	\ccddo{\mu(u_1, q_1), \ldots, \mu(u_k, q_k)}_{l, d}
    \]
    is a language over $\Sigma_d$ belonging to $\Prog{\FMVJ}$.

    Now, it is not so difficult to see that
    \begin{align*}
	\ccddo{u_1, \ldots, u_k}_l
	& = \bigcup_{q_1, \ldots, q_k \in \set{1, l}}
	    L^{(l)}_{(u_1, q_1)} \cdots L^{(l)}_{(u_k, q_k)}\\
	& = \Bigset{w \in \Sigma^* \bigmid
		    \widetilde{w} \in
		    \bigcup_{q_1, \ldots, q_k \in \set{1, l}}
		    \ccddo{\mu(u_1, q_1), \ldots,
			   \mu(u_k, q_k)}_{l, d}}
	\displaypunct{,}
    \end{align*}
    which allows us to conclude that the sequence $(\Psi_n)_{n \in \N}$ of
    $\Sigma_d$-programs such that $\Psi_n(w) = \widetilde{w}$ for all $n \in \N$
    and $w \in \Sigma^n$ is a program-reduction from
    $\ccddo{u_1, \ldots, u_k}_l$ to
    $\bigcup_{q_1, \ldots, q_k \in \set{1, l}}
     \ccddo{\mu(u_1, q_1), \ldots, \mu(u_k, q_k)}_{l, d}$
    of length $n$.
    Hence, $\ccddo{u_1, \ldots, u_k}_l$ does also belong to $\Prog{\FMVJ}$ by
    Proposition~\ref{ptn:P(V)_program-reduction_closure}.
\end{proof}

This finishes to prove Theorem~\ref{thm:TDDO_languages_in_P(J)} by closure of
$\Prog{\FMVJ}$ under Boolean combinations
(Proposition~\ref{lemma-simple-closure-P}) and by the discussion at the
beginning of Subsection~\ref{sse:Non-tameness_of_J}.

\section{Algebraic characterisation of threshold dot-depth one languages}
\label{sec:Algebraic_characterisation_TDDO_languages}
In his Ph.D.\ thesis~\cite{PhD_thesis/Grosshans}, the author conjectured that
the class of threshold dot-depth languages is exactly
$\DLang{\FMVDA} \cap \DLang{\FSVJsdD}$ and proved that all strongly unambiguous
monomials (the basic building blocks in $\DLang{\FMVDA}$) that are imposed to
belong to $\DLang{\FSVJsdD}$ at the same time are in fact threshold dot-depth
one languages. The problem with the proof of this partial result supporting that
conjecture is that it is very complex and technical, without leaving much hope
for an extension to all languages in $\DLang{\FMVDA} \cap \DLang{\FSVJsdD}$.
Here we show this conjecture to be actually true by using a result by
Costa~\cite{Costa-2000}.

We first prove the easy direction, a proof actually already to be found
in~\cite{PhD_thesis/Grosshans}. The result is quite straightforward but a bit
cumbersome to prove.

\begin{proposition}
\label{ptn:TDDO_in_DAintJsdD}
    Any threshold dot-depth one language belongs to
    $\DLang{\FMVDA} \cap \DLang{\FSVJsdD}$.
\end{proposition}

\begin{proof}
    To prove the proposition, by closure under Boolean operations of both
    $\DLang{\FMVDA}$ and $\DLang{\FSVJsdD}$, it suffices to prove that for any
    $\Sigma$ an alphabet, $l \in \N_{>0}$ and $u, u_1, \ldots, u_k \in \Sigma^+$
    ($k \in \N_{>0}$), the languages $\Sigma^* u$, $u \Sigma^*$ and
    $\ccddo{u_1, \ldots, u_k}_l$ do all belong to
    $\DLang{\FMVDA} \cap \DLang{\FSVJsdD}$. This is what we show in the
    following.

    Let $\Sigma$ be an alphabet, $l \in \N_{>0}$ and
    $u, u_1, \ldots, u_k \in \Sigma^+$ ($k \in \N_{>0}$).
    First, it is obvious that $\Sigma^* u$ and $u \Sigma^*$ do belong to
    $\DLang{\FMVDA} \cap \DLang{\FSVJsdD}$.
    We now show that
    $\ccddo{u_1, \ldots, u_k}_l \in \DLang{\FMVDA} \cap \DLang{\FSVJsdD}$.

    \paragraph{Membership in $\DLang{\FSVJsdD}$.}
    As given by Definition~\ref{def:TDDO_alternative}, we have that
    \[
	\ccddo{u_1, \ldots, u_k}_l =
	\bigcup_{q_1, \ldots, q_k \in \set{1, l}}
	L^{(l)}_{(u_1, q_1)} \cdots L^{(l)}_{(u_k, q_k)}
	\displaypunct{,}
    \]
    where for all $q_1, \ldots, q_k \in \set{1, l}$, the language
    $L^{(l)}_{(u_1, q_1)} \cdots L^{(l)}_{(u_k, q_k)}$ is easily seen to be
    dot-depth one. Hence, by closure of $\DLang{\FSVJsdD}$ under finite union,
    we have that $\ccddo{u_1, \ldots, u_k}_l \in \DLang{\FSVJsdD}$.

    \paragraph{Membership in $\DLang{\FMVDA}$.}
    Let now $L = \ccddo{u_1, \ldots, u_k}_l$, let $\sim$ be its syntactic
    congruence and let $\omega$ be the idempotent power of its syntactic monoid
    $M$. Using the equational characterisation of $\FMVDA$, we are now going to
    prove that $M \in \FMVDA$: that is, we are going to prove that
    $(m n)^\omega = (m n)^\omega m (m n)^\omega$ for all $m, n \in M$, so that
    $M$ does belong to $\FMVDA$ and thus $\ccddo{u_1, \ldots, u_k}_l$ to
    $\DLang{\FMVDA}$. To show that each pair of elements of $M$ verifies the
    previous equation, by definition of the syntactic monoid of $L$, it suffices
    to show that $(uv)^\omega \sim (uv)^\omega u (uv)^\omega$ for all
    $u, v \in \Sigma^*$.

    \medskip

    Let $u, v \in \Sigma^*$. Our aim is to show that
    $(uv)^\omega \sim (uv)^\omega u (uv)^\omega$. By definition of the syntactic
    monoid of $L$ and of $\omega$, it is not too difficult to see that this is
    equivalent to showing that
    $(uv)^{\omega'} \sim (uv)^{\omega'} u (uv)^{\omega'}$ where
    $\omega' \in \N_{>0}$ is the smallest multiple of $\omega$ not smaller than
    $\sum_{i = 1}^k l \cdot \length{u_i}$ (why we need $\omega'$ to be as big
    will become clear later on).

    When both $u$ and $v$ are equal to the empty word, we trivially have that
    $(uv)^{\omega'} \sim (uv)^{\omega'} u (uv)^{\omega'}$. So we now assume that
    at least one of $u$ and $v$ is not equal to the empty word.

    \smallskip

    Let $x, y \in \Sigma^*$ be such that $w = x (uv)^{\omega'} y \in L$ and
    consider the word $w' = x (uv)^{\omega'} u (uv)^{\omega'} y$. Let's now
    prove that $w'$ does also belong to $L$.
    When $x$ or $y$ belongs to $L$, then it is obvious that $w'$ does also
    belong to it. We now assume that it is not the case.
    Let $i_1 \in [k]$ be the smallest integer in $[k]$ such that $x$ does not
    belong to $\ccddo{u_1, \ldots, u_{i_1}}_l$ and $i_2 \in [k]$ the biggest
    integer in $[k]$ such that $y$ does not belong to
    $\ccddo{u_{i_2}, \ldots, u_k}_l$, that do exist by the hypothesis we just
    made. Let $\kappa_1 \in \intinterval{0}{\length{x}}$ be the smallest integer
    in $[\length{x}]$ such that
    $x_1 \cdots x_{\kappa_1} \in \ccddo{u_1, \ldots, u_{i_1 - 1}}_l$ when
    $i_1 > 1$ and $0$ otherwise; let symmetrically
    $\kappa_2 \in \intinterval{1}{\length{y} + 1}$ be the biggest integer in
    $[\length{y}]$ such that
    $y_{\kappa_2} \cdots y_{\length{y}} \in \ccddo{u_{i_2 + 1}, \ldots, u_k}_l$
    when $i_2 < k$ and $\length{y} + 1$ otherwise.
    The idea to prove $w' \in L$ is to distinguish between three cases when
    $i_1 \leq i_2$, otherwise it is direct. When both the prefix
    $x (uv)^{l \cdot \length{u_{i_1}}}$ of $w$ and $w'$ belongs to
    $\ccddo{u_1, \ldots, u_{i_1}}_l$ and the suffix
    $(uv)^{l \cdot \length{u_{i_2}}} y$ of $w$ and $w'$ belongs to
    $\ccddo{u_{i_2}, \ldots, u_k}_l$, then we can conclude by using the fact
    that all the letters of the words $u_{i_1 + 1}$ to $u_{i_2 - 1}$ are to be
    found in the remaining factor in the middle of $w$, made solely of powers of
    $uv$. Otherwise, the prefix $x (uv)^{l \cdot \length{u_{i_1}}}$ of $w$ and
    $w'$ does not belong to $\ccddo{u_1, \ldots, u_{i_1}}_l$ or the suffix
    $(uv)^{l \cdot \length{u_{i_2}}} y$ of $w$ and $w'$ does not belong to
    $\ccddo{u_{i_2}, \ldots, u_k}_l$. When the first possibility is true, we can
    show that we necessarily have that the prefix $x (uv)^{\omega'}$ of $w$ and
    $w'$ as a whole does not belong to $\ccddo{u_1, \ldots, u_{i_1}}_l$ and then
    conclude after analysing how $w$ does consequently decompose into one prefix
    in $\ccddo{u_1, \ldots, u_{i_1 - 1}}_l$, one middle factor in
    $\ccddo{u_{i_1}}_l$ and one suffix in $\ccddo{u_{i_1 + 1}, \ldots, u_k}_l$,
    using $\kappa_1$ and $\kappa_2$. We proceed by symmetry when the second
    possibility is true.
    We now move on to the details.

    If $i_1 > i_2$, then we have that $x$ belongs to
    $\ccddo{u_1, \ldots, u_{i_1 - 1}}_l$ (which is well defined as
    $i_1 > i_2 \geq 1$) and that $y$ belongs to $\ccddo{u_{i_1}, \ldots, u_k}_l$
    (which is also well defined as $k \geq i_1$), so that $w'$ obviously belongs
    to $L$.
    Otherwise, $i_1 \leq i_2$.
    We first observe that if $x (uv)^{\omega'}$ belongs to
    $\ccddo{u_1, \ldots, u_{i_1}}_l$, then $x (uv)^{l \cdot \length{u_{i_1}}}$
    does also belong to it. Indeed, assume the hypothesis of the implication is
    true; there are two possible cases.
    Either all letters of $u_{i_1}$ appear in $uv$: in that case we have
    $(uv)^{l \cdot \length{u_{i_1}}} \in {u_{i_1}}^l \shuffle \Sigma^* \subseteq
     \ccddo{u_{i_1}}_l$
    and hence
    $x (uv)^{l \cdot \length{u_{i_1}}} \in \ccddo{u_1, \ldots, u_{i_1}}_l$.
    Or there is at least one letter in $u_{i_1}$ not appearing in $uv$: since
    $x \notin \ccddo{u_1, \ldots, u_{i_1}}_l$, either
    \begin{itemize}
	\item
	    $u_{i_1}$ is a factor of $x (uv)^{\omega'}$ whose first letter is in
	    $x_{\kappa_1 + 1} \cdots x_{\length{x}}$ and whose last letter is in
	    $(uv)^{\omega'}$, so that because $\length{uv} \geq 1$, we have
	    \[
		x_{\kappa_1 + 1} \cdots x_{\length{x}}
		(uv)^{l \cdot \length{u_{i_1}}} \in
		\Sigma^* u_{i_1} \Sigma^* \subseteq \ccddo{u_{i_1}}_l
	    \]
	    and hence
	    $x (uv)^{l \cdot \length{u_{i_1}}} \in
	     \ccddo{u_1, \ldots, u_{i_1}}_l$;
	\item
	    or ${u_{i_1}}^l$ is a subword of
	    $x_{\kappa_1 + 1} \cdots x_{\length{x}} (uv)^{\omega'}$ such that
	    only its at most $\length{u_{i_1}} - 1$ last letters appear in the
	    factor $(uv)^{\omega'}$, so that we have
	    \[
		x_{\kappa_1 + 1} \cdots x_{\length{x}}
		(uv)^{l \cdot \length{u_{i_1}}} \in
		{u_{i_1}}^l \shuffle \Sigma^* \subseteq \ccddo{u_{i_1}}_l
	    \]
	    and hence
	    $x (uv)^{l \cdot \length{u_{i_1}}} \in
	     \ccddo{u_1, \ldots, u_{i_1}}_l$.
    \end{itemize}
    Symmetrically, we can prove that if $(uv)^{\omega'} y$ belongs to
    $\ccddo{u_{i_2}, \ldots, u_k}_l$, then the word
    $(uv)^{l \cdot \length{u_{i_2}}} y$ does also belong to it.
    We now distinguish between three different cases.
    \begin{itemize}
	\item
	    $x (uv)^{l \cdot \length{u_{i_1}}}$ does not belong to
	    $\ccddo{u_1, \ldots, u_{i_1}}_l$.
	    By what we have shown just above, this means that $x (uv)^{\omega'}$
	    does not belong to $\ccddo{u_1, \ldots, u_{i_1}}_l$ either.
	    Since $w \in L$, this necessarily means that $\kappa_2 > 1$ and that
	    there exists $\kappa_2' \in \intinterval{2}{\kappa_2}$ verifying
	    $x_1 \cdots x_{\kappa_1} \in \ccddo{u_1, \ldots, u_{i_1 - 1}}_l$,
	    \[
		x_{\kappa_1 + 1} \cdots x_{\length{x}} (uv)^{\omega'}
		y_1 \cdots y_{\kappa_2' - 1} \in \ccddo{u_{i_1}}_l
	    \]
	    and
	    $y_{\kappa_2'} \cdots y_{\length{y}} \in
	     \ccddo{u_{i_1 + 1}, \ldots, u_k}_l$,
	    implying $i_1 = i_2$ and that $\kappa_2'$ can be taken equal to
	    $\kappa_2$ by the fact that
	    $y \notin \ccddo{u_{i_2}, \ldots, u_k}_l$ and
	    $y_{\kappa_2} \cdots y_{\length{y}} \in
	     \ccddo{u_{i_2 + 1}, \ldots, u_k}_l$.
	    If $(uv)^{\omega'} y$ belongs to $\ccddo{u_{i_1}, \ldots, u_k}_l$,
	    then as it holds that $x \in \ccddo{u_1, \ldots, u_{i_1 - 1}}_l$, we
	    have that $w' = x (uv)^{\omega'} u (uv)^{\omega'} y \in L$.
	    Otherwise, since $u_{i_1}$ contains at least one letter not
	    appearing in $uv$ and $\length{uv} \geq 1$, ${u_{i_1}}^l$ must be a
	    subword of
	    $x_{\kappa_1 + 1} \cdots x_{\length{x}} (uv)^{\omega'}
	     y_1 \cdots y_{\kappa_2' - 1} \in \ccddo{u_{i_1}}_l$
	    with at most $\length{u_{i_1}} - 1$ of its letters appearing in the
	    factor $(uv)^{\omega'}$, so that 
	    \[
		x_{\kappa_1 + 1} \cdots x_{\length{x}}
		(uv)^{\omega'} u (uv)^{\omega'} y_1 \cdots y_{\kappa_2' - 1} \in
		{u_{i_1}}^l \shuffle \Sigma^* \subseteq \ccddo{u_{i_1}}_l
		\displaypunct{,}
	    \]
	    also showing $w' = x (uv)^{\omega'} u (uv)^{\omega'} y \in L$.
	\item
	    $(uv)^{l \cdot \length{u_{i_2}}} y$ does not belong to
	    $\ccddo{u_{i_2}, \ldots, u_k}_l$. Symmetrically to the previous
	    case, we can show that $w' \in L$.
	\item
	    $x (uv)^{l \cdot \length{u_{i_1}}}$ belongs to
	    $\ccddo{u_1, \ldots, u_{i_1}}_l$ on one side and
	    $(uv)^{l \cdot \length{u_{i_2}}} y$ belongs to
	    $\ccddo{u_{i_2}, \ldots, u_k}_l$ on the other side.
	    In this case, for all $i \in \intinterval{i_1 + 1}{i_2 - 1}$, we
	    have that $\alphabet(u_i) \subseteq \alphabet(uv)$, because as
	    $x \notin \ccddo{u_1, \ldots, u_{i_1}}_l$ and
	    $y \notin \ccddo{u_{i_2}, \ldots, u_k}_l$, we must have
	    $(uv)^{\omega'} \in \ccddo{u_{i_1 + 1}, \ldots, u_{i_2 - 1}}_l$.
	    Hence, we have that
	    $(uv)^{\omega' - l \cdot \length{u_{i_1}}} u
	     (uv)^{\omega' - l \cdot \length{u_{i_2}}}$,
	    containing
	    $(uv)^{\sum_{i = i_1 + 1}^{i_2 - 1} l \cdot \length{u_i}}$ as a
	    subword, belongs to
	    ${u_{i_1 + 1}}^l \cdots {u_{i_2 - 1}}^l \shuffle \Sigma^* \subseteq
	     \ccddo{u_{i_1 + 1}, \ldots, u_{i_2 - 1}}_l$.
	    Thus, putting all together, we get that
	    $w' = x (uv)^{\omega'} u (uv)^{\omega'} y \in L$.
    \end{itemize}
    Therefore, in any case we have $x (uv)^{\omega'} u (uv)^{\omega'} y \in L$.

    \smallskip

    Let $x, y \in \Sigma^*$ such that
    $x (uv)^{\omega'} u (uv)^{\omega'} y \in L$. In a way similar to above, we
    can show that then, $x (uv)^{\omega'} y \in L$.

    \medskip

    This shows that $(uv)^\omega \sim (uv)^\omega u (uv)^\omega$ and as it is
    true for all $u, v \in \Sigma^*$, we eventually get that
    $\ccddo{u_1, \ldots, u_k}_l \in \DLang{\FMVDA}$.

    \bigskip

    This concludes the proof of the proposition.
\end{proof}

Let us denote by $\FSVGen{\FMVDA}$ the variety of semigroups generated by the
variety of monoids $\FMVDA$, that is, the smallest variety of semigroups
containing $\FMVDA$.
By~\cite[Chapter~V, Exercise~1.3 and Proposition~1.1]{Books/Eilenberg-1976}, we
have that for any language over some alphabet, its syntactic semigroup belongs
to $\FSVGen{\FMVDA}$ if and only its syntactic monoid belongs to $\FMVDA$. Thus,
the languages in $\DLang{\FMVDA} \cap \DLang{\FSVJsdD}$ are exactly those whose
syntactic semigroup belongs to $\FSVGen{\FMVDA}$ and $\DLang{\FSVJsdD}$; in
other words, $\DLang{\FMVDA} \cap \DLang{\FSVJsdD} = \DLang{\FSVDAintJsdD}$.

In his Ph.D.\ thesis~\cite{PhD_thesis/Costa} and later in~\cite{Costa-2000},
Costa gave a language theoretic characterisation of $\DLang{\FSVDAintLJ}$, where
$\FSVLJ$ is the variety of locally-$\FMVJ$ semigroups, the class of all finite
semigroups $S$ such that for any idempotent $e$ in $S$, the monoid $e S e$
belongs to $\FMVJ$. It is well known that $\FSVJsdD$ is a strict subclass of
$\FSVLJ$~\cite[Theorem~17.3, Example~15.8]{Tilson-1987}, but if we manage to
prove that any language in $\DLang{\FSVDAintLJ}$ is threshold dot-depth one, we
would in particular have proven that any language in
$\DLang{\FMVDA} \cap \DLang{\FSVJsdD}$ is threshold dot-depth one.
Let us present the characterisation of Costa.

\begin{definition}[See~{\cite[p.35]{Costa-2000}}]
    Let $\Sigma$ be an alphabet.
    We let $\mathcal{K}(\FSVDAintLJ)(\Sigma^*)$ be the set of
    languages over $\Sigma$ of the form
    \[
	u_0 A_1^* X_1 A_2^* \cdots X_{n - 1} A_n^* u_n
    \]
    where $n, r \in \N$, $u_0, u_1, \ldots, u_n \in \Sigma^*$,
    $\emptyset \neq A_1, A_2, \ldots, A_n \subseteq \Sigma$ and, for all
    $i \in \intinterval{1}{n - 1}$, we have
    \[
	X_i =
	\begin{cases}
	    \set{u_i} &
		\begin{array}{@{}l@{}}
		    \text{if $u_i \neq \emptyword$, with}\\
		    \alphabet(u_i) \nsubseteq A_i, A_{i + 1}
		\end{array}\\
	    (A_i \setminus A_{i + 1}) (A_i \cap A_{i + 1})^{\geq r}
	    (A_{i + 1} \setminus A_i) &
		\begin{array}{@{}l@{}}
		    \text{otherwise, with}\\
		    \text{$A_i \nsubseteq A_{i + 1}$ and $A_{i + 1} \nsubseteq
		A_i$}
		\end{array}
	\end{cases}
	\displaypunct{.}
    \]
\end{definition}

\begin{remark}
    There's actually a slight mistake in the definition given
    in~\cite[p.35]{Costa-2000}, as the languages defined should be exactly those
    recognised by the automata
    $\mathcal{A}(r; u_0, A_1, u_1, \ldots, \allowbreak A_n, u_n)$ of page 10:
    the $(A_{i + 1} \setminus A_i)$ factor is missing in the definition of $X_i$
    for all $u_i$'s equal to the empty word.
\end{remark}

\begin{theorem}[{\cite[Theorem 9.1]{Costa-2000}}]
\label{thm:Costa}
    For $\Sigma$ an alphabet, $\DLang{\FSVDAintLJ}(\Sigma^*)$ is the set of all
    languages over $\Sigma$ that are Boolean combinations of languages of
    $\mathcal{K}(\FSVDAintLJ)(\Sigma^*)$.
\end{theorem}

The main contribution of this section is that we can indeed prove that all
languages in $\mathcal{K}(\FSVDAintLJ)(\Sigma^*)$ for any alphabet $\Sigma$ are
threshold dot-depth one. We leave the proof, that is technical, at the end of
the section.

\begin{proposition}
\label{ptn:DAintLJ_in_TDDO}
    For any alphabet $\Sigma$, any language in
    $\mathcal{K}(\FSVDAintLJ)(\Sigma^*)$ is threshold dot-depth one.
\end{proposition}

By closure of the class of threshold dot-depth one languages under Boolean
operations and since $\FSVDAintJsdD \subseteq \FSVDAintLJ$, this allows us to
fill the gap of the author's Ph.D.\ thesis using Costa's theorem,
Theorem~\ref{thm:Costa}, obtaining the following result, by combination with
Proposition~\ref{ptn:TDDO_in_DAintJsdD}.

\begin{theorem}
    A language is threshold dot-depth one if and only if it belongs to
    $\DLang{\FMVDA} \cap \DLang{\FSVJsdD} = \DLang{\FSVDAintJsdD}$.
\end{theorem}

Note that this gives an algebraic characterisation of threshold dot-depth one
languages. Also note that, as another corollary of
Proposition~\ref{ptn:DAintLJ_in_TDDO}, since $\DLang{\FSVDAintLJ}$ and
$\DLang{\FSVDAintJsdD}$ are so-called \nem-varieties of languages, we have that
$\FSVDAintJsdD = \FSVDAintLJ$ (see~\cite{Straubing-2002}).

We now prove Proposition~\ref{ptn:DAintLJ_in_TDDO}.

\begin{proof}[Proof of Proposition~\ref{ptn:DAintLJ_in_TDDO}]
    Let $\Sigma$ be an alphabet.

    Let 
    \[
	L = u_0 A_1^* X_1 A_2^* \cdots X_{n - 1} A_n^* u_n
    \]
    where $n, r \in \N$, $u_0, u_1, \ldots, u_n \in \Sigma^*$,
    $\emptyset \neq A_1, A_2, \ldots, A_n \subseteq \Sigma$ and, for all
    $i \in \intinterval{1}{n - 1}$, we have
    \[
	X_i =
	\begin{cases}
	    \set{u_i} &
		\begin{array}{@{}l@{}}
		    \text{if $u_i \neq \emptyword$, with}\\
		    \alphabet(u_i) \nsubseteq A_i, A_{i + 1}
		\end{array}\\
	    (A_i \setminus A_{i + 1}) (A_i \cap A_{i + 1})^{\geq r}
	    (A_{i + 1} \setminus A_i) &
		\begin{array}{@{}l@{}}
		    \text{otherwise, with}\\
		    \text{$A_i \nsubseteq A_{i + 1}$ and $A_{i + 1} \nsubseteq
		A_i$}
		\end{array}
	\end{cases}
	\displaypunct{.}
    \]

    If $n = 0$, then
    \[
	L = u_0 =
	u_0 \Sigma^* \cap \bigcap_{c \in \Sigma} {\ccddo{u_0, c}_2}^\complement
	\displaypunct{,}
    \]
    so $L$ is indeed threshold dot-depth one.
    If $n = 1$, then
    \[
	L = u_0 A_1^* u_n =
	u_0 \Sigma^* \cap \Sigma^* u_n \cap
	\bigcap_{c \in \Sigma \setminus A_1} {\ccddo{u_0, c, u_n}_2}^\complement
	\displaypunct{,}
    \]
    so, again, $L$ is threshold dot-depth one.

    We now assume $n \geq 2$. Given some $u \in \Sigma^+$, we set
    $\overline{u} = u_1, \ldots, u_{\length{u}}$.
    Moreover, for all $i \in \intinterval{1}{n - 1}$, we set
    \[
	Y_i =
	\begin{cases}
	    \set{u_i} & \text{if $u_i \neq \emptyword$}\\
	    (A_i \setminus A_{i + 1}) (A_{i + 1} \setminus A_i) &
	    \text{otherwise}
	\end{cases}
    \]
    and
    \[
	\widetilde{v_i} =
	\begin{cases}
	    u_i & \text{if $u_i \neq \emptyword$}\\
	    \overline{v_{i, 1} v_{i, 2}} & \text{otherwise}
	\end{cases}
    \]
    for any $v_i \in Y_i$.

    We now define
    \[
	K =
	\begin{array}[t]{@{}l@{}}
	    \displaystyle
	    u_0 \Sigma^* \cap \Sigma^* u_n \cap
	    \bigcup_{v_1 \in Y_1, \ldots, v_{n - 1} \in Y_{n - 1}}
	    \ccddo{u_0, \widetilde{v_1}, \ldots, \widetilde{v_{n - 1}}, u_n}_2
	    \cap\\
	    \displaystyle
	    \bigcap_{\substack{v_1 \in Y_1, \ldots, v_{n - 1} \in Y_{n - 1}\\
			       i \in \intinterval{1}{n},
			       c \in \Sigma \setminus A_i}}
	    {\ccddo{u_0, \overline{v_1}, \ldots, \overline{v_{i - 1}}, c,
		    \overline{v_i}, \ldots, \overline{v_{n - 1}},
		    u_n}_2}^\complement \cap\\
	    \displaystyle
	    \bigcap_{\substack{i \in \set{j \in \intinterval{1}{n - 1} \mid
					  u_j = \emptyword}\\
			       b_i \in A_i \setminus A_{i + 1},
			       b_{i + 1} \in A_{i + 1} \setminus A_i\\
			       v_1 \in Y_1, \ldots, v_{i - 1} \in Y_{i - 1},\\
			       v_{i + 1} \in Y_{i + 1}, \ldots,
			       v_{n - 1} \in Y_{n - 1}}}\\
	    \displaystyle
	    \Bigl(
	    \begin{array}[t]{@{}l@{}}
		\displaystyle
		\bigcap_{c \in \Sigma \setminus (A_i \cup A_{i + 1})}
		{\ccddo{u_0, \overline{v_1}, \ldots, \overline{v_{i - 1}}, b_i,
			c, b_{i + 1}, \overline{v_{i + 1}}, \ldots,
			\overline{v_{n - 1}}, u_n}_2}^\complement \cap\\
		\displaystyle
		\bigcap_{v \in (A_i \cap A_{i + 1})^{< r}}
		{\ccddo{u_0, \overline{v_1}, \ldots, \overline{v_{i - 1}},
			b_i v b_{i + 1}, \overline{v_{i + 1}}, \ldots,
			\overline{v_{n - 1}}, u_n}_2}^\complement\Bigr)
		\displaypunct{.}
	    \end{array}
	\end{array}
    \]
    Our goal is now to show that $L = K$, which implies that $L$ is actually
    threshold dot-depth one.

    \paragraph{Direction $L \subseteq K$.}
    Let $w \in L$. Then there exist
    $\alpha_1 \in A_1^*, \ldots, \alpha_n \in A_n^*$ and
    $x_1 \in X_1, \ldots, x_{n - 1} \in X_{n - 1}$ such that
    $w = u_0 \alpha_1 x_1 \alpha_2 \cdots x_{n - 1} \alpha_n u_n$.

    We can prove the following instrumental claim.
    \begin{claim}
	\label{clm:Alignment}
	For each $i \in \intinterval{0}{n - 1}$ and
	$v_1 \in Y_1, \ldots, v_{n - 1} \in Y_{n - 1}$ such that
	$w \in
	 \ccddo{u_0, \overline{v_1}, \ldots, \overline{v_{n - 1}}, u_n}_2$,
	we have that any prefix of $w$ belonging to
	$\ccddo{u_0, \overline{v_1}, \ldots, \overline{v_i}}_2$ starts with
	$u_0 \alpha_1 x_1 \cdots \alpha_i x_i$.
	Similarly, for each $i \in \intinterval{1}{n}$ and
	$v_1 \in Y_1, \ldots, v_{n - 1} \in Y_{n - 1}$ such that
	$w \in
	 \ccddo{u_0, \overline{v_1}, \ldots, \overline{v_{n - 1}}, u_n}_2$,
	we have that any suffix of $w$ belonging to
	$\ccddo{\overline{v_i}, \ldots, \overline{v_{n - 1}}, u_n}_2$ ends with
	$x_i \alpha_{i + 1} \cdots x_{n - 1} \alpha_n u_n$.
    \end{claim}

    Now, by the structure of $w$, it is obvious that
    \[
	w \in u_0 \Sigma^* \cap \Sigma^* u_n \cap
	      \bigcup_{v_1 \in Y_1, \ldots, v_{n - 1} \in Y_{n - 1}}
	      \ccddo{u_0, \widetilde{v_1}, \ldots, \widetilde{v_{n - 1}}, u_n}_2
	\displaypunct{.}
    \]
    It remains to be proved that $w$ does not belong to any of the sets
    complemented in the formula for $K$.

    Assume there exist $v_1 \in Y_1, \ldots, v_{n - 1} \in Y_{n - 1}$, some
    $i \in \intinterval{1}{n}$ and $c \in \Sigma \setminus A_i$ such that
    $w$ belongs to
    $\ccddo{u_0, \overline{v_1}, \ldots, \overline{v_{i - 1}}, c,
	    \overline{v_i}, \ldots, \overline{v_{n - 1}}, u_n}_2$.
    Then, since by Claim~\ref{clm:Alignment} we have that any prefix of $w$
    belonging to $\ccddo{u_0, \overline{v_1}, \ldots, \overline{v_{i - 1}}}_2$
    starts with $u_0 \alpha_1 x_1 \cdots \alpha_{i - 1} x_{i - 1}$ and that any
    suffix of $w$ belonging to
    $\ccddo{\overline{v_i}, \ldots, \overline{v_{n - 1}}, u_n}_2$
    ends with $x_i \alpha_{i + 1} \cdots x_{n - 1} \alpha_n u_n$, and since
    \[
	\ccddo{u_0, \overline{v_1}, \ldots, \overline{v_{i - 1}}, c,
	       \overline{v_i}, \ldots, \overline{v_{n - 1}}, u_n}_2 =
	\ccddo{u_0, \overline{v_1}, \ldots, \overline{v_{i - 1}}}_2 \ccddo{c}_2
	\ccddo{\overline{v_i}, \ldots, \overline{v_{n - 1}}, u_n}_2
	\displaypunct{,}
    \]
    this would imply that a factor of $\alpha_i$ belongs to $\ccddo{c}_2$.
    This would in turn mean that
    $c \in \alphabet(\alpha_i) \cap (\Sigma \setminus A_i)$ while
    $\alpha_i \in A_i^*$: contradiction.
    So,
    \[
	w \in 
	\bigcap_{\substack{v_1 \in Y_1, \ldots, v_{n - 1} \in Y_{n - 1}\\
			   i \in \intinterval{1}{n},
			   c \in \Sigma \setminus A_i}}
	{\ccddo{u_0, \overline{v_1}, \ldots, \overline{v_{i - 1}}, c,
		\overline{v_i}, \ldots, \overline{v_{n - 1}},
		u_n}_2}^\complement
	\displaypunct{.}
    \]

    Similarly, let $i \in  \intinterval{1}{n - 1}$ verifying $u_i = \emptyword$,
    let $b_i \in A_i \setminus A_{i + 1}, b_{i + 1} \in A_{i + 1} \setminus A_i$
    (i.e.\ $b_i b_{i + 1} \in Y_i$) and
    $v_1 \in Y_1, \ldots, v_{i - 1} \in Y_{i - 1},
     v_{i + 1} \in Y_{i + 1}, \ldots, v_{n - 1} \in Y_{n - 1}$.
    Firstly, assume there exists $c \in \Sigma \setminus (A_i \cup A_{i + 1})$
    such that $w$ belongs to
    $\ccddo{u_0, \overline{v_1}, \ldots, \overline{v_{i - 1}}, b_i, c,
	    b_{i + 1}, \allowbreak
	    \overline{v_{i + 1}}, \ldots, \overline{v_{n - 1}}, u_n}_2$.
    Then, since by Claim~\ref{clm:Alignment} we have that any prefix of $w$
    belonging to $\ccddo{u_0, \overline{v_1}, \ldots, \overline{v_{i - 1}}}_2$
    starts with $u_0 \alpha_1 x_1 \cdots \allowbreak \alpha_{i - 1} x_{i - 1}$
    and that any suffix of $w$ belonging to
    $\ccddo{\overline{v_{i + 1}}, \ldots, \overline{v_{n - 1}}, u_n}_2$ ends
    with $x_{i + 1} \alpha_{i + 2} \cdots x_{n - 1} \alpha_n u_n$, and since
    \begin{align*}
	& \ccddo{u_0, \overline{v_1}, \ldots, \overline{v_{i - 1}}, b_i, c,
		 b_{i + 1}, \overline{v_{i + 1}}, \ldots,
		 \overline{v_{n - 1}}, u_n}_2\\
	= & \ccddo{u_0, \overline{v_1}, \ldots, \overline{v_{i - 1}}}_2
	    \ccddo{b_i, c, b_{i + 1}}_2
	    \ccddo{\overline{v_{i + 1}}, \ldots, \overline{v_{n - 1}}, u_n}_2
	\displaypunct{,}
    \end{align*}
    this would imply that a factor of $\alpha_i x_i \alpha_{i + 1}$ belongs to
    $\ccddo{b_i, c, b_{i + 1}}_2$.
    This would in turn mean that
    $c \in \alphabet(\alpha_i x_i \alpha_{i + 1}) \cap
	   \bigl(\Sigma \setminus (A_i \cup A_{i + 1})\bigr)$
    while
    $\alpha_i x_i \alpha_{i + 1} \in
     A_i^* (A_i \setminus A_{i + 1}) (A_i \cap A_{i + 1})^{\geq r}
     (A_{i + 1} \setminus A_i) A_{i + 1}^*$:
    contradiction.
    Secondly, assume there exists $v \in (A_i \cap A_{i + 1})^{< r}$ such that
    $w$ belongs to
    $\ccddo{u_0, \overline{v_1}, \ldots, \overline{v_{i - 1}}, b_i v b_{i + 1},
	    \overline{v_{i + 1}}, \ldots, \overline{v_{n - 1}}, \allowbreak
	    u_n}_2$.
    Then, since by Claim~\ref{clm:Alignment} we have that any prefix of $w$
    belonging to $\ccddo{u_0, \overline{v_1}, \ldots, \overline{v_{i - 1}}}_2$
    starts with $u_0 \alpha_1 x_1 \cdots \alpha_{i - 1} x_{i - 1}$ and that any
    suffix of $w$ belonging to
    $\ccddo{\overline{v_{i + 1}}, \ldots, \overline{v_{n - 1}}, u_n}_2$ ends
    with $x_{i + 1} \alpha_{i + 2} \cdots x_{n - 1} \alpha_n u_n$, and since
    \begin{align*}
	& \ccddo{u_0, \overline{v_1}, \ldots, \overline{v_{i - 1}},
		 b_i v b_{i + 1}, \overline{v_{i + 1}}, \ldots,
		 \overline{v_{n - 1}}, u_n}_2\\
	= & \ccddo{u_0, \overline{v_1}, \ldots, \overline{v_{i - 1}}}_2
	    \ccddo{b_i v b_{i + 1}}_2
	    \ccddo{\overline{v_{i + 1}}, \ldots, \overline{v_{n - 1}}, u_n}_2
	\displaypunct{,}
    \end{align*}
    this would imply that a factor of $\alpha_i x_i \alpha_{i + 1}$ belongs to
    $\ccddo{b_i v b_{i + 1}}_2$. As
    $\alpha_i x_i \alpha_{i + 1} \in
     A_i^* (A_i \setminus A_{i + 1}) (A_i \cap A_{i + 1})^{\geq r}
     (A_{i + 1} \setminus A_i) A_{i + 1}^* \subseteq A_i^* A_{i + 1}^*$,
    we could not have that $(b_i v b_{i + 1})^2$ is a subword of
    $\alpha_i x_i \alpha_{i + 1}$, as $b_{i + 1} \notin A_i$ and
    $b_i \notin A_{i + 1}$, so it would necessarily be the case that
    $b_i v b_{i + 1}$ is a factor of $\alpha_i x_i \alpha_{i + 1}$. However, as
    $v$ is of length less than $r$, the word $b_i v b_{i + 1}$ does not fit as a
    factor anywhere in $\alpha_i x_i \alpha_{i + 1}$, hence we again reach a
    contradiction.
    Therefore,
    \begin{align*}
	w \in &
	    \bigcap_{\substack{i \in \set{j \in \intinterval{1}{n - 1} \mid
					  u_j = \emptyword}\\
			       b_i \in A_i \setminus A_{i + 1},
			       b_{i + 1} \in A_{i + 1} \setminus A_i\\
			       v_1 \in Y_1, \ldots, v_{i - 1} \in Y_{i - 1},\\
			       v_{i + 1} \in Y_{i + 1}, \ldots,
			       v_{n - 1} \in Y_{n - 1}}}\\
	& \Bigl(
	    \begin{array}[t]{@{}l@{}}
		\displaystyle
		\bigcap_{c \in \Sigma \setminus (A_i \cup A_{i + 1})}
		{\ccddo{u_0, \overline{v_1}, \ldots, \overline{v_{i - 1}}, b_i,
			c, b_{i + 1}, \overline{v_{i + 1}}, \ldots,
			\overline{v_{n - 1}}, u_n}_2}^\complement \cap\\
		\displaystyle
		\bigcap_{v \in (A_i \cap A_{i + 1})^{< r}}
		{\ccddo{u_0, \overline{v_1}, \ldots, \overline{v_{i - 1}},
			b_i v b_{i + 1}, \overline{v_{i + 1}}, \ldots,
			\overline{v_{n - 1}}, u_n}_2}^\complement\Bigr)
		\displaypunct{.}
	    \end{array}
    \end{align*}

    Proving Claim~\ref{clm:Alignment} finishes to prove that $L \subseteq K$.
    \begin{proof}[Proof of Claim~\ref{clm:Alignment}]
	We only prove the first part of the claim, the second part can be proven
	in a symmetric way.

	We want to prove that for each $i \in \intinterval{0}{n - 1}$ and
	$v_1 \in Y_1, \ldots, v_{n - 1} \in Y_{n - 1}$ such that
	$w \in
	 \ccddo{u_0, \overline{v_1}, \ldots, \overline{v_{n - 1}}, u_n}_2$,
	any prefix of $w$ belonging to
	$\ccddo{u_0, \overline{v_1}, \ldots, \overline{v_i}}_2$ starts with
	$u_0 \alpha_1 x_1 \cdots \alpha_i x_i$. We are going to prove it by
	induction on $i$.

	\subparagraph{Base case $i = 0$.}
	It is direct to see that for all
	$v_1 \in Y_1, \ldots, v_{n - 1} \in Y_{n - 1}$ such that
	$w \in
	 \ccddo{u_0, \overline{v_1}, \ldots, \overline{v_{n - 1}}, u_n}_2$,
	any prefix of $w$ belonging to $\ccddo{u_0}_2$ starts with $u_0$.

	\subparagraph{Induction.}
	Let $i \in \intinterval{0}{n - 2}$ verifying that for all
	$v_1 \in Y_1, \ldots, v_{n - 1} \in Y_{n - 1}$ such that
	$w \in
	 \ccddo{u_0, \overline{v_1}, \ldots, \overline{v_{n - 1}}, u_n}_2$,
	any prefix of $w$ belonging to
	$\ccddo{u_0, \overline{v_1}, \ldots, \overline{v_i}}_2$ starts with
	$u_0 \alpha_1 x_1 \cdots \alpha_i x_i$.
	We are now going to prove that this also holds for $i + 1$.

	Let $v_1 \in Y_1, \ldots, v_{n - 1} \in Y_{n - 1}$ such that
	$w \in
	 \ccddo{u_0, \overline{v_1}, \ldots, \overline{v_{n - 1}}, u_n}_2$.
	By the inductive hypothesis, we have that any prefix of $w$ belonging to
	$\ccddo{u_0, \overline{v_1}, \ldots, \overline{v_i}}_2$ starts with
	$u_0 \alpha_1 x_1 \cdots \alpha_i x_i$. We claim that any prefix of
	$\alpha_{i + 1} x_{i + 1} \alpha_{i + 2} \cdots
	 x_{n - 1} \allowbreak \alpha_n u_n$
	containing $v_{i + 1}$ as a subword starts with
	$\alpha_{i + 1} x_{i + 1}$: this implies that any prefix of $w$
	belonging to
	$\ccddo{u_0, \overline{v_1}, \ldots, \overline{v_{i + 1}}}_2$ starts
	with $u_0 \alpha_1 x_1 \cdots \alpha_i x_i \alpha_{i + 1} x_{i + 1}$,
	otherwise we would have that some proper prefix of
	$u_0 \alpha_1 x_1 \cdots \alpha_i x_i$ belongs to
	$\ccddo{u_0, \overline{v_1}, \ldots, \overline{v_i}}_2$ or that some
	proper prefix of $\alpha_{i + 1} x_{i + 1}$ belongs to
	$\ccddo{\overline{v_{i + 1}}}_2$.
	We will show that the second situation cannot occur.

	We have two cases.
	\begin{itemize}
	    \item
		Either $u_{i + 1} \neq \emptyword$. Then, we have
		$x_{i + 1} = v_{i + 1} = u_{i + 1}$ and it is obvious that
		$\alpha_{i + 1} x_{i + 1}$ contains $u_{i + 1}$ as a subword.
		But as $x_{i + 1} = u_{i + 1} = z c z'$ with
		$z \in A_{i + 1}^*$, with $c \in \Sigma \setminus A_{i + 1}$ and
		$z' \in \Sigma^*$, it cannot be that a proper prefix of
		$\alpha_{i + 1} x_{i + 1}$ contains $v_{i + 1} = u_{i + 1}$ as a
		subword, otherwise we would have that
		$\alpha_{i + 1} z \in A_{i + 1}^*$ contains $c$.
	    \item
		Or $u_{i + 1} = \emptyword$. Then, we have $x_{i + 1} = a z b$
		and $v_{i + 1} = a' b'$ with
		$a, a' \in A_{i + 1} \setminus A_{i + 2}$, with
		$z \in (A_{i + 1} \cap A_{i + 2})^{\geq r}$ and
		$b, b' \in A_{i + 2} \setminus A_{i + 1}$.
		But since $\alpha_{i + 1} a z \in A_{i + 1}^*$, it cannot
		contain $b'$, so no proper prefix of $\alpha_{i + 1} x_{i + 1}$
		can contain $v_{i + 1}$ as a subword.
	\end{itemize}

	This concludes the proof of Claim~\ref{clm:Alignment}.
    \end{proof}

    \paragraph{Direction $K \subseteq L$.}
    Let $w \in K$.

    Then, since
    \[
	w \in
	u_0 \Sigma^* \cap \Sigma^* u_n \cap
	\bigcup_{v_1 \in Y_1, \ldots, v_{n - 1} \in Y_{n - 1}}
	\ccddo{u_0, \widetilde{v_1}, \ldots, \widetilde{v_{n - 1}}, u_n}_2
	\displaypunct{,}
    \]
    we have that there exist $v_1 \in Y_1, \ldots, v_{n - 1} \in Y_{n - 1}$ and
    $y_1 \in \ccddo{\widetilde{v_1}}_2, \ldots,
     y_{n - 1} \in \ccddo{\widetilde{v_{n - 1}}}_2$
    such that $w = u_0 y_1 \cdots y_{n - 1} u_n$.

    Let $i \in \intinterval{1}{n - 1}$. There are two cases to consider.
    \begin{itemize}
	\item
	    $u_i \neq \emptyword$.
	    In that case, we have $\widetilde{v_i} = u_i$, so that
	    $y_i \in \ccddo{u_i}_2$. Since
	    $u_0 y_1 \cdots y_{i - 1} \in
	     \ccddo{u_0, \overline{v_1}, \ldots, \overline{v_{i - 1}}}_2$
	    and
	    $y_{i + 1} \cdots y_{n - 1} u_n \in
	    \ccddo{\overline{v_{i + 1}}, \ldots, \overline{v_{n - 1}},
		   \allowbreak u_n}_2$
	    and as there exists some $c \in \alphabet(u_i) \setminus A_i$, it
	    cannot be that ${u_i}^2$ is a subword of $y_i$, otherwise we would
	    have that
	    $w \in \ccddo{u_0, \overline{v_1}, \ldots, \overline{v_{i - 1}}, c,
			  \overline{v_i}, \ldots, \allowbreak
			  \overline{v_{n - 1}}, u_n}_2$.
	    Hence, there exist $\alpha_i, \beta_i \in \Sigma$ such that
	    $y_i = \alpha_i u_i \beta_i$, which actually verify
	    $\alpha_i \in A_i^*$ and $\beta_i \in A_{i + 1}^*$ for the same
	    reasons as just above.
	\item
	    $u_i = \emptyword$.
	    In that case, we have
	    $\widetilde{v_i} = \overline{v_{i, 1} v_{i, 2}}$, so that
	    $y_i \in \ccddo{v_{i, 1}, v_{i, 2}}_2$ with
	    $v_{i, 1} v_{i, 2} \in
	     (A_i \setminus A_{i + 1}) (A_{i + 1} \setminus A_i)$,
	    which means that there exist
	    $a_i \in A_i \setminus A_{i + 1}, b_i \in A_{i + 1} \setminus A_i$
	    verifying that $a_i b_i$ is a subword of $y_i$. We can take these
	    $a_i, b_i$ to be such that $y_i$ can be decomposed as
	    $\alpha_i a_i z_i b_i \beta_i$ where
	    $\alpha_i, \beta_i \in \Sigma^*$ and
	    $z_i \in \Bigl((A_i \cap A_{i + 1}) \cup
			   \bigl(\Sigma \setminus (A_i \cup A_{i + 1})\bigr)
		     \Bigr)^*$.
	    Since
	    $u_0 y_1 \cdots y_{i - 1} \in
	     \ccddo{u_0, \overline{v_1}, \ldots, \overline{v_{i - 1}}}_2$
	    and
	    $y_{i + 1} \cdots y_{n - 1} u_n \in
	     \ccddo{\overline{v_{i + 1}}, \ldots, \overline{v_{n - 1}}, u_n}_2$
	    and as $a_i b_i \in Y_i$, it must actually be that
	    $\alpha_i \in A_i^*$ and $\beta_i \in A_{i + 1}^*$, otherwise we
	    would have that
	    $w \in \ccddo{u_0, \overline{v_1}, \ldots, \overline{v_{i - 1}}, c,
			  \overline{a_i b_i}, \overline{v_{i + 1}}, \ldots,
			  \overline{v_{n - 1}}, u_n}_2$
	    for some $c \in \Sigma \setminus A_i$ or that
	    $w \in \ccddo{u_0, \overline{v_1}, \ldots, \overline{v_{i - 1}},
			  \overline{a_i b_i}, c, \overline{v_{i + 1}}, \ldots,
			  \overline{v_{n - 1}}, u_n}_2$
	    for some $c \in \Sigma \setminus A_{i + 1}$.
	    Moreover, it must be that $z_i \in (A_i \cap A_{i + 1})^*$,
	    otherwise we would have that
	    $w \in \ccddo{u_0, \overline{v_1}, \ldots, \overline{v_{i - 1}},
			  a_i, c, b_i, \overline{v_{i + 1}}, \ldots,
			  \overline{v_{n - 1}}, u_n}_2$
	    for some $c \in \Sigma \setminus (A_i \cup A_{i + 1})$. Finally,
	    using the last complemented set in the formula for $K$, we also
	    necessarily have that $\length{z_i} \geq r$, so that
	    $z_i \in (A_i \cap A_{i + 1})^{\geq r}$.
    \end{itemize}
    Therefore, in any case we have that $y_i = \alpha_i x_i \beta_i$ with
    $x_i \in X_i$ as well as $\alpha_i \in A_i^*$ and $\beta_i \in A_{i + 1}^*$.
    This allows us to conclude that $w \in L$.
\end{proof}

\section{Conclusion}
\label{sec:Conclusion}
Although $\Prog{\FMVJ}$ is very small compared to $\AC[0]$, we have shown that
programs over monoids in $\FMVJ$ are an interesting subject of study in that
they allow to do quite unexpected things. The ``feedback-sweeping'' technique
allows one to detect presence of a factor thanks to such programs as long as
this factor does not appear too often as a subword: this is the basic principle
behind threshold dot-depth one languages, that our article shows to belong
wholly to $\Prog{\FMVJ}$.

The result that the class of threshold dot-depth one languages corresponds
exactly to $\DLang{\FMVDA} \cap \DLang{\FSVJsdD}$ is of independent interest for
automata theory: it means that the class of threshold dot-depth one languages
corresponds exactly to the intersection of the class of dot-depth one languages
and of the class of languages recognised by monoids in $\FMVDA$, two well-known
classes at the bottom of the dot-depth hierarchy (see~\cite{Pin-2017}).
Obtaining similar results for higher levels of the dot-depth hierarchy could be
a nice research goal in automata theory.

Concerning the question whether threshold dot-depth one languages with
additional positional modular counting do correspond exactly to the languages in
$\DLang{\StVQDA} \cap \DLang{\StVQJsdD}$, we think that with the algebraic
characterisation of threshold dot-depth languages we now have, it should
readily be solved affirmatively using the finite category machinery
of~\cite{Dartois-Paperman-2014}.

\bibliographystyle{abbrv}
\bibliography{Bibliography}

\begin{thebibliography}{10}

\bibitem{Ajtai-1983}
M.~Ajtai.
\newblock {$\Sigma_1^1$}-formulae on finite structures.
\newblock {\em Annals of pure and applied logic}, 24(1):1--48, 1983.

\bibitem{Barrington-1989}
D.~A.~M. Barrington.
\newblock Bounded-width polynomial-size branching programs recognize exactly
  those languages in {NC}{${^1}$}.
\newblock {\em J. Comput. Syst. Sci.}, 38(1):150--164, 1989.

\bibitem{Barrington-Therien-1988}
D.~A.~M. Barrington and D.~Th{\'{e}}rien.
\newblock Finite monoids and the fine structure of {NC}{${^1}$}.
\newblock {\em J. {ACM}}, 35(4):941--952, 1988.

\bibitem{PhD_thesis/Costa}
J.~C. Costa.
\newblock {\em Quelques Intersections de Vari{\'e}t{\'e}s de Semigroupes Finis
  et de Vari{\'e}t{\'e}s de Langages, Op{\'e}rations Implicites}.
\newblock PhD thesis, Université Pierre-et-Marie-Curie (Paris-VI), Paris,
  France, 1998.

\bibitem{Costa-2000}
J.~C. Costa.
\newblock Free profinite semigroups over some classes of semigroups locally in
  {D(G)}.
\newblock {\em {IJAC}}, 10(4):491--537, 2000.

\bibitem{Dartois-Paperman-2014}
L.~Dartois and C.~Paperman.
\newblock Adding modular predicates to first-order fragments.
\newblock {\em CoRR}, abs/1401.6576, 2014.

\bibitem{Books/Eilenberg-1974}
S.~Eilenberg.
\newblock {\em Automata, Languages, and Machines}, volume~A.
\newblock Academic Press, New York, 1974.

\bibitem{Books/Eilenberg-1976}
S.~Eilenberg.
\newblock {\em Automata, Languages, and Machines}, volume~B.
\newblock Academic Press, New York, 1976.

\bibitem{Furst-Saxe-Sipser-1984}
M.~L. Furst, J.~B. Saxe, and M.~Sipser.
\newblock Parity, circuits, and the polynomial-time hierarchy.
\newblock {\em Mathematical Systems Theory}, 17(1):13--27, 1984.

\bibitem{PhD_thesis/Grosshans}
N.~Grosshans.
\newblock {\em The limits of Ne{\v{c}}iporuk's method and the power of programs
  over monoids taken from small varieties of finite monoids}.
\newblock PhD thesis, University of Paris-Saclay, France, 2018.

\bibitem{Grosshans-2020}
N.~Grosshans.
\newblock The power of programs over monoids in {J}.
\newblock In {\em {LATA} 2020, Milan, Italy, March 4-6, 2020, Proceedings},
  pages 315--327, 2020.

\bibitem{Grosshans-McKenzie-Segoufin-2017}
N.~Grosshans, P.~McKenzie, and L.~Segoufin.
\newblock The power of programs over monoids in {DA}.
\newblock In {\em {MFCS} 2017, August 21-25, 2017 - Aalborg, Denmark}, pages
  2:1--2:20, 2017.

\bibitem{Klima-Polak-2010}
O.~Kl{\'{\i}}ma and L.~Pol{\'{a}}k.
\newblock Hierarchies of piecewise testable languages.
\newblock {\em Int. J. Found. Comput. Sci.}, 21(4):517--533, 2010.

\bibitem{Lautemann-Tesson-Therien-2006}
C.~Lautemann, P.~Tesson, and D.~Th{\'{e}}rien.
\newblock An algebraic point of view on the {Crane Beach} property.
\newblock In {\em {CSL} 2006, Szeged, Hungary, September 25-29, 2006}, pages
  426--440, 2006.

\bibitem{Maciel-Peladeau-Therien-2000}
A.~Maciel, P.~P{\'{e}}ladeau, and D.~Th{\'{e}}rien.
\newblock Programs over semigroups of dot-depth one.
\newblock {\em Theor. Comput. Sci.}, 245(1):135--148, 2000.

\bibitem{McKenzie-Peladeau-Therien-1991}
P.~McKenzie, P.~P{\'{e}}ladeau, and D.~Th{\'{e}}rien.
\newblock {NC}{${^1}$}: The automata-theoretic viewpoint.
\newblock {\em Computational Complexity}, 1:330--359, 1991.

\bibitem{PhD_thesis/Peladeau}
P.~P{\'e}ladeau.
\newblock {\em Classes de circuits bool{\'e}ens et vari{\'e}t{\'e}s de
  mono{\"i}des}.
\newblock PhD thesis, Université Pierre-et-Marie-Curie (Paris-VI), Paris,
  France, 1990.

\bibitem{Peladeau-Straubing-Therien-1997}
P.~P{\'{e}}ladeau, H.~Straubing, and D.~Th{\'{e}}rien.
\newblock Finite semigroup varieties defined by programs.
\newblock {\em Theor. Comput. Sci.}, 180(1-2):325--339, 1997.

\bibitem{Books/Pin-1986}
J.~Pin.
\newblock {\em Varieties Of Formal Languages}.
\newblock Plenum Publishing Co., 1986.

\bibitem{Pin-2017}
J.~Pin.
\newblock The dot-depth hierarchy, 45 years later.
\newblock In {\em The Role of Theory in Computer Science - Essays Dedicated to
  Janusz Brzozowski}, pages 177--202, 2017.

\bibitem{Pin-Straubing-2005}
J.~Pin and H.~Straubing.
\newblock Some results on {$\mathcal{C}$}-varieties.
\newblock {\em {ITA}}, 39(1):239--262, 2005.

\bibitem{Simon-1975}
I.~Simon.
\newblock Piecewise testable events.
\newblock In {\em Automata Theory and Formal Languages, 2nd {GI} Conference,
  Kaiserslautern, May 20-23, 1975}, pages 214--222, 1975.

\bibitem{Straubing-1985}
H.~Straubing.
\newblock Finite semigroup varieties of the form {$V * D$}.
\newblock {\em Journal of Pure and Applied Algebra}, 36:53--94, 1985.

\bibitem{Straubing-2000}
H.~Straubing.
\newblock When can one finite monoid simulate another?
\newblock In {\em Algorithmic Problems in Groups and Semigroups}, pages
  267--288. Springer, 2000.

\bibitem{Straubing-2001}
H.~Straubing.
\newblock Languages defined with modular counting quantifiers.
\newblock {\em Inf. Comput.}, 166(2):112--132, 2001.

\bibitem{Straubing-2002}
H.~Straubing.
\newblock On logical descriptions of regular languages.
\newblock In {\em {LATIN} 2002, Cancun, Mexico, April 3-6, 2002, Proceedings},
  pages 528--538, 2002.

\bibitem{PhD_thesis/Tesson}
P.~Tesson.
\newblock {\em Computational Complexity Questions Related to Finite Monoids and
  Semigroups}.
\newblock PhD thesis, McGill University, Montreal, 2003.

\bibitem{Tesson-Therien-2001}
P.~Tesson and D.~Th{\'{e}}rien.
\newblock The computing power of programs over finite monoids.
\newblock {\em J. Autom. Lang. Comb.}, 7(2):247--258, 2001.

\bibitem{Tilson-1987}
B.~Tilson.
\newblock Categories as algebra: an essential ingredient in the theory of
  monoids.
\newblock {\em Journal of Pure and Applied Algebra}, 48(1-2):83--198, 1987.

\end{thebibliography}

\end{document}